\documentclass[12pt,a4paper]{amsart}
\usepackage{amsmath,amstext,amssymb,amsthm,amsfonts}
\newcommand{\nc}{\newcommand}
\nc{\one}{\mbox{\bf 1}}
\nc{\invtensor}{\underset{\leftarrow}{\otimes}}
\nc{\const}{\operatorname{const}}

\nc{\ad}{\operatorname{ad}}
\nc{\tr}{\operatorname{tr}}
\nc{\tp}{\operatorname{top}}
\nc{\rank}{\operatorname{rank}}
\nc{\corank}{\operatorname{corank}}
\nc{\codim}{\operatorname{codim}}
\nc{\sdim}{\operatorname{sdim}}
\nc{\mult}{\operatorname{mult}}

\nc{\spn}{\operatorname{span}}
\nc{\Sym}{\operatorname{Sym}}
\nc{\sym}{\operatorname{sym}}
\nc{\id}{\operatorname{id}}
\nc{\Id}{\operatorname{Id}}
\nc{\Ree}{\operatorname{Re}}

\nc{\htt}{\operatorname{ht}}
\nc{\Ker}{\operatorname{Ker}}
\nc{\rker}{\operatorname{rKer}}
\nc{\im}{\operatorname{Im}}
\nc{\osp}{\mathfrak{osp}}
\nc{\sgn}{\operatorname{sgn}}
\nc{\F}{\operatorname{F}}
\nc{\Mod}{\operatorname{Mod}}
\nc{\Mat}{\operatorname{Mat}}
\nc{\Soc}{\operatorname{Soc}}
\nc{\Inj}{\operatorname{Inj}}
\nc{\Hom}{\operatorname{Hom}}
\nc{\End}{\operatorname{End}}
\nc{\supp}{\operatorname{supp}}
\nc{\Card}{\operatorname{Card}}
\nc{\Ann}{\operatorname{Ann}}
\nc{\Ind}{\operatorname{Ind}}
\nc{\Coind}{\operatorname{Coind}}
\nc{\wt}{\operatorname{wt}}
\nc{\ch}{\operatorname{ch}}
\nc{\Stab}{\operatorname{Stab}}
\nc{\Sch}{{\mathcal S}\mbox{\em ch}}
\nc{\Irr}{\operatorname{Irr}}
\nc{\Spec}{\operatorname{Spec}}
\nc{\Prim}{\operatorname{Prim}}
\nc{\Aut}{\operatorname{Aut}}
\nc{\Ext}{\operatorname{Ext}}

\nc{\Fract}{\operatorname{Fract}}
\nc{\gr}{\operatorname{gr}}
\nc{\deff}{\operatorname{def}}
\nc{\HC}{\operatorname{HC}}

\nc{\red}{\operatorname{red}}
\nc{\wdchi}{\widetilde{\chi}}
\nc{\wdH}{\widetilde{H}}
\nc{\wdN}{\widetilde{N}}
\nc{\wdM}{\widetilde{M}}
\nc{\wdO}{\widetilde{O}}
\nc{\wdR}{\widetilde{R}}
\nc{\wdS}{\widetilde{S}}
\nc{\wdV}{\widetilde{V}}

\nc{\wdC}{\widetilde{C}}

\nc{\Obj}{\operatorname{Obj}}
\nc{\Dglie}{\operatorname{{\mathcal D}glie}}
\nc{\Fin}{\operatorname{{\mathcal F}in}}

\nc{\Adm}{\operatorname{\mathcal{A}dm}}
\nc{\Sg}{{\cS(\fg)}}
\nc{\Shg}{{\cS(\fhg)}}
\nc{\Ug}{{\cU(\fg)}}
\nc{\Uhg}{{\cU(\fhg)}}
\nc{\Sh}{{\cS(\fh)}}
\nc{\Uh}{{\cU(\fh)}}
\nc{\Uhh}{{\cU(\fhh)}}
\nc{\Zg}{{{\mathcal{Z}}(\fg)}}

\nc{\Vir}{{\mathcal{V}ir}}
\nc{\NS}{{\mathcal{N}S}}

\nc{\tZg}{{\widetilde{\mathcal Z}({\mathfrak g})}}
\nc{\Zk}{{\mathcal Z}({\mathfrak k})}

\nc{\Up}{{\mathcal U}({\mathfrak p})}
\nc{\Ah}{{\mathcal A}({\mathfrak h})}
\nc{\Ag}{{\mathcal A}({\mathfrak g})}
\nc{\Ap}{{\mathcal A}({\mathfrak p})}
\nc{\Zp}{{\mathcal Z}({\mathfrak p})}
\nc{\cZ}{\mathcal Z}
\nc{\cS}{\mathcal S}
\nc{\cT}{\mathcal{T}}

\nc{\cA}{\mathcal A}
\nc{\cU}{\mathcal U}
\nc{\cH}{\mathcal H}
\nc{\cM}{\mathcal M}
\nc{\cL}{\mathcal L}
\nc{\cF}{\mathcal F}
\nc{\cP}{\mathcal P}

\nc{\fg}{\mathfrak g}

\nc{\fo}{\mathfrak o}

\nc{\CO}{\mathcal O}
\nc{\CR}{\mathcal R}
\nc{\Cl}{\mathcal {C}\ell}

\nc{\cW}{\mathcal{W}}
\nc{\bM}{\mathbf{M}}
\nc{\bL}{\mathbf{L}}
\nc{\bN}{\mathbf{N}}

\nc{\zq}{\mathpzc q}

\nc{\fl}{\mathfrak l}
\nc{\fn}{\mathfrak n}
\nc{\fm}{\mathfrak m}
\nc{\fp}{\mathfrak p}
\nc{\fh}{\mathfrak h}
\nc{\ft}{\mathfrak t}
\nc{\fk}{\mathfrak k}
\nc{\fb}{\mathfrak b}
\nc{\fs}{\mathfrak s}
\nc{\fB}{\mathfrak B}

\nc{\vareps}{\varepsilon}
\nc{\varesp}{\varepsilon}
\nc{\veps}{\varepsilon}

\nc{\fsl}{\mathfrak{sl}}
\nc{\fgl}{\mathfrak{gl}}
\nc{\fso}{\mathfrak{so}}

\nc{\fpq}{\mathfrak{pq}}
\nc{\fq}{\mathfrak q}
\nc{\fsq}{\mathfrak{sq}}
\nc{\fpsq}{\mathfrak{psq}}

\nc{\fhg}{\hat{\fg}}
\nc{\fhn}{\hat{\fn}}
\nc{\fhh}{\hat{\fh}}
\nc{\fhb}{\hat{\fb}}
\nc{\hrho}{\hat{\rho}}

\nc{\hsl}{\hat{\fsl}}
\nc{\fpo}{\mathfrak{po}}
\nc{\dirlim}{\underset{\rightarrow}{\lim}\,}
\nc{\nen}{\newenvironment}
\nc{\ol}{\overline}
\nc{\ul}{\underline}
\nc{\ra}{\rightarrow}
\nc{\lra}{\longrightarrow}
\nc{\Lra}{\Longrightarrow}

\nc{\Lla}{\Longleftarrow}

\nc{\Llra}{\Longleftrightarrow}

\nc{\thla}{\twoheadleftarrow}

\nc{\hra}{\hookrightarrow}

\nc{\iso}{\overset{\sim}{\lra}}

\nc{\ssubset}{\underset{\not=}{\subset}}

\nc{\vac}{|0\rangle}

\nc{\Thm}[1]{Theorem~\ref{#1}}
\nc{\Prop}[1]{Proposition~\ref{#1}}
\nc{\Lem}[1]{Lemma~\ref{#1}}
\nc{\Cor}[1]{Corollary~\ref{#1}}
\nc{\Conj}[1]{Conjecture~\ref{#1}}
\nc{\Claim}[1]{Claim~\ref{#1}}
\nc{\Defn}[1]{Definition~\ref{#1}}
\nc{\Exa}[1]{Example~\ref{#1}}
\nc{\Rem}[1]{Remark~\ref{#1}}
\nc{\Note}[1]{Note~\ref{#1}}
\nc{\Quest}[1]{Question~\ref{#1}}
\nc{\Hyp}[1]{Hypoth\`ese~\ref{#1}}
\nen{thm}[1]{\label{#1}{\bf Theorem.\ } \em}{}
\nen{prop}[1]{\label{#1}{\bf Proposition.\ } \em}{}
\nen{lem}[1]{\label{#1}{\bf Lemma.\ } \em}{}
\nen{cor}[1]{\label{#1}{\bf Corollary.\ } \em}{}
\nen{conj}[1]{\label{#1}{\bf Conjecture.\ } \em}{}

\nen{claim}[1]{\label{#1}{\bf Claim.\ } \em}{}

\nen{defn}[1]{\label{#1}{\bf Definition.\ } }{}
\nen{exa}[1]{\label{#1}{\bf Example.\ } }{}
\nen{rem}[1]{\label{#1}{\em Remark.\ } }{}
\nen{note}[1]{\label{#1}{\em Note.\ } }{}
\nen{exer}[1]{\label{#1}{\em Exercise.\ } }{}
\nen{sket}[1]{\label{#1}{\em Sketch of proof.\ } }{}
\nen{quest}[1]{\label{#1}{\bf Question.\ } \em}{}

\nen{hyp}[1]{\label{#1}{\bf Hypoth\`ese.\ } \em}{}
\begin{document}
\setcounter{section}{-1}

\title[On complete reducibility]{On complete reducibility
for infinite-dimensional Lie algebras}
\author[Maria Gorelik]{Maria Gorelik~$^\dag$}

\address{Dept. of Mathematics, The Weizmann Institute of Science,
Rehovot 76100, Israel}
\email{maria.gorelik@weizmann.ac.il}
\thanks{$^\dag$
Supported in part by Minerva Grant and ISF Grant No.
 1142/07}

\author[Victor Kac]{Victor Kac~$^\ddag$}

\address{Dept. of Mathematics, 2-178, Massachusetts Institute of Technology,
Cambridge, MA 02139-4307, USA}
\email{kac@math.mit.edu}
\thanks{$^\ddag$ Supported in part by NSF Grant.}

\begin{abstract}
In this paper we study the complete reducibility of representations
of infinite-dimensional Lie algebras from the perspective of
representation theory of vertex algebras.
\end{abstract}

\maketitle

\section{Introduction}\label{intro}
The first part of the paper is concerned with restricted representations
of an arbitrary Kac-Moody algebra $\fg=\fg(A)$, associated to a symmetrizable
generalized Cartan matrix $A$~(\cite{Kbook}, Chapter 1).
Let $\fh$ be the Cartan subalgebra
of $\fg$ and let $\fg=\fh\oplus(\oplus_{\alpha\in\Delta}\fg_{\alpha})$
be the root space decomposition of $\fg$ with respect to $\fh$.
A $\fg$-module $M$ is called {\em restricted} if for any $v\in M$,
we have $\fg_{\alpha}v=0$ for all but finitely many positive roots $\alpha$.
This condition allows the action of the Casimir operator on $M$,
which is the basic tool of representation theory of
Kac-Moody algebras~\cite{Kbook} and also, in the case of an affine matrix
$A$, allows the basic constructions of the vertex algebra theory,
like normally ordered product of quantum fields.

In representation theory of Kac-Moody algebras one usually considers
$\fh$-diagonalizable $\fg$-modules $M=\oplus_{\mu\in\Omega(M)} M_{\mu}$,
where
$$M_{\mu}:=\{v\in M|\ hv=\mu(h)v, h\in\fh\}\not=0,$$
is the weight space, attached to a weight $\mu$, and $\Omega(M)$
is the set of weights. The $\fh$-diagonalizable module $M$
is called {\em bounded} if the set $\Omega(M)$ is bounded
by a finite set of elements $\lambda_1,\ldots,\lambda_s\in\fh^*$, i.e.
for any $\mu\in\Omega(M)$ one has $\mu\leq \lambda_i$ for some
$1\leq i\leq s$.

The category $\CR$ of restricted $\fg$-modules contains the extensively
studied category $\CO$, whose objects are $\fh$-diagonalizable
bounded $\fg$-modules. Recall that all irreducibles
in the category $\CO$ are the irreducible highest weight modules $L(\lambda)$
with highest weight $\lambda\in\fh^*$.

One of the basic facts of representation theory of Kac-Moody algebras
is that the subcategory $\CO_{int}$ of $\CO$, which consists of modules,
all of whose irreducible subquotients are integrable (i.e. are
isomorphic to $L(\lambda)$ with $\lambda$ dominant integral),
is semisimple~(\cite{Kbook}, Chapter 10).
This generalization of Weyl's complete
reducibility theorem can be generalized further as follows.

Let $(.,.)$ be a non-degenerate symmetric invariant bilinear form on $\fg$.
When restricted to $\fh$ it is still non-degenerate, and one can normalize
it, such that $(\alpha_i,\alpha_i)\in \mathbb{Q}_{>0}$ for all simple roots
$\alpha_i$~(\cite{Kbook}, Chapter 2). Then
$\Delta=\Delta^{re}\coprod\Delta^{im}$, where $\Delta^{re}\ (\text{resp.}
\Delta^{im}) =\{\alpha\in\Delta|\ (\alpha,\alpha)>0\
(\text{resp.}  (\alpha,\alpha)\leq 0)\}$; for $\alpha\in\Delta^{re}$
we let $\alpha^{\vee}\in\fh$ be such that
$\langle \lambda,\alpha^{\vee}\rangle=\frac{2(\lambda,\alpha)}
{ (\alpha,\alpha)}$. The Weyl group
$W$ of $\fg$ is the subgroup of $GL(\fh^*)$,
generated by orthogonal reflections
$r_{\alpha}$ in the roots $\alpha\in\Delta^{re}$~(\cite{Kbook}, Chapter 3).
Given $\lambda\in\fh^*$, let
$$\Delta(\lambda):=\{\alpha|\ \alpha\in\Delta^{re}\ \&\
\langle \lambda,\alpha^{\vee}\rangle\in\mathbb{Z}\},\ \
\Delta(\lambda)^{\vee}:=\{\alpha^{\vee}|\ \alpha\in\Delta(\lambda)\},\ \
\Delta_+(\lambda):= \Delta(\lambda)\cap\Delta_+^{re},$$
and let $W(\lambda)$ be the subgroup of $W$, generated by $\{r_{\alpha}|
\alpha\in\Delta(\lambda)\}$.

Set
$$\fh':=\fh\cap [\fg,\fg].$$
Notice that for any $\lambda$ one has $\Delta(\lambda)^{\vee}\subset \fh'$.
An element $\lambda\in\fh^*$, is called {\em rational} if
$$\mathbb{C}\Delta(\lambda)^{\vee}=\fh'.$$
An element $\lambda\in\fh^*$, is called {\em non-critical} if
$$2(\lambda+\rho,\alpha)\not\in\mathbb{Z}_{>0}(\alpha,\alpha)\
\text{for all }\alpha\in\Delta_+^{im}.$$

By~\cite{KK}, Thm. 2, all simple subquotients of a Verma module
$M(\mu)$ with $\mu$ non-critical are of the form $L(w(\mu+\rho)-\rho)$
for some $w\in W(\mu)$. We call $\lambda\in\fh^*$
and the corresponding $L(\lambda)$ {\em weakly admissible} if
it is non-critical and for any non-critical $\mu\not=\lambda$,
$L(\lambda)$ is not a subquotient of $M(\mu)$.
By~\cite{KK}, Thm. 2, a non-critical $\lambda$ is weakly admissible iff
$$\langle \lambda+\rho,\alpha^{\vee}\rangle\geq 0\ \text{
for all }\alpha\in\Delta_+(\lambda).$$

We denote by $\CO_{w.adm}$ the subcategory of $\CO$, which consists
of modules, all of whose simple subquotients are weakly admissible.
Clearly, any integrable module is weakly admissible, hence
the category $\CO_{int}$ is a subcategory of $\CO_{w.adm}$.
 It follows from~\cite{DGK} that
the category $\CO_{w.adm}$ is semisimple as well.

From the vertex algebra theory viewpoint it is important to study modules,
which are not necessarily $\fh$-diagonalizable,
but are $\fh$-locally finite (see~\cite{M}). In this case a $\fg$-module $M$
decomposes into a direct sum of {\em generalized weight spaces}:
$$M_{\mu}:=\{v\in M|\ (h-\mu(h))^Nv=0, h\in\fh, \ N>>0\},$$
and we still denote by $\Omega(M)$ the set of generalized weights.
We still call $M$  {\em bounded} if $\Omega(M)$ is bounded
by a finite set of elements of $\fh^*$.
We denote by $\tilde{\CO}$ the category of all $\fh$-locally finite bounded
$\fg$-modules. While in the category $\CO$,
any self-extension of $L(\lambda)$ is trivial, this is not the case for the
category $\tilde{\CO}$. In fact, if $\fh'\not=\fh$
(which is the case when $\fg$ is affine), then any $L(\lambda)$ has
obvious non-trivial self-extensions. Hence we may expect triviality
of self-extensions of $L(\lambda)$ in $\tilde{\CO}$ only if
$L(\lambda)$ is viewed as a $[\fg,\fg]$-module (recall that
$\fg=[\fg,\fg]+\fh$). Our main result in this direction is the following.

\subsection{}
\begin{thm}{thm01}
Let $\lambda\in\fh^*$ be a non-critical weight which is dominant, i.e.
$$\langle\lambda+\rho,\alpha^{\vee}\rangle>0\
\text{ for all }\alpha\in\Delta_+(\lambda).$$
Then $\Ext^1(L(\lambda),L(\lambda))$ is in canonical bijection
with the annihilator $\Delta(\lambda)^{\vee}$ in $\fh^*$ ($=\{\mu\in\fh^*|\
\langle \mu, \Delta(\lambda)^{\vee}\rangle=0\}$).

Consequently, if, in addition, $\lambda$ is rational,
then any self-extension of $L(\lambda)$, viewed as a $[\fg,\fg]$-module,
is trivial.
\end{thm}

A $\fg$-module $L(\lambda)$ and the corresponding $\lambda$ are called
{\em admissible}
if it is weakly admissible and, viewed as  a $[\fg,\fg]$-module,
has no non-trivial self-extensions. Denote by $\tilde{\CO}_{adm}$
the subcategory of $\tilde{\CO}$, consisting of $\fg$-modules,
all of whose simple subquotients are admissible. \Thm{thm01} implies that
${\CO}_{int}$ is a subcategory of $\tilde{\CO}_{adm}$.
Again, the following result is easily derived from~\cite{DGK}.

\subsection{}
\begin{thm}{thm02}
Any $\fg$-module from the category  $\tilde{\CO}_{adm}$, viewed as a
$[\fg,\fg]$-module, is completely reducible.
\end{thm}

A $\fg$-module $L(\lambda)$ and the corresponding $\lambda$
are called {\em KW-admissible} if
$\lambda$ is rational non-critical and dominant.
By~\Thm{thm01}, the KW-admissibility is a slightly
stronger condition than the admissibility, namely, an admissible
$\lambda\in\fh^*$ is KW-admissible iff $\lambda+\rho$ is regular, i.e.
$(\lambda+\rho,\alpha)\not=0$ for all $\alpha\in\Delta^{re}$.

KW-admissible modules have been introduced in~\cite{KWmod}, where
their characters were computed. The importance of this notion
comes from the fact that if $\fg$ is an affine Kac-Moody
algebra~(\cite{Kbook}, Chapter 6), then the (normalized) complete character
of a KW-admissible $\fg$-module $L(\lambda)$ is
a modular function~\cite{KWmod},\cite{KWmod2}
(and conjecturally these are all $\fg$-modules
$L(\lambda)$ with this property, which is known only
for $\fg=\hat{\fsl}_2$~\cite{KWmod}).
This fact, in turn, implies modular invariance of characters of modules
over the associated $W$-algebras~\cite{KRW},\cite{KWdet}.

Observe that  rationality, weak admissibility, KW-admissibility,
and admissibility of $\lambda\in\fh^*$
are, in fact, the properties that depend only on the restriction of $\lambda$
to $\fh'$. Indeed,take $\nu\in\fh^*$ such that $\langle\nu,\fh'\rangle=0$.
Since any root is a linear combination of simple roots,
$\lambda\in\fh^*$ is non-critical iff $\lambda+\nu$ is non-critical.
Therefore $\lambda$ is weakly admissible (resp., KW-admissible)
iff $\lambda+\nu$ is weakly admissible (resp., KW-admissible).
Since $L(\nu)$ is trivial as a $[\fg,\fg]$-module,
and $L(\lambda+\nu)\cong L(\lambda)\otimes L(\nu)$,
$\lambda$ is admissible iff $\lambda+\nu$ is  admissible.

The KW-admissibility and admissibility of a $\fg$-module $L(\lambda)$
can be understood in geometric terms as follows. Given a rational
$\Lambda\in(\fh')^*$, we associate to it the following (infinite) polyhedron:
$$\mathcal{P}(\Lambda):=\{\lambda\in (\fh')^*|\ \langle \lambda+\rho,
\alpha^{\vee}\rangle\geq 0\ \text{ for all }
\alpha\in \Delta_+(\Lambda)\}.$$
An element $\lambda\in \mathcal{P}(\Lambda)$ is called
{\em integral} if $\Delta(\lambda)=\Delta(\Lambda)$. The set of
KW-admissible weights is the union of the sets of
all interior integral points of all the polyhedra
$\mathcal{P}(\Lambda)$, whereas the set of admissible weights is obtained by
adding some integral points of the boundary of $\mathcal{P}(\Lambda)$'s, and
that of weakly admissible rational weights by adding all the integral points
of the boundary.

It would be important to replace in~\Thm{thm02} the boundedness condition by
the restrictness condition, but for general Kac-Moody algebras we can do it
only in the integrable case:

\subsection{}
\begin{thm}{thm03}
Let $M$ be a restricted $\fh$-locally finite $\fg$-module with
finite-dimensional generalized weight spaces. If any irreducible
subquotient of $M$ is isomorphic to an integrable
$\fg$-module $L(\lambda),\ \lambda\in P_+$, then, viewed as a
$[\fg,\fg]$-module, $M$ is completely reducible.
\end{thm}

Another version of this theorem is
\subsection{}
\begin{thm}{thm04}
The category of all restricted integrable $[\fg,\fg]$-modules
is semisimple (recall that a $[\fg,\fg]$-module $M$
is called integrable if all Chevalley generators $e_i$ and $f_i$ are locally
nilpotent on $M$).
\end{thm}

If $\fg$ is an  affine Kac-Moody algebra
we can treat the admissible case as well.
Recall that $\fg=[\fg,\fg]\oplus\mathbb{C}D$ for some $D\in\fh$
and $[\fg,\fg]$ contains the canonical
central element $K$, such that its eigenvalue $k$ in all
integrable $\fg$-module $L(\lambda)$ runs over all
non-negative integers~(\cite{Kbook}, Chapters 6 and 12).
It follows from~\cite{KK}, Thm. 2 that a $\fg$-module
$L(\lambda)$ is non-critical iff $k\not=-h^{\vee}$, where
$h^{\vee}$ is the dual Coxeter number, which is a positive
integer~(\cite{Kbook}, Chapter 6). By~\cite{KWmod2}, if $L(\lambda)$
is KW-admissible, then $k+h^{\vee}\in\mathbb{Q}_{>0}$.

Furthermore, if $M$ is a restricted $[\fg,\fg]$-module on which
$K=k\id_M$ and $k\not=-h^{\vee}$, then, by the Sugawara
construction~(\cite{Kbook}, Chapter 12), the action of $[\fg,\fg]$
can be extended to that of $\fg$ by letting $D=-L_0$.
We call  a simple $[\fg,\fg]$-module
$L(\lambda)$ ($\lambda\in(\fh')^*$)
{\em $k$-admissible}, if $K=k\id_M$, if $L(\lambda)$
is weakly admissible (i.e. a Verma $[\fg,\fg]$-module $M(\mu)$
with $\mu\not=\lambda$ has no subquotients
isomorphic to $L(\lambda)$), and if $L(\lambda)$ has no
non-trivial self-extensions $N$ with $K=k\id_N$.

Given $k\not=-h^{\vee}$, denote by $\Adm_k$ the category of
restricted $[\fg,\fg]$-modules $M$ with $K=k\id_M$,
which are locally finite over $\fh'+\mathbb{C}L_0$,
and such that any simple subquotient of $M$ is $k$-admissible.
It is easy to see that $\Adm_k$ is empty if
$k+h^{\vee}\in\mathbb{Q}_{\leq 0}$.

\subsection{}\label{introk}
\begin{thm}{thm05}
For $k\in \mathbb{Q}$
the category $\Adm_k$ is semisimple with finitely many simple objects.
A  weight $\lambda$, such that $\lambda+\rho$ is regular,
 is $k$-admissible iff it is KW-admissible.
\end{thm}

From the vertex algebra theory viewpoint the most important modules over
a (non-twisted) affine Lie algebra $\fg$, associated to a finite-dimensional
simple Lie algebra $\dot{\fg}$, are irreducible vacuum modules $L(k\Lambda_0)$
of level $k$, where $\Lambda_0|_{\dot{\fh}}=0$, $\dot{\fh}$ being
a Cartan subalgebra of $\dot{\fg}$. This $\fg$-module is KW-admissible
iff $k+h^{\vee}=\frac{p}{q}$,
where $p,q$ are  coprime positive integers, $p\geq h^{\vee}$
(resp. $p\geq h$, the Coxeter number), if $q$ and the ``lacety'' $l=1,2,3$
of $\dot{\fg}$ are coprime (resp. $l$ divides $q$)~\cite{KW4}.
However, the $\fg$-module $L(k\Lambda_0)$ for rational $k$
is $k$-admissible for a slightly larger set of levels, namely, we add
the ``boundary'' values $p=h^{\vee}-1$ (resp. $p=h-1$) and arbitrary
positive integer values of $q$, provided that $\gcd (q,l)=1$
(resp. $\gcd (q,l)=l$).

Similar picture persists in the description of all $k$-admissible weights
vs all KW-admissible weights of level $k$ for rational $k$. Namely,
there is a finite number of finite polyhedra in $(\fh')^*$,
such that the set of KW-admissible weights  is the set of interior
integral points of these  polyhedra , whereas the set of
$k$-admissible weights is obtained by adding some boundary integral points;
the set of admissible weights is obtained  by removing from the latter
some of the boundary  points. For example, for
$\dot{\fg}={\fsl}_2$, all KW-admissible modules of level $k=-2+\frac{p}{q}$
are of the form $L(\lambda_{r,s})$, where $\langle\lambda_{r,s}, K\rangle=k,
(\lambda_{r,s},\alpha)=(s-1)-(r-1)\frac{p}{q}$ and $\alpha$ is a simple root
of $\fsl_2$, $p$ and $q$ are positive coprime integers, $p\geq 2$,
$r$ and $s$ are integers, $1\leq r\leq q, 1\leq s\leq p-1$~\cite{KWmod2}.
On the other hand, the category $\Adm_k$ for rational $k$
is non-empty  for  $\dot{\fg}={\fsl}_2$ iff $k=-2+\frac{p}{q}$
for all positive coprime integers $p,q$
and all simple objects of this category are
$L(\lambda_{r,s})$, where $r,s\in\mathbb{Z}, 1\leq r\leq q, 0\leq s\leq p$.
The set of admissible weights in this case coincides with the set
of KW-admissible weights.

\Thm{thm05} explains why the simple affine vertex algebra $V_k$,
associated to $\fg$, with non-negative integral level $k$,
is {\em regular}, i.e. any
(restricted) $V_k$-module is completely reducible~\cite{DLM1}. Indeed, $V_k$
satisfies Zhu's $C_2$ condition (\cite{Zh},\cite{FZ},\cite{KWa}), hence any
$V_k$-module $M$ is locally finite over
$\fh'\oplus \mathbb{C}L_0$ (\Prop{hlocfin}). Moreover, any irreducible
$V_k$-module is one of the $L(\lambda)$ (cf. Subsection~\ref{92}),
and it is easy to see that $L(\lambda)$ must be integrable.
Thus, $M$ is in the category $\Adm_k$, hence is admissible.
If $k\not\in\mathbb{Z}_{\geq 0}$,
the $C_2$ condition obviously fails, hence $V_k$ is not regular.

Along the same lines the complete reducibility can be studied over
the Virasoro, Neveu-Schwarz, and other $W$-algebras.

For example, in the case of the Virasoro algebra
$\Vir=(\sum_j \mathbb{C}L_j)+\mathbb{C}C$, we denote by
$\Adm_c$ the category of all restricted $\Vir$-modules
$M$, for which $C=c\id_M$, $L_0$ is locally finite on $M$, and every
irreducible subquotient of $M$ is a $c$-admissible
highest weight $\Vir$-module $L(h,c)$, where $h$ is the lowest eigenvalue
of $L_0$. As before, $L(h,c)$ is called {\em $c$-admissible} if it is a
subquotient of a Verma module $M(h',c)$ only for $h'=h$, and it
has no non-trivial self-extensions, where $C$ acts by $c\id$.
Of course, ``restricted'' means that for any $v\in M$ one has
$L_nv=0$ for $n>>0$.

\subsection{}\label{sect06}
\begin{thm}{thm06}
Set $c(k):=1-\frac{6(k+1)^2}{k+2}$, where $k\in\mathbb{Q}\setminus\{-2\}$.
The category $\Adm_{c(k)}$ of $\Vir$-modules is semisimple with finitely
many simple objects if $k>-2$  and is empty if $k<-2$.
\end{thm}

Recall that $c(k)$ with $k+2=p/q$, where $p,q\in\mathbb{Z}_{\geq 2}$ are
coprime, is called a {\em minimal series
central charge}~\cite{BPZ} and denoted by $c^{p,q}=1-\frac{6(p-q)^2}{pq}$.
In this case all simple
objects of $\Adm_{c(k)}$ are $L(h_{r,s}^{p,q},c^{p,q})$, where
$$h_{r,s}^{p,q}=\{\frac{(rp-sq)^2-(p-q)^2}{4pq}|\
r=0,\ldots, q,\ s=0,\ldots,p, \  (r,s)\not=(0,p)\text{ or } (q,0)\}.$$
(Note that $h_{r,s}^{p,q}=h_{q-r,p-s}^{p,q}$).
The $\Vir$-modules $L(h_{r,s}^{p,q},c^{p,q})$ with $0<r<q, 0<s<p$
are called the {\em minimal series modules}~\cite{BPZ}.
Since the simple vertex algebra $V_c$, associated with the Virasoro algebra,
with the minimal series central charge $c=c^{p,q}$, satisfies the $C_2$
condition (see e.g.~\cite{GK}), it follows from~\cite{M} that $L_0$ is
locally finite on any (restricted) $V_c$-module $M$. It is also well known
that all irreducible $V_c$-modules are the minimal series modules, hence
the category $\Adm_{c^{p,q}}$ contains the category of
$V_{c^{p,q}}$-modules, as a (strictly smaller) subcategory.
Hence, by~\Thm{thm06}, $M$ is completely
reducible, proving regularity of $V_c$~\cite{DLM1}.

On the contrary,
since the vertex algebra $V_c$ with $c$ not a minimal series central charge
has infinitely many irreducible modules, the category of its modules is much
larger than the category $\Adm_{c(k)}$ with $k\in\mathbb{Z}\setminus\{-2\}$.
Indeed, letting $p=k+2$, it follows from~\Thm{thm06} that the category
$\Adm_{c^{p,1}}$ is semisimple with finitely many
simple objects if $p$ is a positive integer,
and is empty if $p$ is a negative integer. All simple
modules of this category are $L(h_{1,s}^{p,1},c^{p,1})$,
with $s=1,2,\ldots,p$, provided that $p$ is a positive integer.

Since the minimal series modules are precisely the irreducible
$\Vir$-modules, for which the (normalized) characters are modular
invariant functions~\cite{KWmod}, there is an obvious analogy
between them and the KW-admissible modules over affine Lie algebras.
In fact, conjecturally, the latter are precisely the simple highest
weight modules over the
corresponding simple vertex algebras. This is known only for
$\fg=\hat{\fsl}_2$~\cite{AM}.

\subsection{}
The contents of the paper is as follows. In Section 1 we prove a general
lemma on complete reducibility of a category of modules over an
associative algebra $\cU$, satisfying certain conditions,
see~\Lem{lemcH}. We introduce the
notions of highest weight modules, Jantzen-type filtrations and
Shapovalov-type determinants in this general setup. This is used to find
the conditions, under which any extension between two isomorphic
(resp. non-isomorphic) highest weight modules  over $\cU$ splits,
see~\Prop{corally} (resp.~\Cor{corlala}).

In Sections 2 and 3 the general approach, developed in  Section 1, is
used to establish the complete reducibility results
for symmetrizable Kac-Moody algebras. In Section 4 some stronger results
for affine Lie algebras are obtained
(see Theorems~\ref{thm01}--\ref{thm05} above). Here
we also  classify $k$-admissible irreducible vacuum modules
and give a geometric description of all  $k$-admissible  vs
KW-admissible highest weights.
In Sections 5 and 6 we classify admissible modules
for the Virasoro and the Neveu-Schwarz algebra,
using~\Prop{corally}.

After recalling some background material on vertex algebras in
Section~\ref{sectva},
we study the admissible modules over minimal $W$-algebras
in Section~\ref{Wa}. In particular, we classify the admissible
irreducible vacuum modules over those $W$-algebras (see~\Cor{corvacW})
and describe all admissible  weights vs KW-admissible highest weights
via the quantum Hamiltonian reduction (Subsection~\ref{polyhW}).

In conclusion, in Chapter 9 we state a conjecture,
which is supposed to imply that all simple $W$-algebras
$\cW_k(\fg,f)$ satisfying the $C_2$-condition, are regular.

Our ground field is $\mathbb{C}$. All tensor products are considered
over $\mathbb{C}$, unless otherwise specified.

{\em Acknowledgement.}
The first named author wishes to thank M.~Duflo and A.~Joseph
for fruitful discussions. A part of the work was done during
the first author's stay at IHES, whose hospitality is
greatly appreciated.

\section{Verma modules and Jantzen-type filtrations}\label{sectJan}
In this section we introduce, for a natural class
of associative superalgebras, the notions of Verma module,
Shapovalov form and  Jantzen filtration. For a simple
highest weight module $L$ we introduce a linear map
$\Upsilon_L:\Ext^1(L,L)\to \fh^*$
and describe its image in terms of a Jantzen-type filtration.
We prove~\Lem{lemcH}, \Prop{corally}, and~\Lem{proplala}.

\subsection{Notation}
For a commutative Lie algebra $\fh$ and
an $\fh$-module $N$, we still denote by $N_{\nu}$
the generalized weight space:
$$N_{\nu}:=\{v\in N|\ \forall h\in\fh\
\exists r \text{ s.t. } (h-\langle \nu,h\rangle)^r v=0\},$$
and by $\Omega(N)$ the set of generalized weights.
We call $N_{\nu}$ the generalized weight space of weight $\nu$ and we
say that a non-zero $v$ is a {\em generalized weight vector} if
 $v\in N_{\nu}$  for some $\nu$.
Recall that, if $N$ is a $\fg$-module with a locally  finite action of
$\fh$, then $N$ is a direct sum of its generalized weight spaces
and
\begin{equation}\label{subqu}
  \Omega(N)=\cup_{L\in \supp(N)}\Omega(L),
\end{equation}
where $\supp(N)$ stands for the set of irreducible subquotients of $N$.

\subsection{Assumptions}\label{assm}
Let  $\cU$ be  a unital associative superalgebra
and let $\fh$ be an even finite-dimensional subspace of $\cU$. We
will assume some of the following properties:

(U1) the map $\fh\hookrightarrow \cU$ induces an injective map of associative
algebras $\cS(\fh)\hookrightarrow \cU$;

(U2) $\cU$ contains two unital subalgebras $\cU_{\pm}$
such that the multiplication map gives the bijection
$\cU=\cU_-\otimes\cU(\fh)\otimes\cU_+$;

(U3) with respect to the adjoint action $(ad h)(u):=hu-uh$,
$\fh$ acts diagonally on $\cU_{\pm}$ and
the weight spaces $\cU_{\pm;\nu}$ in $\cU_{\pm}$ are finite-dimensional,
the $0$th weight space $\cU_{\pm;0}$ being $\mathbb{C}\cdot 1$;

(U4) there exists $h\in\fh$ such that $\langle \nu, h\rangle<0$
for any  $\nu\in\Omega(\cU_-)\setminus\{0\}$;

(U4')  $\ \exists h\in\fh^*\text{ s.t. }
\forall\nu\in\Omega(\cU_+)\setminus\{0\}
\text{ (resp. $\forall\nu\in\Omega(\cU_-)\setminus\{0\}$) }
\langle \nu,h\rangle>0
\text{ (resp. $\langle \nu,h\rangle<0$)}$;

(U5) $\forall\nu\ \ \dim\cU_{+;\nu}=\dim\cU_{-;\nu}$;

(U6) $\cU$ admits an antiautomorphism $\sigma$, which interchanges
$\cU_+$ with $\cU_-$ and fixes the elements of $\fh$.

The assumption~(U4') is a strong form of (U4).
Note that (U1)--(U5) imply (U4'), and that (U1)--(U4) and (U6) imply
(U5).

In this section $\cU$ satisfies the assumptions (U1)--(U4).

\subsubsection{}
For a contragredient Lie superalgebra $\fg(A,\tau)$, the Virasoro algebra or
the Neveu-Schwarz superalgebra we have the natural triangular decomposition
$\fg=\fn_-\oplus\fh\oplus\fn_+$. In these cases we let $\cU$ be the
universal enveloping algebra of $\fg$ and let $\cU_{\pm}:=\cU(\fn_{\pm})$.
Then $\cU$ satisfies the assumptions (U1)--(U6). For  minimal
$\cW$-algebras, $\cU$ is the universal enveloping algebra of the vertex
algebra $\cW$; this is a topological algebra, satisfying assumptions
(U1)--(U5), where (U2) is substituted  by its topological
version~(\ref{eqtriW}).

\subsubsection{}\label{Qpm}
Let $Q_{\pm}$ be the $\mathbb{Z}_{\geq 0}$-span of $\Omega(\cU_{\pm})$:
$$Q_{\pm}=\sum_{\alpha\in  \Omega(\cU_{\pm})}\mathbb{Z}_{\geq 0}\alpha.$$
Introduce the standard partial order on $\fh^*$: $\nu\leq \nu'$
if $\nu-\nu'\in Q_-$. If $\cU$ satisfies (U5) then $Q_+=-Q_-$.

\subsection{Verma module}
Set
$$\cU'_{\pm}:=\sum_{\nu\not=0}\cU_{\pm;\nu}.$$
By (U4), $\cU'_-$ is a two-sided ideal of $\cU_-$.

We view $\cM:=\cU/\cU\cU'_+$
as a left $\cU$-module and a right $\cS(\fh)$-module;
we call it a {\em universal Verma module}.
We consider the adjoint action of $\cS(\fh)$ on $\cM$
($(\ad h)v:=hv-vh$ for $h\in\fh, v\in \cM$) and
define the weight spaces $\cM_{\nu}$ with respect to this action.
For each $\lambda\in\fh^*$ the {\em Verma module}
$M(\lambda)$ is the evaluation of $\cM$ at $\lambda$, that is the $\cU$-module
$\cU/\cU(\cU'_+ +\sum_{h\in\fh}
(h-\langle\lambda, h\rangle))$.
Observe that $\fh$ acts diagonally on
$M(\lambda)$, a weight space $M(\lambda)_{\lambda+\nu}$
being the image of $\cM_{\nu}$, hence $\Omega(M(\lambda))=\lambda+Q_-$,
and $\dim M(\lambda)_{\lambda+\nu}=\dim\cU_{-;\nu}<\infty$.

Any quotient $N$ of the $\cU$-module $M(\lambda)$ is called
a {\em module with highest weight }
$\lambda$. In other words, $N$ is a module with highest weight
$\lambda$ if $\dim N_{\lambda}=1$, $\cU'_+ N_{\lambda}=0$ and
$N_{\lambda}$ generates $N$. Note that any  highest weight module
is indecomposable.

Each Verma module has a unique maximal proper submodule $M'(\lambda)$;
the irreducible module $L(\lambda):=M(\lambda)/M'(\lambda)$
is called the irreducible module with highest weight $\lambda$.

Our study of complete reducibility in various categories will be
based on the following lemma.

\subsubsection{}
\begin{lem}{lemcH}
Assume that $\cU$ satisfies the assumptions (U1)-(U4).
Assume, in addition, that $\cU_+$ is finitely generated
(as an associative algebra).

Let $\cA$ be a category of
$\cU$-modules with the following properties:

(i) $\cA$ is closed under taking subquotients;

(ii) $\Ext^1_{\cA} (L',L)=0$ for any irreducible modules $L,L'$ in $\cA$;

(iii) each irreducible module in $\cA$ is of the form
$L(\lambda), \lambda\in\fh^*$,

and, in addition, one of the following conditions hold:

(iv) each module  contains an irreducible submodule;

or

(iv') $\cU$ satisfies (U4'), each module in $\cA$ is $\fh$-locally finite
and $\cA$ has finitely many irreducible modules.

Then $\cA$ is a semisimple category.
\end{lem}
\begin{proof}
By (ii), $\cA$ does not contain indecomposable modules of length two
so any highest weight module in $\cA$ is irreducible.

Let $\cA$ satisfy (i)-(iv). Let $N$ be a module in $\cA$.
Assume that $N$ is not completely reducible.
Let $N'$ be the socle of $N$, that is the sum of all irreducible submodules.
By (iv), $N'\not=0$. By~\cite{Lang}, Ch. XVII,
 a sum of irreducible submodules can be decomposed in a  direct sum:
$N'=\oplus_{i\in I} L_i$, where $L_i$ are irreducible.
Let $L$ be an irreducible submodule
of $N/N'$  and $N''\subset N$ be the preimage of
$L$ under the map $N\to N/N'$. We obtain  the exact sequence
\begin{equation}\label{exseq1}
0\to \oplus_{i\in I} L_i\to N''\to^{\phi} L\to 0,
\end{equation}
which does not split, because
$N'$ is the socle of $N''$. By (iii), $L_i, L$ are
irreducible highest weight modules, so both $N',L$ are $\fh$-diagonalizable.
Thus $N''$ is $\fh$-locally finite.

Let $\nu$ be the highest weight of $L$: $L=L(\nu)$.
Let $v\in N''_{\nu}$ be a preimage of the highest weight vector in $L(\nu)$.
Since $\phi(\cU'_+v)=0$  and $\cU_+$ is finitely generated, there
exists a  finite set $J\subset I$ such that
$\cU'_+v\subset\oplus_{j\in J} L_j$. Let $\psi$ be the natural
map $\psi: N''\to N''/ \oplus_{j\in J} L_j$. Since
any highest weight module in $\cA$ is irreducible, the equality
$\cU'_+\psi(v)=0$ implies that  $\psi(v)$ generates a
module, which admits a finite filtration with the subquotients
isomorphic to $L(\nu)$. Thus, by (ii),  $\cU(\psi(v))\cong L(\nu)$.
The exact sequence
$$0\to \oplus_{j\in J} L_j\to \psi^{-1}(\cU(\psi(v)))\to^{\psi}\ L(\nu)\to
0$$
splits since $J$ is finite and $\Ext^1_{\cA} (L(\nu),L_i)=0$ by (iii).
Let $L'\subset \psi^{-1}(\cU(\psi(v)))\subset N''$
be such that $\psi^{-1}(\cU(\psi(v)))=\oplus_{j\in J} L_j\oplus L'$.
Since $v\in\oplus_{j\in J} L_j\oplus L'$, one has
$\phi(L')\not=0$.
Since $L'\cong L(\nu)$, we obtain $\phi(L')=L$ so
the exact  sequence~(\ref{exseq1}) splits, a contradiction.
Hence (i)--(iv) imply semisimplicity of $\cA$.

Let $\cA$ satisfy the assumptions (i)--(iv)'.
It is enough to verify that $\cA$ satisfies (iv).
Let $N$ be a module in $\cA$.
By (iv)', $\supp(N)$ is finite: $\supp(N)=\{L(\lambda_i)\}_{i=1}^s$.
Let $\lambda_1$ be such that, for $h\in\fh$ introduced in (U4'), the value
$\langle \lambda_1,h\rangle$ is maximal, that is
$\langle \lambda_1-\lambda_i,h\rangle\not <0$.
Then, by~(\ref{subqu}) and (U4'),
for any $\nu\in \Omega(N)$ one has $\langle \lambda_1-\nu,h\rangle\not <0$,
so $\cU'_+ N_{\lambda_1}=0$, by (U4').
Let $v\in N_{\lambda_1}$ be an $\fh$-eigenvector. Since $\cU'_+v=0$,
the submodule generated by $v$ is a highest weight module and
it belongs to $\cA$ by (i). By above,
this submodule is irreducible. This establishes (iv) and completes the proof.
\end{proof}

\subsubsection{}
\begin{rem}{}
The assumption that $\cU_+$ is finitely generated was used only for the claim
 that $\cU'_+v\subset\oplus_{j\in J} L_j$ for a finite
set $J\subset I$.  As a result, this assumption can be weakened
as follows: for any $N\in\cA$ and each $v\in N$ the left $\cU_+$-module
$\cU'_+/\Ann_{\cU'_+}v$ is finitely generated. This holds for the
continuous representations of finitely generated vertex algebras.
\end{rem}

\subsubsection{}
\begin{rem}{}
The following example shows why (i), (ii), (iv) do not imply
semisimplicity. Let $A=\mathbb{C}[x_1,x_2,\ldots]$ be a polynomial
ring in infinitely many variables, let $W$ be a one-dimensional
$A$-module $\mathbb{C}w$, such that $x_iw=0$ for all $i$,
  and let, for $i=1,2,\ldots$,  $V_i$ be the following one-dimensional
$A$-module: $x_j|_{V_i}=\delta_{ij}\Id_{V_i}$.
Then $\Ext^1(W,V_i)=0$ for all $i$, but $\Ext^1(W,\oplus_i V_i)\not=0$.
Namely, there exists a
non-splitting extension of $\oplus_i V_i$ by $W$ with the basis
$w,v_1,v_2,\ldots$, and the action of $A$ given by
$x_iw=v_i,\ x_iv_j=\delta_{ij}v_j$.
\end{rem}

\subsection{Shapovalov-type map}\label{shmap}
In this section we introduce in our setup a generalization
of Shapovalov's bilinear form and Shapovalov's determinant~\cite{Sh}.

\subsubsection{}\label{HC}
By (U2), $\cU=\cU'_-\cS(\fh) \cU'_+\oplus \cU'_-\cS(\fh)\oplus\cS(\fh)$;
let
$$\HC: \cU\to \cS(\fh)$$
be the projection with respect to this decomposition.
Identify $\cM$ with $\cU_-\otimes\cS(\fh)$ (as $\cU_-$-$\cS(\fh)$
bimodules)
and introduce a bilinear map $\cU\otimes\cM\to\cS(\fh)$  by
$$B(u,v):=\HC(uv).$$
 Notice that for $u\in\cU, v\in\cM$ one has $B(u,v)=P(uv)$,
where $P$ is the projection $\cM\to \cS(\fh)$
with the kernel $\cU'_-\otimes\cS(\fh)$. One has
\begin{equation}\label{Bush}\begin{array}{l}
B(au,v)=B(u,va)=B(u,v)a
\ \text{ for any } a\in \cS(\fh),\ u\in \cU,\ v\in \cM.\end{array}
\end{equation}

Let $v_0$ be the canonical generator of the module $\cM$
(i.e., the image of $1\in\cU$).
Denote by  $v_{\lambda}$ the image of $v_0$ in $M(\lambda)$
and by $B(\lambda)$ the evaluation of the map $B$ in $M(\lambda)$, i.e.
$B(\lambda): \cU\otimes M(\lambda)\to \mathbb{C}$
%% 26/8/09 it was $B(\lambda): M(\lambda)\otimes M(\lambda)\to \mathbb{C}$ %
((\ref{Bush}) ensures that the evaluation is well defined).

Denote by $P_{\lambda}: M(\lambda)\to M(\lambda)_{\lambda}$
the projection with the kernel
$\cU'_-v_{\lambda}=\sum_{\nu\not=\lambda} M(\lambda)_{\nu}$.
Then, for each $v\in M(\lambda)$, one has
\begin{equation}\label{Blambda}
B(\lambda)(u,v)v_{\lambda}=
P_{\lambda}(uv).
\end{equation}
By (U4), $\cU_-'$ is a two-sided ideal of $\cU_-$,  so $B(\cU_-'\cU,\cM)=0$.
Using~(\ref{Bush}) we get
\begin{equation}\label{sh+}
B(\lambda)(\cU,v)=B(\lambda)(\cU_{+,\nu},v)
\text{ for each } v\in M(\lambda)_{\lambda-\nu}.
\end{equation}

\subsubsection{}\label{MI}
We claim that
\begin{equation}\label{M'}
\{v\in M(\lambda)|\  B(\lambda)(\cU_+,v)=0\}=M'(\lambda).
\end{equation}
Indeed, by~(\ref{sh+}), $M':=\{v|\  B(\lambda)(\cU_+,v)=0\}$
is a submodule, and $M'_{\lambda}=0$, because
 $B(\lambda)(1,v_{\lambda})=1$. Thus $M'\subset M'(\lambda)$.
On the other hand, using~(\ref{Blambda}), we get
$$v\in M'(\lambda)\ \Longrightarrow\ \cU v\cap M(\lambda)_{\lambda}=0
\ \Longrightarrow\ P_{\lambda}(\cU v)=0\ \Longrightarrow\
B(\lambda)(\cU,v)=0.$$
So $M'(\lambda)\subset M'$, as required.

\subsubsection{}\label{shdet}
Assume that $\cU$ satisfies (U1)--(U5).

One has $B(\cU_{\mu},\cM_{\nu})=0$ if $\mu+\nu\not=0$.
For $\nu\in Q_+$ let
$B_{\nu}:\ \cU_{+;\nu}\otimes\cM_{-\nu}\to \cS(\fh)$ be
 the restriction of $B$. A matrix of the bilinear map $B_{\nu}$ is called
a {\em Shapovalov matrix}.
By (U5), $\cU_{+;\nu}$ and $\cM_{-\nu}$ have the same
dimension, so this is a square matrix.
The determinant of $B_{\nu}$ is the determinant of this matrix;
this is an element in $\cS(\fh)$ defined up to multiplication
by an invertible scalar. This is called a  {\em Shapovalov determinant}.
Define $B_{\nu}(\lambda)$ as the evaluation of $B_{\nu}$
at $\lambda$; clearly, this coincides with the restriction
of $B(\lambda)$ to $\cU_{+;\nu}\otimes M(\lambda)_{\lambda-\nu}$.
Since $M'(\lambda)=\{v|\  B(\lambda)(\cU_+,v)=0\}$
one has
$$\det B_{\nu}(\lambda)=0\ \Longleftrightarrow\
M'(\lambda)_{\lambda-\nu}\not=0.$$

\subsection{Jantzen-type filtrations}\label{Jan}
The Jantzen filtration and the sum formula were described
by Jantzen in~\cite{Jan}
for a Verma module over a semisimple Lie algebra.
We describe this construction in our setup, i.e. for
$\cU$ satisfying (U1)--(U5).

Assume  that $\cU$ satisfies (U1)--(U4).
Let $R$ be  the localization of  $\mathbb{C}[t]$ by the maximal ideal $(t)$.
We shall extend the scalars  from $\mathbb{C}$
to $R$. For a $\mathbb{C}$-vector superspace $V$ denote by $V_R$ the
$R$-module $V\otimes R$. Then $\cU_R, \cU_{\pm;R}$ are algebras
and $\cU_R=\cU_{-;R}\otimes \cS(\fh)_R\otimes\cU_{+;R}$.
Identify $\cS(\fh)_R$ with $\cS_R(\fh_R)$.
For any $\tilde{\lambda}\in\fh^*_R=\Hom_R (\fh_R,R)$ denote
by $M_R(\tilde{\lambda})$ the corresponding Verma module over $\cU_R$,
and denote by $v_{\tilde{\lambda}}$
the canonical generator of $M_R(\tilde{\lambda})$.
Define  a filtration on $M_R(\tilde{\lambda})$ as follows:
for $m\in\mathbb{Z}_{\geq 0}$
$$\begin{array}{l}
M_R(\tilde{\lambda})^m:=\{v\in M_R(\tilde{\lambda})|\ \cU_R v\cap
Rv_{\tilde{\lambda}}\subset Rt^m v_{\tilde{\lambda}}\}.
\end{array}$$

Clearly, $\{M_R(\tilde{\lambda})^m\}$ is a decreasing filtration,
$M_R(\tilde{\lambda})^0=M_R(\tilde{\lambda})$.
Let $\lambda\in\fh^*$ be the evaluation of
$\tilde{\lambda}$ at $t=0$ (the composition of
$\tilde{\lambda}: \fh^*_R\to R$ and $R\to R/tR=A$).
One has $M(\lambda)=M_R(\tilde{\lambda})/tM_R(\tilde{\lambda})$.
Denote by $\{M(\lambda)^m\}$ the image of
$\{M_R(\tilde{\lambda})^m\}$ in $M(\lambda)$.
Then $\{M(\lambda)^m\}$ is a decreasing filtration
of $M(\lambda)$ by $\cU$-submodules, and, by~(\ref{M'}), one has
\begin{equation}\label{M01}
M(\lambda)^0=M(\lambda),\ \ \ M(\lambda)^1=M'(\lambda).
\end{equation}
We call the filtration $\{M(\lambda)^m\}$ a {\em Jantzen-type filtration}.

Define the Shapovalov map $\cU_R\otimes \cM_R\to \cS(\fh)_R$ and
its evaluation $B(\tilde{\lambda}): \cU\otimes_R M_R(\tilde{\lambda})\to R$
as above. Then $B(\lambda)$ is the evaluation of $B(\tilde{\lambda})$ at $t=0$.
One readily sees that for $m\in\mathbb{Z}_{\geq 0}$ one has
\begin{equation}\label{U+Rv}
\begin{array}{l}
M_R(\tilde{\lambda})^m=\{v\in M_R(\tilde{\lambda})|\
B(\tilde{\lambda})  (\cU_{+;R},v)\subset R t^m\}.
\end{array}\end{equation}

Now assume that $\cU$ satisfies (U1)--(U5).
Observe that the Shapovalov matrix for $\cU_R$ written
with respect to bases lying in $\cU$ coincide with
the Shapovalov matrix for $\cU$ written
with respect to the same bases. Consequently,
the Shapovalov determinants $\det B_{\nu}\in \Sh$ viewed
as elements of the algebra $\cS(\fh)_R$ coincide with the Shapovalov
determinants $\det B_{\nu}$ constructed for ${\cU}_R$.

Recall that the matrix of a bilinear form on a free module of  finite rank
 over a local ring with values in this ring is diagonalizable. This implies
that, for each $\nu\in Q_+$, $\dim M(\lambda)^r_{\lambda-\nu}$
is equal to the corank of
the map $B_{\nu}(\tilde{\lambda})$
modulo $Rt^r$. Using the standard  reasoning of Jantzen~\cite{Jan},\cite{Tony}
we obtain the following sum formula
\begin{equation}\label{sumformula}
\sum_{r=1}^{\infty}\dim M(\lambda)^r_{\lambda-\nu}=
\upsilon(\det B_{\nu}(\tilde{\lambda})) \text{ for } \nu\in Q_+,
\end{equation}
 where $\upsilon: R\to\mathbb{Z}_{\geq 0}$ is given by
$\upsilon(a)=r$ if $a\in Rt^r\setminus (Rt^{r+1})$.
In particular, the above sum is finite iff
$\det B_{\nu}(\tilde{\lambda})\not=0$. Thus
\begin{equation}\label{Mrl}
\exists r:\ M(\lambda)^r_{\lambda-\nu}=0\ \Longleftrightarrow\
\det B_{\nu}(\tilde{\lambda})\not=0.
\end{equation}
We call a Jantzen-type filtration $\{M(\lambda)^m\}$  {\em non-degenerate}
if for each $\nu$ there exists $r$ such that $M(\lambda)^r_{\lambda-\nu}=0$.

We say that $\mu\in\fh^*$ is {\em $\lambda$-generic} if $\mu$ is transversal
to all hypersurfaces $\det B_{\nu}=0$ at point $\lambda$,
i.e. transversal to all irreducible components of $\det B_{\nu}=0$ passing
through $\lambda$ for each $\nu\in Q_+$, where $\det B_{\nu}$ is
non-zero. For example, for symmetrizable Kac-Moody  algebras,
$\det B_{\nu}\not=0$ for all $\nu$ and $\rho$
is $\lambda$-generic for each $\lambda\in\fh^*$.

The {\em Jantzen filtration} on the $\fg$-module $M(\lambda)$
is the filtration $\{M(\lambda)^m\}$ for $\tilde{\lambda}=\lambda+t\nu$, where
$\nu$ is $\lambda$-generic. We do not know whether the Jantzen filtration
depends on $\nu$ (for semisimple Lie algebras the fact that
the Jantzen filtration does not depend on $\nu$ follows from~\cite{BB}).
Below we will use more general filtrations, taking
$\tilde{\lambda}=\lambda+t\lambda_1+t^2\lambda_2$, where $\lambda,\lambda_1,
\lambda_2\in\fh^*$. In our setup $\lambda_1$ will be fixed and
not always $\lambda$-generic, but we will take $\lambda_2$ to be
$\lambda$-generic; then, the Jantzen-type filtration
$\{M(\lambda)^m\}$  is non-degenerate by~(\ref{Mrl}).

\subsection{The map $\Upsilon:\Ext^1(L(\lambda),L(\lambda))\to \fh^*$}
\label{mapUps}
In this subsection we introduce for a highest weight module $M$ a
map $\Upsilon_{M}:\Ext^1_{\fg}(M,M)\to\fh^*$, and
establish some useful properties of this map (see~\Cor{corally}).

\subsubsection{Definition of $\Upsilon$}
\label{extlala}
Let $M$ be a module with the  highest weight $\lambda$
(i.e. a quotient of $M(\lambda)$), and let $v_{\lambda}\in M$ be
the highest weight vector, i.e. the image of the canonical generator
of $M(\lambda)$.
Introduce the natural map
$$\Upsilon_{M}:\Ext^1_{\fg}(M,M)\to\fh^*$$
as follows. Let $0\to M\to^{\phi_1} N\to^{\phi_2} M\to 0$ be
an exact sequence. Let
$v:=\phi_1(v_{\lambda})$ and fix $v'\in N_{\lambda}$ such
that $\phi_2(v')=v_{\lambda}$. Observe that $v,v'$
is a basis of $N_{\lambda}$ and so there exists $\mu\in\fh^*$ such that
 for any $h\in\fh$ one has $h(v')=\lambda(h)v'+\mu(h)v$
(i.e., the representation $\fh\to\End(N_{\lambda})$
is  $h\mapsto \begin{pmatrix} \lambda(h) & \mu(h)\\ 0 &\lambda(h)
\end{pmatrix}$). The map $\Upsilon_{M}$ assigns $\mu$ to the exact sequence.

Notice that if $0\to M\to N_1\to M\to 0$
and $0\to M\to N_2\to M\to 0$ are two exact sequences then
$$N_1\cong N_2\ \Longleftrightarrow\
\Upsilon_{M}(N_1)=c\Upsilon_{M}(N_2)\ \text{ for some }
c\in\mathbb{C}\setminus\{0\}.$$
If $N$ is an extension of $M$ by $M$ (i.e., $N/M\cong M$)
we denote by $\Upsilon_{M}(N)$ the corresponding one-dimensional subspace
of $\fh^*$, i.e. $\Upsilon_{M}(N)=\mathbb{C}\mu$, where
$\mu$ is the image of the exact sequence $0\to M\to N\to M\to 0$.

\subsubsection{}
\begin{lem}{corUps}
The map $\Upsilon_{M}$ has the following properties:

($\Upsilon 1$) $\Upsilon_{M}:\Ext^1(M,M)\to \fh^*$ is injective;

($\Upsilon 2$) $\Upsilon_{M(\lambda)}:\Ext^1(M(\lambda),M(\lambda))
\to \fh^*$ is bijective;

($\Upsilon 3$) if $\Upsilon_{M(\lambda)}(N)=\Upsilon_{M}(N')$,
then $N'$  is isomorphic to a quotient of $N$;

($\Upsilon 4$) if $\Upsilon_{M(\lambda)}(N)=\Upsilon_{L(\lambda)}(N')$,
then $N'\cong N/N''$, where $N''$ is the maximal submodule of $N$
which intersects trivially the highest weight space $N_{\lambda}$.
\end{lem}
\begin{proof}
Property ($\Upsilon 1$) is clear. For ($\Upsilon 2$) let us
construct a preimage of a non-zero element $\mu\in\fh^*$.
Take a two-dimensional $\fh$-module
$E$ spanned by $v,v'$ such that $h(v)=\lambda(h)v,\
 h(v')=\lambda(h)v'+\mu(h)v$, view $E$ as a trivial $\cU_+$-module
(i.e., $\cU'_+E=0$), and consider the $\cU$-module
$N:=\cU\otimes_{\cS(\fh)\otimes\cU_+}E$.
Clearly, $\Upsilon_{M(\lambda)}(N)=\mu$, proving ($\Upsilon 2$).
Observe that $N$ is universal in the following sense:
any module $N'$ satisfying
(i) $N'_{\lambda}\cong E$ as $\fh$-module,
(ii) $\cU'_+ N'_{\lambda}=0$,
(iii) $N'$ is generated by $N'_{\lambda}$,
is a quotient of $N$. In particular, if $M$ is a quotient
of  $M(\lambda)$ and
$\Upsilon_{M(\lambda)}(N)=\Upsilon_{M}(N')$,
then $N'$  is a quotient of $N$. This proves ($\Upsilon 3$).

For ($\Upsilon 4$)
suppose that  $\Upsilon_{M(\lambda)}(N)=\Upsilon_{L(\lambda)}(N')$.
 We claim that $N'\cong N/N''$, where
$N''$ is the maximal submodule of $N$ which intersects
$N_{\lambda}$ trivially. Indeed,
by above, $N'$ is   a quotient of $N$: $N'\cong N/X$.
Since $\dim N'_{\lambda}=
\dim N_{\lambda}=2$, $X_{\lambda}=0$, hence $X\subset N''$.
Then $N''/X$ is isomorphic to a submodule of $N'$ and
$(N''/X)_{\lambda}=0$. Hence $N''/X=0$ as required.
\end{proof}

\subsubsection{The map $\Upsilon_{L(\lambda)}$}
\label{KMVirSchw}
Let us describe the image $\im\Upsilon_{L(\lambda)}$
in terms of the Jantzen-type filtration. For $\mu,\mu'\in\fh^*$
let $\{M(\lambda)^i\}$ be the non-degenerate Jantzen-type filtration
described in Subsection~\ref{Jan}, which is the image of
the  filtration $\{M_R(\lambda+t\mu+t^2\mu')^i\}$  in
$M(\lambda)=M_R(\lambda+t\mu+t^2\mu')/tM_R(\lambda+t\mu+t^2\mu')$.

Set $\tilde{M}:=M_R(\lambda+t\mu+t^2\mu'),\ \
\tilde{N}:=\tilde{M}/t^2 \tilde{M}$
and view $\tilde{N}$ as a $\fg$-module. Observe that $\tilde{N}$
has a submodule $t\tilde{N}\cong M(\lambda)$ and that
$\Upsilon_{M(\lambda)}$
maps the exact sequence
$0\to t\tilde{N}\to \tilde{N}\to M(\lambda)\to 0$ to $\mu$.

\subsubsection{}
\begin{cor}{coruex}
If $\Upsilon_{L(\lambda)}: \Ext^1(L(\lambda),L(\lambda))\to \fh^*$ maps
$0\to L(\lambda)\to N\to L(\lambda)\to 0$ to $\mu$,  then
$N\cong M_R(\lambda+t\mu+t^2\mu')/M_R(\lambda+t\mu+t^2\mu')^2$.
\end{cor}
\begin{proof}
Retain notation of~\ref{shmap} and \ref{Jan}.
Let $\{\tilde{M}^i\}$  be  the Jantzen filtration of
$\tilde{M}=M_R(\lambda+t\mu+t^2\mu')$ and
$\{\tilde{N}^i\}$ be its image  in $\tilde{N}=\tilde{M}/t^2 \tilde{M}$.
In light of ($\Upsilon 4$), it is enough to show that $\tilde{N}^2$ is
the maximal submodule of $\tilde{N}$ which intersects
$\tilde{N}_{\lambda}$ trivially.
Let $v\in\tilde{M}$ be the canonical generator of
the Verma module $\tilde{M}$. Clearly, $\tilde{M}^2\cap Rv=Rt^2v$,
so $\tilde{N}^2_{\lambda}=0$.
If $u_-\in\cU_{-,R}$ is a weight element such that
$u_-v\not\in \tilde{M}^2$ then, by~(\ref{U+Rv}), there exists a weight element
$u_+\in\cU(\fn_+)$
such that $\HC(u_+u_-)(\lambda+t\mu+t^2\mu')\not\in Rt^2$ so
$u_+(u_-v)\in Rv\setminus Rt^2v$. As a result, the submodule spanned by
the image of $u_-v$ in $\tilde{N}$ intersects $\tilde{N}_{\lambda}$
non-trivially. The assertion follows.
\end{proof}

\subsubsection{}
\begin{prop}{corally}
Let $\lambda,\mu,\mu'\in\fh^*$, and let $\{M(\lambda)^i\}$ be the
image  in $M(\lambda)$ of the Jantzen filtration
$\{M_R(\lambda+t\mu+t^2\mu')^i\}$.
Then
$$\mu\in\im\Upsilon_{L(\lambda)}\ \Longleftrightarrow\
M(\lambda)^1=M(\lambda)^2.$$
\end{prop}
\begin{proof}
Recall that, by~(\ref{M01}), $M(\lambda)^1$ is a maximal
proper submodule of $M(\lambda)$.
Retain notation of the proof of~\Cor{coruex}.
In order to deduce the assertion  from~\Cor{coruex}
we have to show that $\tilde{M}/\tilde{M}^2$
is an extension of $L(\lambda)$ by $L(\lambda)$ iff
$M(\lambda)^1=M(\lambda)^2$.
Recall that $M(\lambda)^i=\phi_t(\tilde{M}^i)$, where
$\phi_t: \tilde{M}\to \tilde{M}/t\tilde{M}=M(\lambda)$.
Thus $M(\lambda)/M(\lambda)^1$ is isomorphic to
$\tilde{M}/(\tilde{M}^1+t\tilde{M})=\tilde{M}/\tilde{M}^1$ and
so $\tilde{M}/\tilde{M}^1\cong L(\lambda)$ by (i).
One has the exact sequence of $\cU$-modules
$$0\to \tilde{M}^1/\tilde{M}^2\to
\tilde{M}/\tilde{M}^2\to \tilde{M}/\tilde{M}^1\cong L(\lambda)\to 0.$$
As a result, $\tilde{M}/\tilde{M}^2$ is an extension of
$L(\lambda)$ by $L(\lambda)$ iff $\tilde{M}^1/\tilde{M}^2\cong L(\lambda)$.
By above, $M(\lambda)^1/M(\lambda)^2$ is isomorphic to
the quotient of $\tilde{M}^1/\tilde{M}^2$ by
$(t\tilde{M}+\tilde{M}^2)/\tilde{M}^2$. One has
$$(t\tilde{M}+\tilde{M}^2)/\tilde{M}^2\cong
t\tilde{M}/(t\tilde{M}\cap\tilde{M}^2)=t\tilde{M}/t\tilde{M}^1\cong
\tilde{M}/\tilde{M}^1=L(\lambda),$$
because $t\tilde{M}\cap\tilde{M}^2=t\tilde{M}^1$.
Hence $\tilde{M}^1/\tilde{M}^2\cong L(\lambda)$ iff
$M(\lambda)^1/M(\lambda)^2=0$; this completes the proof.
\end{proof}

\subsection{Duality}\label{dual}
Assume that $\cU$ satisfies (U1)--(U6).

For a $\cU$-module $N$ we view $N^*$ as a left $\cU$-module
 via the antiautomorphism $\sigma$ (see (U6)): $af(n)=f(\sigma(a)n),\ a\in\cU,
f\in N^*, n\in N$.
If $N$ is a locally finite $\fh$-module with finite-dimensional
generalized weight spaces, $N^*$ has a submodule
$N^{\#}:=\oplus_{\nu\in\Omega(N)} N^*_{\nu}$. One has $(N^{\#})^{\#}\cong N$.
Clearly, $\Omega(N^{\#})=\Omega(N)$. This implies
$L(\nu)^{\#}\cong L(\nu)$ for any $\nu\in\fh^*$.

Retain notation of~\ref{mapUps}.
The following  criterion will be used in Subsection~\ref{UpsWL}.
\subsubsection{}
\begin{lem}{lemcri}
Let  $0\to M(\lambda)\to N\to M(\lambda)\to 0$
be a preimage of $\mu\in\fh^*$ under $\Upsilon_{M(\lambda)}$
in $\Ext^1(M(\lambda),M(\lambda))$. Then
 $\mu\not\in \im\Upsilon_{L(\lambda)}$ iff $N$ has a subquotient
which is isomorphic to a submodule of
$M(\lambda)^{\#}$ and is not isomorphic to $L(\lambda)$.
\end{lem}
\begin{proof}
Set $M:=M(\lambda),\ L:=L(\lambda)$.
Let $N'$ be the maximal submodule of $N$ which intersects
$N_{\lambda}$ trivially.
By ($\Upsilon 3$) of ~\Lem{corUps}, $\mu\in \im\Upsilon_L$ iff $N/N'$
is an extension of $L$ by $L$.

Assume that $\mu\in \im\Upsilon_L$.
Suppose that $N_1\subset N_2$ are submodules of $N$ and $N_2/N_1$
is isomorphic to a submodule of $M^{\#}$.
Since $M$ is generated by its highest  weight  space,
 any submodule of $M^{\#}$ intersects its highest weight space
$M^{\#}_{\lambda}$ non-trivially.
For any $v\in N_2\cap N'$ a submodule generated by $v$ intersects
the highest weight space $N_{\lambda}$ trivially, so $v\in N_1$.
Thus $N_2\cap N'\subset N_1$ and so $N_2/N_1$ is a quotient
of $N_2/(N_2\cap N')$. In its turn, $N_2/(N_2\cap N')$
is a submodule of $N/N'$ so $N_2/N_1$ is a subquotient of $N/N'$. Since
$N/N'$ is an extension of $L$ by $L$, $N_2/N_1\cong L$
or $N_2/N_1\cong N/N'$. Since $\dim M^{\#}_{\lambda}=1$, $M^{\#}$
 does not have a submodule
isomorphic to an extension of $L$ by $L$. Hence $N_2/N_1\cong L$.

Now assume that $\mu\not\in \im\Upsilon_L$ so, by ($\Upsilon 3$),
$N/N'$ is not an extension of $L$ by $L$.
Write
$$0\to M\to^{\iota} N\to^{\phi} M\to 0.$$
Let $M'$ be a maximal proper submodule of $M$. Note that
$N'\subset \phi^{-1}(M')$. Let us show that
$E:=\phi^{-1}(M')/N'$ is the required subquotient of $N$.
Indeed, since $E$ is a submodule of $N/N'$,
any submodule of $E$ intersects $E_{\lambda}$ non-trivially.
Clearly, $E_{\lambda}$ is one-dimensional.
As a result, $E^{\#}$ is a quotient of $M$, so
$E$ is a submodule of $M^{\#}$.
One has
$$(N/N')/E= (N/N')/(\phi^{-1}(M')/N')\cong N/\phi^{-1}(M')\cong L.$$
By above, $N/N'$ is not an extension of $L$ by $L$. Hence
$E\not\cong L$ as required.
\end{proof}

\subsection{Weakly admissible modules}\label{weakad}
It is well known that representation theory of an affine Lie algebra
at the critical level $k=-h^{\vee}$ is much more complicated than for
a non-critical level $k\not=-h^{\vee}$. For any Kac-Moody algebra one has
a similar set of critical weights
$$C=\{\nu\in\fh^*|\ 2(\nu+\rho,\alpha)\in\mathbb{Z}_{>0}(\alpha,\alpha)
\ \text{ for some }
\alpha\in\Delta_+^{im}\}.$$

In our general setup, introduced in Subsection~\ref{assm}, fix a subset
$C\subset \fh^*$,
called the subset of {\em critical weights}.
We call $\lambda\in\fh^*$ and the corresponding irreducible $\cU$-module
$L(\lambda)$ {\em weakly admissible} if $\lambda\not\in C$
and for any $\nu\not=\lambda$, such that $\nu\not\in C$, $L(\lambda)$ is
not a subquotient of the Verma module $M(\nu)$.

For example, for  Kac-Moody algebras $C$ is as above,
for the Virasoro and  for the Neveu-Schwartz algebras,
$C$ is empty. For minimal $\cW$-algebras $W^k(\fg,e_{-\theta})$,
$C$ is empty if $k\not=-h^{\vee}$.

The following lemma is similar to the one proven in~\cite{DGK}.

\subsubsection{}
\begin{lem}{proplala}
Let $\lambda,\lambda'\in\fh^*$  be two distinct elements.
If the exact sequence $0\to L(\lambda')\to N\to L(\lambda)\to 0$
is non-splitting, then either $N$ is a quotient of $M(\lambda)$ or
     $N^{\#}$ is a quotient of $M(\lambda')$.
\end{lem}
\begin{proof}
Consider the case $\lambda'-\lambda\not\in Q_+$. Then $\lambda$ is a maximal
element in $\Omega(N)$ and $N_{\lambda}$ is one-dimensional.
Therefore $N_{\lambda}$ generates
a submodule $N'$ which is isomorphic to a quotient of $M(\lambda)$.
Since $N$ is indecomposable,  $N'=N$ and
so $N$ is a quotient of $M(\lambda)$.

Let now $\lambda'-\lambda\in Q_+$. Recall that
$L(\nu)^{\#}=L(\nu)$, hence we have the dual  exact sequence
$0\to L(\lambda)\to N^{\#}\to L(\lambda')\to 0$.
By the first case, $N^{\#}$ is a quotient of $M(\lambda')$.
\end{proof}

\subsubsection{}
\begin{cor}{corlala}
If $\lambda\not=\lambda'$ are weakly admissible weights,
then $\Ext^1(L(\lambda),L(\lambda'))=0$.
\end{cor}

We will use the following criterion of weakly admissibility.
\subsubsection{}
\begin{lem}{critwa}
The weight $\lambda\in\fh^*$ is not weakly admissible iff there exists
$\lambda'\in\lambda +Q_+$ such that $\det B_{\lambda'-\lambda}(\lambda')=0$.
\end{lem}
The proof is similar to the proof of Thm. 2 in~\cite{KK}.

\subsection{Admissible modules}\label{admis}
Let $\cH$ be a category of $\cU$-modules. We say that
an irreducible module $L(\lambda)$ in $\cH$ is {\em $\cH$-admissible}
if it is weakly admissible and
$\Ext^1_{\cH}(L(\lambda),L(\lambda))=0$.

If $\cU=\cU(\fg)$, where $\fg$
is a Kac-Moody algebra, we call $L(\lambda)$ and the corresponding
highest weight $\lambda$ {\em admissible }
if it is admissible in the category of $\fg$-modules with a diagonal action
of $\fh''$, where $\fg=[\fg,\fg]\oplus\fh'',\ \fh''\subset\fh$.
For an affine Lie algebra $\fg$ (in this case $\fh''=\mathbb{C}D$),
we call $L(\lambda)$ and the corresponding
highest weight $\lambda$ {\em $k$-admissible}
($k\in\mathbb{C}$) if it is admissible in the category of
$\fg$-modules with a diagonal action of $D$,
which are annihilated by $K-k$.

We call a $\cU$-module
$M$  {\em bounded} if $\fh$ acts locally finitely on $M$ and
$\Omega(M)$ is bounded by a finite set of elements of $\fh^*$.
Denote by $\tilde{\CO}$ the category of all bounded
$\cU$-modules, and by $\tilde{\CO}_{adm}$ its subcategory, consisting
of all bounded modules lying in $\cH$,
all of whose irreducible subquotients are $\cH$-admissible.

\subsubsection{}
\begin{cor}{corCO'}
Any $\cU$-module from the category  $\tilde{\CO}_{adm}$
is completely reducible.
\end{cor}
\begin{proof}
By~\Lem{lemcH}, it is enough to show that any module $N$ in
$\tilde{\CO}_{adm}$
has an irreducible submodule. Let $\lambda$ be a maximal element in $\Omega(N)$
(this exists since $\Omega(N)$ is bounded).
Let $v\in N_{\lambda}$ be an $\fh$-eigenvector. Then a submodule $N'$,
generated by $v$, is a quotient of $M(\lambda)$. A quotient
of $M(\lambda)$, which is not irreducible,
has a non-splitting quotient of length two. However,
$N'$ does not have such a quotient, since
$\Ext^1(L,L')=0$ for
irreducible modules $L\not\cong L'$ in $\tilde{\CO}_{adm}$.
Hence $N'$ is irreducible as required.
\end{proof}

\section{$\Ext^1(L(\lambda),L(\lambda))$ for the Lie superalgebra
$\fg(A,\tau)$}

In this section we prove Propositions~\ref{prop1} and~\ref{propalmin}.

\subsection{The construction of $\fg(A,\tau)$}\label{assmA}
Let $A=(a_{ij})$ be an $(n\times n)$-matrix over $\mathbb{C}$
 and let $\tau$ be a subset of $I:=\{1,\ldots,n\}$.
Let $\fg=\fg(A,\tau)=\fn_-\oplus\fh\oplus\fn_+$
be the associated  Lie superalgebra constructed as in~\cite{K77},\cite{Kbook}.
Recall that, in order to construct $\fg(A,\tau)$, one considers a
realization of $A$, i.e. a triple $(\fh,\Pi,\Pi^{\vee})$, where $\fh$
is a vector space of dimension $n+corank A$, $\Pi\subset\fh^*$
(resp. $\Pi^{\vee}\subset\fh$) is a linearly independent set of vectors
$\{\alpha_i\}_{i\in I}$ (resp. $\{\alpha_i^{\vee}\}_{i\in I}$),
such that $\langle \alpha_i,\alpha_j^{\vee}\rangle=a_{ji}$,
and constructs a Lie superalgebra $\tilde{\fg}(A,\tau)$
on generators $e_i,f_i,\fh$, subject to relations:
$$\begin{array}{l}
[\fh,\fh]=0,\ \ [h,e_i]=\langle \alpha_i,h\rangle e_i,\ \
[h,f_i]=-\langle \alpha_i,h\rangle f_i,\ \text{ for } i\in I,h\in\fh,\ \
[e_i,f_j]=\delta_{ij}\alpha_i^{\vee},\\
\ p(e_i)=p(f_i)=\ol{1} \text{ if } i\in\tau,\ \
p(e_i)=p(f_i)=\ol{0} \text{ if } i\not\in\tau,\ \ \ p(\fh)=\ol{0}.
\end{array}$$
Then $\fg(A,\tau)=\tilde{\fg}(A,\tau)/J=\fn_-\oplus\fh\oplus\fn_+$,
where $J$ is the maximal ideal of
$\tilde{\fg}(A,\tau)$, intersecting $\fh$ trivially, and
$\fn_+$ (resp. $\fn_-$) is the subalgebra generated by the images of the
$e_i$'s (resp. $f_i$'s). One readily sees that if $a_{ii}=0$
and $i\in\tau$ (i.e., $e_i,f_i$ are odd), then $e_i^2,f_i^2\in J$
so $e_i^2=f_i^2=0$ in $\fg(A,\tau)$.

We say that a simple root $\alpha_i$ is {\em even} (resp., {\em odd})
if $i\not\in\tau$ (resp., $i\in\tau$) and
that $\alpha_i$ is {\em isotropic} if $a_{ii}=0$.

Let $\cU$ (resp. $\cU_{\pm}$) be the universal enveloping algebra
of $\fg$ (resp. $\fn_{\pm}$).
Observe that $(\cU,\fh,\cU_{\pm})$ satisfies the assumptions (U1)--(U6)
of Subsection~\ref{assm} with $\sigma$ identical on
$\fh$ and $\sigma(e_i)=f_i$.
As before, we let $\fh'=\fh\cap[\fg,\fg]$. By the construction,
the simple coroots $\alpha_i^{\vee}, i\in I,$ form a basis of $\fh'$.

Note that, multiplying the $i$th row of the matrix $A$ by a non-zero number
corresponds to multiplying $e_i$ and $\alpha^{\vee}_i$ by this number,
thus giving an isomorphic Lie superalgebra. Hence we may assume
from now on that $a_{ii}=2$ or $0$ for all $i\in I$.

The above construction includes all Kac-Moody algebras, all basic
finite-dimensional Lie superalgebras and the associated affine
superalgebras, and the above assumption holds for all of them (but it does
not hold for generalized Kac-Moody Lie algebras).

\subsection{}
\begin{prop}{prop1}
Let $\fg=\fg(A,\tau)$ and
assume that $a_{ii}\not=0$ for all $i\not\in\tau$. Then:

(i) For any $\lambda\in\fh^*$ the image of $\Upsilon_{L(\lambda)}$ contains
the subspace $\{\mu\in\fh^*|\ \langle \mu,\fh'\rangle=0\}$.

(ii) Let $\lambda\in\fh^*$ be such that
$\langle \lambda,\alpha^{\vee}\rangle$ is a non-negative integer
(resp. an even non-negative integer, resp. zero) for any even
(resp. odd non-isotropic, resp. odd isotropic) simple root $\alpha$.
Then
$$\im\Upsilon_{L(\lambda)}=\{\mu\in\fh^*|\ \langle\mu,\fh'
\rangle=0\}.$$
In other words, under the above condition on $\lambda$,
$\fh'$ acts semisimply on any extension of $L(\lambda)$ by $L(\lambda)$.
Writing $\fh=\fh'\oplus\fh''$, we obtain that the extension splits
if $\fh''$ acts semisimply.
\end{prop}
\begin{proof}
For (i) fix $\mu$ such that $\langle\mu,\fh'\rangle=0$,
set $\mu'=0$ and
retain notation of~\ref{KMVirSchw}. In the light of~\Prop{corally},
it is enough to verify that $M(\lambda)^1=M(\lambda)^2$.
Denote the highest
weight vector of $M_R(\lambda+t\mu)$ by $v_{\lambda+t\mu}$
and its image, the highest weight vector of $M(\lambda)$,
by $v_{\lambda}$. Let $u\in\cU(\fn_-)$ be such that
 $uv_{\lambda}\in M(\lambda)^1$. Then for any $u_+\in\cU(\fn_+)$
one has $\HC(u_+u)(\lambda)=0$. Since $\HC(u_+u)=\HC([u_+,u])\in
\cU([\fg,\fg])$, we conclude that $\HC(u_+u)\in\cS(\fh')$ and
so $\HC(u_+u)(\lambda+t\mu)=\HC(u_+u)(\lambda)=0$.
Hence $uv_{\lambda+t\mu}\in M_R(\lambda+t\mu)^i$ for all $i$, and thus
$uv_{\lambda}\in M(\lambda)^2$, as required.

(ii) Assume that $N$ is an extension of $L(\lambda)$ by $L(\lambda)$.
Fix a simple root $\alpha$ and let
$e,h:=\alpha^{\vee}, f$ be the Chevalley generators of $\fg$
corresponding to $\alpha$. Recall that $e^2=f^2=0$
if $\alpha$ is odd and isotropic. Otherwise one has
\begin{equation}\label{HCef}
\begin{array}{ll}\HC(e^kf^k)=k!h(h-1)\ldots (h-(k-1)), &
\text{ if $\alpha$ is even},\\
\HC(e^{2k}f^{2k})=(-1)^k2^kk!h(h-2)\ldots (h-2(k-1)), &
\text{ if  $\alpha$ is odd non-isotropic},\\
\HC(e^{2k+1}f^{2k+1})=(-1)^k2^kk!h(h-2)\ldots (h-2k), &
\text{ if $\alpha$ is odd non-isotropic},\\
\HC(ef)=h, &
\text{ if $\alpha$ is odd isotropic},
\end{array}\end{equation}
where $\HC$ stands for the Harish-Chandra projection defined in
Subsection~\ref{HC}.
Set $k:=\langle \lambda,\alpha^{\vee}\rangle+1$.
From~(\ref{HCef}) it follows that
$L(\lambda)_{\lambda-k\alpha}=0$ so $N_{\lambda-k\alpha}=0$.
For $v\in N_{\lambda}$ one has $f^kv\in N_{\lambda-k\alpha}=0$
so $e^kf^kv=0$; however, $\fn_+v=0$ so
$e^kf^kv=\HC(e^kf^k)v$. Therefore the polynomial $P(h):=\HC(e^kf^k)
\in\mathbb{C}[h]$
annihilates $N_{\lambda}$.  From~(\ref{HCef}) one sees that
$P(h)$ does not
have multiple roots. Since $P(h)N_{\lambda}=0$, we conclude that
$h=\alpha^{\vee}$ acts diagonally on $N_{\lambda}$.
Since $\fh'$  is spanned by the simple coroots,
$\fh'$ acts diagonally on $N_{\lambda}$.
This means that $\langle\mu,\fh'\rangle=0$ for any
$\mu\in\im\Upsilon_{L(\lambda)}$.
\end{proof}

\subsubsection{}
\begin{rem}{remprop1}
The  exact sequence $0\to L(\lambda)\to N\to L(\lambda)\to 0$
splits over $[\fg,\fg]$ iff $\Upsilon_{L(\lambda)}$ maps
this sequence to $\mu\in \fh^*$  such that $\langle\mu,\fh'\rangle=0$.
\end{rem}

\subsection{Symmetrizable case}
\label{symmKm}
Assume that the matrix $A$ is symmetrizable, i.e.
for some invertible diagonal $n\times n$-matrix $D$
the matrix $DA$ is symmetric.
In this case the Lie superalgebra $\fg=\fg(A,\tau)$ admits
a non-degenerate invariant bilinear form $(-,-)$~\cite{Kbook}, Chapter 2.
We denote  by $(-,-)$ also the induced non-degenerate bilinear form on
$\fh^*$. For $X\subset \fh^*$ we set
$X^{\perp}:=\{\mu\in\fh^*| (\mu,X)=0\}$.

We denote by $\Delta_+$ (resp. $\Delta_-$) the
set of weights of $\ad\fh$ on $\fn_+$ (resp. $\fn_-$).
Introduce $\rho\in\fh^*$ by the conditions
$\langle \rho,\alpha^{\vee}_i\rangle=
a_{ii}/2$ for each simple coroot $\alpha^{\vee}_i$.
Notice that for any simple non-isotropic root $\alpha$
and any $\mu\in\fh^*$ one has
$\langle \mu,\alpha^{\vee}\rangle=2(\alpha,\mu)/(\alpha,\alpha)$.
For any non-isotropic $\alpha\in\Delta_+$ introduce $\alpha^{\vee}\in\fh$
by the condition
$\forall \mu\in\fh^*\ \
\langle \mu,\alpha^{\vee}\rangle=\frac{2(\mu,\alpha)}{(\alpha,\alpha)}$.

Recall that a $\fg$-module (or $[\fg,\fg]$-module)
$N$ is called {\em restricted}
if for every $v\in N$, we have $\fg_{\alpha}v=0$
for all but a finite number of positive roots $\alpha$.

Recall the standard construction of the Casimir operator
(\cite{Kbook}, Chapter 2). Choose the dual bases $\{h_i\}, \{h^i\}$
in $\fh$ (i.e., $(h_i,h^j)=\delta_{ij}$) and for each positive root
$\alpha\in\Delta_+$ choose the dual bases $\{a_i\}$ in $\fg_{\alpha}$
and $\{b_i\}\in\fg_{-\alpha}$ (i.e., $(a_i,b_j)=\delta_{ij}$).
Let $\rho^*\in\fh$ be such that $\langle \mu,\rho^*\rangle=(\mu,\rho)$
for each $\mu\in\fh^*$.
Define the Casimir operator by the formula
$$Cas:=2\rho^*+
\sum_i h_ih^i+2\sum_{\alpha\in\Delta_+}\sum_i b_ia_i$$
Note that  the Casimir operator
acts on any restricted $\fg$-module by a $\fg$-endomorphism.

\subsubsection{}\label{CC}
Recall the formula for Shapovalov determinant
for an arbitrary symmetrizable $\fg(A,\tau)$~\cite{GK}
(in the non-super case it was proven in~\cite{KK}):
\begin{equation}\label{shap}
det_{\nu}(\lambda)=\prod_{\alpha\in\Delta_+}\prod_{n=1}^{\infty}
\bigl(2(\lambda+\rho,\alpha)-n(\alpha,\alpha)\bigr)^{(-1)^{p(\alpha)
(n+1)}P(\nu-n\alpha)},
\end{equation}
where $p(\alpha)$ is the parity of $\alpha$ and $P$ is
the partition function for $\fg$, i.e. $P(\mu)$
is the number  of partitions of $\mu$ in a sum
of positive roots, where each odd root occurs at most once,
counting their multiplicities. Observe that, if
$\alpha, 2\alpha\in \Delta_+$ and $p(\alpha)=\ol{1}$, then
the factor corresponding to $(\alpha,2n)$ cancels with
the factor corresponding to $(2\alpha,n)$.

\subsubsection{}\label{setS}
Introduce $S\subset \fh^*$ by the condition: $\lambda\in S$ iff

(S1) if $\alpha\in\Delta_+$ is even non-isotropic
and $\alpha/2$ is not an odd root, then
$\langle\lambda+\rho, \alpha^{\vee}\rangle\not\in\mathbb{Z}_{<0}$;

(S2) if $\alpha\in\Delta_+$ is odd non-isotropic then
$\langle\lambda+\rho, \alpha^{\vee}\rangle$ is not a negative odd integer;

(S3) $(\lambda+\rho, \beta)\not=0$ for every isotropic root $\beta$.

\subsubsection{}
\begin{lem}{corsetS}
(i) If $\lambda,\lambda'\in S$ are such that
$[M(\lambda):L(\lambda')]\not=0$ then $\lambda=\lambda'$.

(ii) $\Ext^1(L(\lambda'),L(\lambda))=0$ for  $\lambda,\lambda'\in S$,
$\lambda\not=\lambda'$.
\end{lem}
\begin{proof}
For $\lambda\in S$ one has $(\lambda+\rho, \beta)\not=0$
for any isotropic root $\beta$. Using~(\ref{shap}) and the standard
argument~\cite{Jan},\cite{KK}, we conclude that
if $[M(\lambda):L(\nu)]\not=0$ and $\nu\not=\lambda$, then
either $\nu=\nu'-\alpha$, where $\alpha$ is isotropic
and $(\nu'+\rho,\alpha)=0$, or
 $\nu=\nu'-k\alpha$ for some non-isotropic root
$\alpha\in\Delta^+$ such that
$k:=\langle \nu'+\rho,\alpha^{\vee}\rangle$
is a positive (resp. odd positive)
integer if $\alpha$ is  even  such that
$\alpha/2$ is not an odd root (resp. $\alpha$ is odd).
 In the first case $(\nu+\rho,\alpha)=0$
 and in the second case $\langle \nu+\rho,\alpha^{\vee}\rangle=-k$.
This proves (i). Combining~\Lem{proplala} and (i) we get (ii).
\end{proof}

\subsubsection{}\label{CC'}
\begin{defn}{defmi}
For $\lambda\in\fh^*$ let $C'(\lambda)\subset \Delta_+\times \mathbb{Z}_{>0}$
consist of the pairs
$(\alpha,m)$, such that  $2(\lambda+\rho,\alpha)=m(\alpha,\alpha)$ and

(i) if $\alpha$ is even, then $\alpha/2$
is not an odd root;

(ii) if $\alpha$ is odd and non-isotropic, then
$m$ is odd.

By the argument of~\cite{KK} Thm. 2, if $(\alpha,m)\in C'(\lambda)$ and
$\alpha$ is not isotropic (resp. $\alpha$ is isotropic), then $M(\lambda)$
has a singular vector of weight $\lambda-m\alpha$ (resp. $\lambda-\alpha$),
and, moreover, if $[M(\lambda):L(\lambda-\nu)]\not=0$, then either
$\nu=0$ or $[M(\lambda-m\alpha):L(\lambda-\nu)]\not=0$
for some $(\alpha,m)\in C'(\lambda)$. In particular,
$M(\lambda)$ is irreducible iff $C'(\lambda)$ is empty.

Let $C(\lambda)$ be the projection of $C'(\lambda)$ to $\Delta_+$, i.e.
$C(\lambda):=\{\alpha\in \Delta_+|\ \exists m: (\alpha,m)\in C'(\lambda)\}$.
We call $\alpha\in C(\lambda)$ {\em $\lambda$-minimal} if for some pair
$(\alpha,m)\in C'(\lambda)$ one has
\begin{equation}\label{lambdamin}
\forall (\beta, n)\in C'(\lambda)\setminus\{(\alpha,m)\}\ \
[M(\lambda-n\beta):L(\lambda-m\alpha)]=0.\end{equation}
\end{defn}

Notice that,  if an isotropic root $\alpha$ is $\lambda$-minimal, then
formula~(\ref{lambdamin}) holds for the pair $(\alpha,1)$.
Observe that for a non-isotropic root $\alpha$ there exists at most one value
of $m$ such that $(\alpha,m)\in C'(\lambda)$.

Let  $\alpha$ be such that $m\alpha$ is a minimal element
of the set $\{n\beta| (\beta,n)\in C'(\lambda)\}$ with respect to the order
introduced in~\ref{Qpm}.
If $\alpha$ is such that $m\alpha=n\beta$ for $(\beta,n)\in C'(\lambda)$
forces $\alpha=\beta$, then $\alpha$ is $\lambda$-minimal.
Notice that for isotropic $\alpha$, if $m\alpha$ is a minimal element
of the set $\{n\beta| (\beta,n)\in C'(\lambda)\}$, then $m=1$
and $\alpha\not=n\beta$ for $(\beta,n)\in C'(\lambda)\setminus\{(\alpha,1)\}$
so $\alpha$ is $\lambda$-minimal. If $\alpha\in C(\lambda)$ is a simple root,
then $\alpha$ is  $\lambda$-minimal, because
$m\alpha=n\beta$ forces $\alpha=\beta$.

\subsubsection{}
\begin{prop}{propalmin}
Let $\fg=\fg(A,\tau)$, where $A$ is  a symmetrizable matrix,
such that $a_{ii}\not=0$ for $i\not\in\tau$.
Then for any $\lambda\in\fh^*$ one has
$$C(\lambda)^{\perp}\subset
\im\Upsilon_{L(\lambda)}\subset \{\mu\in\fh^*|\ (\mu,\alpha)=0\
\text{ for all $\lambda$-minimal $\alpha$}\}.$$
\end{prop}
\begin{proof}
Let $\mu'$ be such that $(\mu',\beta)\not=0$ for all $\beta\in\Delta$.
Using the formula~(\ref{shap}) and the observation after  the formula
we obtain
$$det_{\nu}(\lambda+t\mu+t^2\mu')=a\prod_{(\alpha,m)\in C'(\lambda)}
(t(\mu,\alpha)+t^2(\mu',\alpha))^{(-1)^{p(\alpha)(n+1)P(\nu-n\alpha)}},
$$
where $a$ is an invertible element of the local ring $R$.
By~(\ref{sumformula}) the following sum formula holds
$$\sum_{i=1}^{\infty}\ch M(\lambda)^i=
\sum_{(\alpha,m)\in C'(\lambda)}
 (-1)^{(m-1)p(\alpha)}k(\alpha)\ch M(\lambda-m\alpha),
$$
where $k(\alpha)=1$ if $(\mu,\alpha)\not=0$ and $k(\alpha)=2$
if $(\mu,\alpha)=0$. Note that $(-1)^{(m-1)p(\alpha)}=-1$ forces
$\alpha$ to be odd isotropic.

To verify the second inclusion, assume that $\alpha$ is
$\lambda$-minimal and $(\mu,\alpha)\not=0$. The above formula implies
$$\sum_{i=1}^{\infty} [M(\lambda)^i:L(\lambda-m_{\alpha}\alpha)]=1,$$
where $m_{\alpha}:=2\langle \lambda+\rho,\alpha^{\vee}\rangle$
if $\alpha$ is not isotropic and $m_{\alpha}:=1$ otherwise. Therefore
$[M(\lambda)^1:L(\lambda-m_{\alpha}\alpha)]=1, \ \
[M(\lambda)^2:L(\lambda-m_{\alpha}\alpha)]=0$.
Hence, by~\Prop{corally}, $\mu\not\in\im\Upsilon_{L(\lambda)}$ as required.

Now take $\mu$ such that $(\mu,\alpha)=0\ \text{ for all }
\alpha\in C(\lambda)$. Retain notation of~\ref{KMVirSchw}. Set
$$M:=M(\lambda),\ \ \tilde{M}(s):=M_R(\lambda+st\mu+t^2\mu')\ \text{ for }
s\in\mathbb{R},$$
and identify $\tilde{M}(s)/t\tilde{M}(s)$ with $M$.
Let $\{\tilde{M}(s)^i\}$ be the Jantzen filtration for $\tilde{M}(s)$ and
let $\{\cF^i_s(M)\}$ be the image of this filtration in $M$. Then
$M(\lambda)^i=\cF^i_1(M)$. In the light of~\Prop{corally}
$$\mu\in\im\Upsilon_{L(\lambda)}\ \Longleftrightarrow\ \cF^1_1(M)=\cF^2_1(M).$$
Clearly, $\cF^1_0(M)=\cF^2_0(M)$; below we will deduce from this that
$\cF^1_1(M)=\cF^2_1(M)$.

Fix $\nu\in Q_+$ and set $d^j(s):=\dim \cF^j_s(M)_{\lambda-\nu}$.
Let us show that $d^j(s)$ are constant functions of $s$.
Indeed, since $(\alpha,s\mu)=0$ for all $\alpha\in C(\lambda)$,
the above sum formula  gives
$$\sum_{j=1}^{\infty}\ch \cF^j_s(M)=\sum_{(\alpha,m)\in C'(\lambda)}
(-1)^{(m-1)p(\alpha)}2\ch M(\lambda-m\alpha).
$$
Since the right-hand side does not depend on $s$,
the sum $\sum_{j=1}^{\infty} d^j(s)$ does not depend on $s$.
Denote by $S_{\nu}(s)$ the Shapovalov matrix of $\cU(\fn_-)_{-\nu}$
evaluated at $\lambda+st\mu+t^2\mu'$ (the entries of $S_{\nu}(s)$
lie in $R[s]$).
Then for each value $s_0\in\mathbb{R}$ one has:
$d^j(s_0)=\dim \cF^j_{s_0}(M)_{\lambda-\nu}$ is equal
the corank of $S_{\nu}(s_0)$ modulo $t^j$. Set $m_j:=\max_s d^j(s)$.
Since the corank of matrix, depending on one real parameter,
takes its maximal value for almost all values of the parameter, one has
$d^j(s)=m_j$ for each $s\in \mathbb{R}\setminus J_j$, where
$J_j\subset \mathbb{R}$ is a finite set. Thus for some $s\in\mathbb{R}$
one has $d^j(s)=m_j$ for all $j$. Since
 $\sum_{j=1}^{\infty} d^j(s)<\infty$ is a constant,
 $\sum_{j=1}^{\infty} d^j(s)=\sum_{j=1}^{\infty} m_j$ so
$d^j(s)=m_j$ for any $s$.  Hence $d^j(s)$ does not depend on $s$.

Now $\cF^1_0(M)=\cF^2_0(M)$ gives $d^1(s)=d^2(s)$, hence
$\dim \cF^1_s(M)_{\lambda-\nu}=\dim \cF^2_s(M)_{\lambda-\nu}$ for all $s$.
Therefore $\cF^1_1(M)=\cF^2_1(M)$ and so
$\mu\in\im\Upsilon_{L(\lambda)}$; this establishes the first inclusion.
\end{proof}

\subsubsection{}
\begin{rem}{rem2}
Recall that $\fh'=\fh\cap [\fg,\fg]$ is spanned by $\Pi^{\vee}$.
From~\Rem{remprop1} we see that the exact sequence
 $0\to L(\lambda)\to N\to L(\lambda)\to 0$ splits over
$[\fg,\fg]$ iff $\Upsilon_{L(\Lambda)}$ maps this sequence to
$\mu\in \Delta^{\perp}$.
By~\Prop{propalmin}, $\Delta^{\perp}\subset \im\Upsilon_{L(\lambda)}$;
by above, the equality means that any self-extension of $L(\lambda)$
splits over $[\fg,\fg]$.
\end{rem}

\subsubsection{}
\begin{exa}{exsl3}
Let $\fg=\fsl_3$. Take $\lambda$ such that
$m:=\langle\lambda+\rho,\alpha^{\vee}\rangle\in\mathbb{Z}_{>0},
\langle\lambda+\rho,\beta^{\vee}\rangle=0$.
In this case, $C'(\lambda)=\{(\alpha, m),\ (\alpha+\beta,m)\}$
and $\alpha$ is the only $\lambda$-minimal root. ~\Prop{propalmin}
gives $0\subset \Upsilon_{L(\Lambda)}\subset \{\alpha\}^{\perp}$.
One can deduce from~\Prop{corally} that,  in, fact one has
$\Upsilon_{L(\Lambda)}=\{\alpha\}^{\perp}$, see~\Exa{exsl23}.
\end{exa}

\subsection{Generalized Verma modules}\label{genVerma}
Let $\fg=\fg(A,\tau)$, where $A$ is  a symmetrizable matrix,
such that $a_{ii}\not=0$ for $i\not\in\tau$.
This assumption means that every triple of Chevalley generators
$e_i,f_i, \alpha_i^{\vee}$ with $i\not\in\tau$
(resp. $i\in\tau, a_{ii}\not=0$; resp. $i\in\tau, a_{ii}=0$)
spans $\fsl_2$ (resp. $\osp(1,2)$; resp. $\fsl(1,1)$).
Retain notation of~\ref{assmA}.
For $J\subset \Pi$ let $Q_J$ be the $\mathbb{Z}$-span of $J$. Set
$$\fn_{\pm,J}:=\displaystyle\oplus_{\alpha\in Q_J}\fg_{\pm\alpha}, \ \
\fh_J:=\sum_{\alpha\in J}\mathbb{C}\alpha^{\vee}.$$
Note that $\fh_J=[\fn_{+},\fn_{-,J}]\cap\fh=[\fn_{+,J},\fn_{-,J}]\cap\fh$.

Fix $\lambda\in\fh^*$ such that $\langle \lambda,\fh_J\rangle=0$.
Let $\mathbb{C}_{\lambda}$ be a one-dimensional $\fh$-module corresponding
to $\lambda$, endowed by the trivial action of $\fn_++\fh_J+\fn_{-,J}$.
The {\em generalized Verma module } $M_J(\lambda)$ is
$$M_J(\lambda):=\Ind_{\fn_++\fh+\fn_{-,J}}^{\fg}\mathbb{C}_{\lambda}.$$

Retain notation of~\ref{CC'} and observe that $J\subset C(\lambda)$
and that the elements of $J$ are $\lambda$-minimal. Therefore,
$\Delta^{\perp}\subset\im\Upsilon_{L(\lambda)}\subset J^{\perp}$,
by~\Prop{propalmin}.
Fix $\mu,\mu'\in J^{\perp}$.
The Jantzen-type filtrations on the generalized Verma module $M_J(\lambda)$
can be defined as in~\ref{Jan}. Namely, we define
the generalized Verma module  $N:=M_{J,R}(\lambda+t\mu+t^2\mu')$,
denote by  $v_0$ its canonical highest-weight generator, and set
$N^i:=\{v\in N|\
\cU(\fn_+)v\cap Rv_0\subset Rt^iv_0\}$. Then $\{N^i\}$ is a decreasing
filtration of $N$ and its image in
$M_J(\lambda)=N/tN$ is the Jantzen-type filtration $\{M_J(\lambda)^i\}$.
Repeating the reasonings of Subsection~\ref{mapUps} we obtain that
$M_J^1$ is the
maximal proper submodule of $M_J(\lambda)$, that is
$L(\lambda)=M_J/M_J^1$, and that
$$\mu\in\im\Upsilon_{L(\lambda)}\ \Longleftrightarrow\
M_J(\lambda)^1=M_J(\lambda)^2,$$
where the Jantzen-type filtration $\{M_J(\lambda)^i\}$ is induced from
$M_{J,R}(\lambda+t\mu+t^2\mu')$ (for any $\mu'\in J^{\perp}$).
Furthermore, \Prop{propalmin} gives
\begin{equation}\label{eqgen1}
\Delta^{\perp}\subset\im\Upsilon_{L(\lambda)}\subset J^{\perp}.
\end{equation}

Similar facts hold for other generalized Verma modules
(that is the modules induced from a finite-dimensional representation
of $\fn_++\fh_J+\fn_{-,J}$).

\section{Complete reducibility for a symmetrizable Kac-Moody algebra}

In this section we prove Theorems~\ref{thm01},\ref{thm03},\ref{thm04}
(see~\ref{propLie1},~\ref{corP+},~\ref{thmrestin} respectively).

In this section $\fg=\fg(A)$ is a symmetrizable Kac-Moody algebra~\cite{Kbook}.
In this case $\tau=\emptyset$ and $a_{ii}=2, a_{ij}\in\mathbb{Z}_{\leq 0}$
for each $i,j\in I$. We can (and will) normalize
$(-,-)$ in such a way that $(\alpha_i,\alpha_i)\in\mathbb{Z}_{>0}$
for all $i$, so that $(\rho,\alpha_i)>0$.

The results of this section
extend to symmetrizable Kac-Moody superalgebras $\fg(A,\tau)$,
such that $A$ has the same properties as in the non-super case,
and, in addition, $a_{ij}\in2\mathbb{Z}$ if $i\in\tau$. This includes
$\osp(1,n)$ and ${\osp}(1,n)\hat{ }$.

\subsection{}
We retain notation of~\ref{extlala}.
Recall that a $[\fg,\fg]$-module $N$ is called {\em integrable}
if, for each $i\in I$,
the Chevalley generators $e_i, f_i$ act locally nilpotently on $N$.
This condition implies that $N$ is $\fh'$-diagonalizable.

Recall~(\cite{Kbook}, Chapter 10) that
$L(\lambda)$ is integrable iff $\lambda\in P_+$, where
$$P_+:=\{\lambda\in\fh^*|\ \langle
\lambda,\alpha^{\vee}\rangle\in\mathbb{Z}_{\geq 0}\
\text{ for each simple root } \alpha\}.$$

\subsection{}\label{P+}
\begin{prop}{corP+}
Let $N$ be  a restricted $\fg$-module such that $\fh$
and the Casimir operator act locally finitely and
such that any irreducible subquotient
of $N$ is of the form $L(\lambda)$ with $\lambda\in P_+$. Then
$N$ is completely reducible over $[\fg,\fg]$.
\end{prop}
\begin{proof}
In view of~\Lem{lemcH}, it is enough to show that

(i) $\Ext^1_{[\fg,\fg]}(L(\lambda),L(\lambda'))=0$
for $\lambda,\lambda'\in P_+$,

(ii) $N$ contains an irreducible submodule.

Recall that the Casimir operator acts on $L(\mu)$
by the scalar $c_{\mu}:=(\mu+\rho,\mu+\rho)-(\rho,\rho)$~(\cite{Kbook},
Chapter 2).
For any $\lambda,\lambda'\in P_+$ such that $\lambda'>\lambda$ one has
$c_{\lambda'}-c_{\lambda}=(\lambda'-\lambda,\lambda'+\lambda+2\rho)>0$.
Thus $\Ext^1(L(\lambda),L(\lambda'))=0$ for $\lambda,\lambda'\in P_+,
\lambda\not=\lambda'$. Moreover, $\Ext^1_{[\fg,\fg]}(L(\lambda),L(\lambda))=0$
for $\lambda\in P_+$, by~\Prop{prop1} (ii). Hence (i) holds.

For (ii)  we may (and will) assume that $N$ is indecomposable.
Since the Casimir operator of $\fg$ acts on $N$ locally finitely,
$N$ admits a decomposition into generalized eigenspaces with respect
to this action. Since $N$ is indecomposable, the Casimir operator
has a unique eigenvalue on $N$. In particular, the Casimir
operator acts on all irreducible subquotients
of $N$ by the same scalar.
Let $\supp(N)$ be the set of irreducible subquotients of $N$.
By the assumption, $L(\lambda')\in\supp(N)$
forces $\lambda'\in P_+$.  Take
$\lambda$ such that $L(\lambda)\in\supp(N)$. By above,
 $L(\lambda')\not\in\supp(N)$ if
$\lambda'>\lambda$ or $\lambda'<\lambda$. Thus
by~(\ref{subqu}), $\Omega(N)\cap (\lambda+Q_+)=\{\lambda\}$
and so  $\fn_+N_{\lambda}=0$. Let $v\in N_{\lambda}$ be an
eigenvector of $\fh$. The submodule generated by $v$
is a quotient of $M(\lambda)$ and so, by above, it is isomorphic
to $L(\lambda)$. Hence $N$ contains an irreducible submodule.
The assertion follows.
\end{proof}

\subsubsection{}
\begin{thm}{thmrestin}
Let $\fg$ be a symmetrizable Kac-Moody algebra. Any
 restricted integrable $[\fg,\fg]$-module
is completely reducible and its irreducible submodules are of the form
$L(\lambda)$ with $\lambda\in P_+$.
\end{thm}
\begin{proof}
By above, $\Ext^1_{[\fg,\fg]}(L(\lambda),L(\lambda'))=0$
for $\lambda,\lambda'\in P_+$. In the light of~\Lem{lemcH},
it is enough to show that
 each restricted integrable  module contains an irreducible submodule
of the form $L(\lambda)$ with $\lambda\in P_+$.

Let $\Pi=\{\alpha_1,\ldots,\alpha_l\}$ be the set of simple roots.
Write $\beta\in\Delta_+$
as $\beta=\sum_{i=1}^l m_i\alpha_i$ and set $\htt\beta=\sum m_i$.
For $m\in\mathbb{Z}_{\geq 0}$ set $\fn_{>m}:=
\sum_{\beta:\htt\beta>m} \fn_{\beta}$
and note that $\fn_{>m}$ is an ideal of $\fn$.
Take $v\in N$. Since $N$ is restricted, $\fn_{>m}v=0$
for some $m\in\mathbb{Z}_{\geq 0}$.
Set $N':=\cU(\fn)v$. By above, $\fn_{>m}N'=0$
so $N'$ is a module over $\fm:=\fn/\fn_{>m}$ which is
a finite-dimensional nilpotent Lie algebra.

Recall the following proposition~\cite{Kbook}, 3.8. If $\fp$
is a Lie algebra and $M$ is an $\fp$-module,
satisfying the following condition: $\fp$ is generated by
ad-locally finite elements which are also locally finite on $M$,
then  $\fp$ is a span of elements which act locally finitely
on $M$.  Since $N$ is integrable, $\{e_{\alpha_i}\}_{i=1}^l$
act locally finitely on $N'$ and thus
there exists a basis $\{u_j\}_{j=1}^{s}$
of $\fm$, where each $u_j$ acts locally finitely on $N'$.
Since $N'$ is a cyclic $\fm$-module (generated by $v$),
this implies that it is finite-dimensional. Then,
by the Lie Theorem, $\fm$ has an eigenvector $v'\in N'$.
Since $\fm$ is generated by $e_i, i\in I$ and these elements
act locally nilpotently on $N'$, one has $\fm v'=0$.
Thus $\fn v'=0$ that is $N^{\fn}\not=0$.
Recall that $\fh'$ acts diagonally
on $N$ so $N^{\fn}$ contains an eigenvector $v''$ for
$\fh'$. The vector $v''$ generates a submodule
which is a quotient of a Verma module. Since this
submodule is integrable, it is isomorphic to $L(\lambda)$
with $\lambda\in P_+$.
\end{proof}

\subsection{The set $\Delta(\lambda)$}
Retain notation of~\ref{symmKm}. Let $\Pi$ be the set of simple roots and
let $W$ be the Weyl group of $\fg$.
Recall that a root $\alpha$ is {\em real} if
$W\alpha\cap \Pi\not=\emptyset$;
a root is {\em imaginary} if it is not real.
One has: $\alpha$ is real iff $(\alpha,\alpha)>0$.

\subsubsection{}\begin{defn}{Deltal}
For $\lambda\in\fh^*$ let $\Delta(\lambda)$ to be the set  of
real roots $\alpha$ such that
$m_{\alpha}:=\langle\lambda+\rho, \alpha^{\vee}\rangle\in\mathbb{Z}$,
and   set
$\Delta(\lambda)^{\vee}:=\{\alpha^{\vee}|\ \alpha\in \Delta(\lambda)\}$.
\end{defn}

Notice that $\Delta(\lambda)^{\vee}\subset \fh'$. Recall that $\lambda$
is called rational iff $\Delta(\lambda)^{\vee}$ spans $\fh'$, which
is equivalent to $\mathbb{C}\Delta(\lambda)=\mathbb{C}\Delta$.

Set $\Delta_+(\lambda):=
\{\alpha\in\Delta(\lambda)|\ \alpha\in\Delta_+\}$.

\subsubsection{}\label{Wlambda}
Let $W(\lambda)$ be the subgroup of the Weyl group $W$ of $\fg$ generated
by the reflections $s_{\alpha}$ with $\alpha^{\vee}\in \Delta(\lambda)$.
Introduce the dot action by the usual formula
$w.\lambda:=w(\lambda+\rho)-\rho$ for any $w\in W$.
Notice that $s_{\alpha}(\Delta(\lambda))=\Delta(\lambda)$
for any $\alpha\in \Delta(\lambda)$. As a result,
$\Delta(\lambda)$ is $W(\lambda)$-invariant and
$\Delta(w.\lambda)=\Delta(\lambda)$ for any $w\in W(\lambda)$.
Thus $\fh^*$ is a disjoint union of $W(\lambda)$-orbits (with respect
to the dot action) and $\Delta(\lambda)$ is the same for
every $\lambda$ in a given orbit.

Recall that a weight $\lambda\in\fh^*$ is regular if
$(\lambda,\alpha)\not=0$ for all $\alpha\in\Delta^{re}$.
We call $\lambda\in\fh^*$ {\em shifted-regular} if
$\lambda+\rho$ is regular.
Observe that if $\lambda$ is shifted-regular, then $w.\lambda$
is shifted-regular for any $w\in W$.

\subsubsection{}
\begin{rem}{regmax}
For $\alpha\in \Delta_+(\lambda)$ one has $s_{\alpha}.\lambda>\lambda$
iff $m_{\alpha}=\langle\lambda+\rho, \alpha^{\vee}\rangle<0$.
In particular, a shifted-regular weight $\lambda$ is  maximal in its
 $W(\lambda).$-orbit iff $m_{\alpha}>0$ for all
$\alpha\in \Delta_+(\lambda)$. If $\lambda$ is  maximal in its
 $W(\lambda).$-orbit and is shifted-regular, then
$\Stab_{W(\lambda).}\lambda=\{\id\}$.
\end{rem}

\subsubsection{}\label{Pilambda}
Recall that a non-empty subset $\Delta'$ of a root system is called a
{\em root subsystem} if $s_{\alpha}\beta\in \Delta'$ for any
$\alpha,\beta\in \Delta'$. One readily sees that $\Delta(\lambda)$ is
a root subsystem of $\Delta$.  Let $\Pi(\lambda)$ be the set
of indecomposable elements of $\Delta_+(\lambda)$:
$\alpha\in \Pi(\lambda)$ iff
$\alpha\not\in\sum_{\beta\in \Delta(\lambda)_+\setminus\{\alpha\}}
\mathbb{Z}_{\geq 0}\beta$.
By~\cite{MP}, 5.7  the following properties hold:

($\Pi 1$) $W(\lambda)$ is generated by $s_{\alpha}$ with
$\alpha\in \Pi(\lambda)$ and $\Delta(\lambda)=W(\lambda)
\bigl(\Pi(\lambda)\bigr)$;

($\Pi 2$)
$\Delta_+(\lambda)\subset \sum_{\alpha\in \Pi(\lambda)}
\mathbb{Z}_{\geq 0}\alpha$;

($\Pi 3$) for any $\alpha\in \Pi(\lambda)$ the set
$\Delta(\lambda)_+\setminus\{\alpha\}$ is invariant
under the reflection $s_{\alpha}$.

Note that the elements of $\Pi(\lambda)$ can be linearly dependent.

\subsubsection{}\label{noncrit}
We recall that $\lambda\in\fh^*$ is {\em non-critical} if
for any positive imaginary root $\alpha$ one has
$2(\lambda+\rho, \alpha)\not\in\mathbb{Z}_{>0}(\alpha,\alpha)$.
Remark that, for a non-critical weight $\lambda$, the orbit
$W.\lambda$ consists of non-critical  weights, since the set
of positive imaginary roots is $W$-invariant (\cite{Kbook}, Chapter 5).

By~\cite{KK}, Thm. 2 if $\lambda$ is non-critical then all irreducible
subquotients of $M(\lambda)$ are of the form $L(w.\lambda)$ for
$w\in W(\lambda)$. In particular, a non-critical weight $\lambda$
is weakly admissible (see Subsection~\ref{admis})
iff $\lambda$ is maximal in its $W(\lambda).$-orbit, or, equivalently, iff
$\langle\lambda+\rho,\beta^{\vee}\rangle\in\mathbb{Z}_{\geq 0}$
for all $\beta\in \Pi(\lambda)$.

\subsection{}
\begin{thm}{propLie1}
Let $\fg$ be a symmetrizable Kac-Moody algebra.
Let $\lambda\in\fh^*$ be a non-critical shifted-regular weight, maximal in its
$W(\lambda).$-orbit. Then $\Upsilon_{L(\lambda)}$ induces a bijection
$$\Ext^1(L(\lambda),L(\lambda))\iso \Delta(\lambda)^{\perp}.$$
\end{thm}

\subsubsection{}
\begin{cor}{corLie0}
(i) If  $\lambda$ is admissible, then $\lambda$ is rational.

(ii) A non-critical shifted-regular weight is
admissible iff it is rational and maximal in its $W(\lambda).$-orbit.
\end{cor}

\subsubsection{Proof of~\Thm{propLie1}}
Retain notation of~\ref{defmi}. Since $\lambda$ is non-critical and
maximal in its $W(\lambda).$-orbit, \Rem{regmax} implies that
$C(\lambda)=\Delta_+(\lambda)$, so, by
($\Pi 2$), $C(\lambda)$ and $\Pi(\lambda)$ span
the same subspace of $\fh^*$.
In the light of~\Prop{propalmin}, it is enough to verify that
 each $\alpha\in\Pi(\lambda)$ is $\lambda$-minimal in
the sense of~\Defn{defmi}. This follows from the following lemma.

\subsubsection{}\begin{lem}{lemmik}
Let $\lambda\in\fh^*$ be non-critical. Each $\alpha\in\Pi(\lambda)$ such that
$\langle \lambda+\rho,\alpha^{\vee}\rangle>0$ is $\lambda$-minimal in
the sense of~\Defn{defmi}.
 \end{lem}
\begin{proof}
For $\beta \in \Delta_+(\lambda)$ set
$m_{\beta}:=\langle \lambda+\rho,\beta^{\vee}\rangle$.
By~\Rem{regmax},
$C(\lambda)=\{\beta\in \Delta_+(\lambda)|\ m_{\beta}>0\}$.
It is enough to verify that for any
$\beta\in\Delta_+(\lambda)\setminus\{\alpha\}$ such that $m_{\beta}>0$ one has
$[M(s_{\beta}.\lambda):L(s_{\alpha}.\lambda)]=0$.
Take $\beta\in\Delta_+(\lambda)\setminus\{\alpha\}$ such that $m_{\beta}>0$
and assume that $[M(s_{\beta}.\lambda):L( s_{\alpha}.\lambda)]>0$.
By~\cite{KK}, Thm. 2 we have
a chain $\nu_1=s_{\beta}.\lambda,\ \nu_2,\ldots, \nu_n=s_{\alpha}.\lambda$
such that $\nu_{i+1}=s_{\beta_i}.\nu_i=\nu_i-k_i\beta_i$, where
$\beta_i\in\Delta_+(\nu_i)$ and $k_i>0$. Since $\nu_i,\lambda$ lie
in the same $W(\lambda)$-orbit, one has $\Delta_+(\nu_i)=\Delta_+(\lambda)$.
Therefore
$$m_{\alpha}\alpha-m_{\beta}\beta=s_{\beta}.\lambda-s_{\alpha}.\lambda=
\sum_i k_i\beta_i=\sum_{\gamma\in \Pi(\lambda)} k_{\gamma}\gamma\ \
\text{ for some }k_{\gamma}\geq 0.
$$
Writing $\beta=\sum_{\gamma\in \Pi(\lambda)} n_{\gamma}\gamma$,
we obtain $(m_{\alpha}-k_{\alpha}-m_{\beta}n_{\alpha})\alpha=
\sum_{\gamma\in\Pi(\lambda)\setminus\{\alpha\}}
(k_{\gamma}+m_{\beta}n_{\gamma})\gamma$. Since $\alpha\in \Pi(\lambda)$,
this implies $k_{\gamma}=n_{\gamma}=0$ for $\gamma\not=\alpha$,
so $\beta$ is proportional to $\alpha$, which is impossible, since
$\alpha$ is a real root.
\end{proof}

\subsubsection{Proof of~\Cor{corLie0}}
\label{pfcorLie1}
Since $C(\lambda)\subset \Delta(\lambda)$
for non-critical $\lambda$,~\Prop{propalmin} implies that
$\Delta(\lambda)^{\perp}\subset \im \Upsilon_{L(\lambda)}$.
If $\lambda$ is admissible, then $\im \Upsilon_{L(\lambda)}=\Delta^{\perp}$
so $\Delta(\lambda)^{\perp}\subset\Delta^{\perp}$ that is
$\mathbb{C}\Delta(\lambda)=\mathbb{C}\Delta$; this gives (i).
Now (ii) follows from (i) and~\Thm{propLie1}.
\qed

\subsubsection{Example}\label{exasl2}
For the affine Lie algebra $\hat{{\fsl}}_2$
the weight $\lambda$ is $k$-admissible
for $k:=\langle \lambda,K\rangle$ iff $k\not=-2$,
$\lambda$ is maximal in its $W(\lambda).$-orbit and
$W(\lambda).\lambda\not=\lambda$. Indeed, $W(\lambda).\lambda\not=\lambda$
means that $M(\lambda)$ is not irreducible. If $\lambda$ is $k$-admissible,
then $k\not=-2$, $\lambda$ is maximal in its $W(\lambda).$-orbit and
$M(\lambda)$ is not irreducible. Assume that
$k\not=-2$, $\lambda$ is maximal in its $W(\lambda).$-orbit and
$M(\lambda)$ is not irreducible.
Combining~\Prop{propalmin} and~\Lem{lemmik},
we conclude that $\Upsilon_{L(\lambda)}\subset \alpha^{\perp}$
for some real root $\alpha$. Then $\Upsilon_{L(\lambda)}\cap \{\mu|\
\langle \mu,K\rangle=0\}\subset \{\mu|\ (\mu,\alpha)=(\mu,\delta)=0\}=
\mathbb{C}\delta$. Hence $\lambda$ is $k$-admissible.

\subsection{KW-admissible modules}\label{admKM}
We call $\lambda\in\fh^*$  {\em KW-admissible } if it is a non-critical
rational shifted-regular weakly admissible weight; this means that

(A0) for any positive imaginary root $\alpha$ one has
$2(\lambda+\rho, \alpha)\not\in\mathbb{Z}_{>0}(\alpha,\alpha)$
(i.e. $\lambda$ is  non-critical);

(A1) for any $\alpha\in \Delta_+$ one has
$\langle\lambda+\rho, \alpha^{\vee}\rangle\not\in\mathbb{Z}_{\leq 0}$
(i.e. $\lambda$ is dominant), or,
equivalently: $\langle\lambda+\rho, \alpha^{\vee}\rangle\in\mathbb{Z}_{>0}$
for each $\alpha\in\Pi(\lambda)$;

(A2) $\mathbb{C}\Delta(\lambda)^{\vee}=\fh'$ (i.e. $\lambda$ is rational).

The set of KW-admissible weights for an affine Lie algebra
$\fg$ was described in~\cite{KWmod2}.
From~\Cor{corLie0} one obtains

\subsubsection{}
\begin{cor}{coradm1}
The set of KW-admissible weights coincides with the set of
shifted-regular admissible weights.
\end{cor}

\subsubsection{}
\begin{rem}{remadm}
We will use the following fact: if $\lambda$ is a shifted-regular admissible
weight and $\alpha$ is a simple root such that $\alpha\not\in\Delta(\lambda)$,
then $s_{\alpha}.\lambda$ is a shifted-regular admissible
weight. This follows from the above corollary and the equality
$s_{\alpha}(\Delta_+\setminus\{\alpha\})=\Delta_+\setminus\{\alpha\}$.
\end{rem}

\subsubsection{Example: $\hat{\fsl}_2, \fsl_3$}\label{exsl23}
In these cases any admissible weight is  shifted-regular (hence KW-admissible);
this follows from the fact that if $\lambda$ is not shifted-regular
 (and non-critical for $\hat{\fsl}_2$), then
the maximal proper submodule of $M(\lambda)$ is either zero or isomorphic
to a Verma module. Indeed, assume that $\lambda\in Adm$ is  not
shifted-regular.
It is easy to see that in this case $\Pi(\lambda)=\{\alpha,\beta\}$, where
$m:=\langle\lambda+\rho,\alpha^{\vee}\rangle\in\mathbb{Z}_{>0},
\langle\lambda+\rho,\beta^{\vee}\rangle=0$.
The Shapovalov determinant $\det S_{m\alpha}$ evaluated at
$\lambda+t\mu+t^2\mu'$ is proportional to $t(\mu,\alpha)+t^2(\mu',\alpha)$.
If $(\mu,\alpha)=0$ the sum formula gives
$\sum_{i=1}^{\infty}\dim M(\lambda)^i_{\lambda-m\alpha}=2$.
Since $M(\lambda)^1=M(\lambda-m\alpha)$,
one has $M(\lambda)^2=M(\lambda-m\alpha)$. Thus, by~\Prop{corally},
$\mu\in\im\Upsilon_{L(\lambda)}$ if $(\mu,\alpha)=0$.
Hence $L(\lambda)$ admits self-extensions over $[\fg,\fg]$ and
thus  $\lambda\not\in Adm$, a contradiction.

\subsubsection{Example: $\fsl_4, \hat{\fsl}_3$}
In these cases there are admissible weights, which are not shifted-regular;
for instance, $\lambda$ satisfying
$\langle\lambda+\rho,\alpha_1^{\vee}\rangle=
\langle\lambda+\rho,\alpha_3^{\vee}\rangle
=1,\ \langle\lambda+\rho,\alpha_2^{\vee}\rangle=0$, where
$\alpha_1,\alpha_2,\alpha_3$ are the simple roots.
In order to verify that $L(\lambda)$ does not have self-extensions over
$[\fg,\fg]$, note that, by~\Prop{propalmin},
$(\mu,\alpha_1)=(\mu,\alpha_3)$ for
$\mu\in \im\Upsilon_{L(\lambda)}$. The module
$M(\lambda)$ contains a subsingular vector $v$
of weight $\lambda-(\alpha_1+2\alpha_2+\alpha_3)$. It is not hard to show
$v\not\in M(\lambda)^2$ if $(\mu,\alpha_2)\not=0$,
hence $\mu\in \mathbb{C}\delta$ as required.

For $\hat{\fsl}_3$ such $\lambda$ is a particular case
of admissible weights described in~\Rem{remaffv}.

\section{Finiteness of $\Adm_k$ for affine Lie algebras}
In this section $\fg$ is an affine Lie algebra. (However
all results can be extended to ${\osp}(1,2l)\hat{ }$).
Recall that $\fg=[\fg,\fg]\oplus \mathbb{C}D$.
We denote by $K$ ($\in [\fg,\fg]$) the canonical central element of $\fg$.
One says that a $\fg$ (or $[\fg,\fg]$)-module
$N$ has {\em level} $k\in\mathbb{C}$ if
$K|_N=k\cdot\id$, and that $\lambda\in\fh^*$  has level $k$
if $\langle \lambda,K\rangle=k$.

Recall that a simple $[\fg,\fg]$-module
$L(\lambda)$ (and its highest weight  $\lambda$) is
$k$-admissible if it is weakly admissible of level $k$, and
each self-extension $N$ of $L(\lambda)$ satisfying
$K|_N=k\cdot\id$ splits over $[\fg,\fg]$.

\subsection{Main results}
In this section we deduce~\Thm{thm05} from~\Thm{thm01}.
A key fact  is that for rational $k$
the category $\Adm_k$ has finitely many irreducibles, see~\Cor{corkh}.
The semisimplicity follows from this fact and~\Lem{lemcH}.
Indeed, extend the action of  $[\fg,\fg]$ on modules in $\Adm_k$
to that of $\fg$ by letting $D=-L_0$, where $L_0$ is a Virasoro operator,
see~\ref{viraf}. The category $\Adm_k$, viewed as a category
of $\fg$-modules, satisfies the assumptions of~\Lem{lemcH}
and so it is semisimple.

The fact that for rational $k$
three sets: shifted-regular $k$-admissible weights,
shifted-regular admissible weights and KW-admissible weights coincide
is proven in~\Cor{corLie1}. In Subsection~\ref{polyh}
we prove the description of $k$-admissible weights given in
Subsection~\ref{introk}. In Subsection~\ref{vac}
we establish a criterion of $k$-admissibility for vacuum modules.

\subsection{Notation}\label{notaf}
Let $\delta\in\fh^*$ be the minimal imaginary root, i.e.
$\langle\delta,\fh'\rangle=0,\ \langle\delta,D\rangle=1$.
One has $\Delta^{\perp}=\mathbb{C}\delta$.

Let $r$ be the {\em tier number} of $\fg$ ($r=1,2,3$ is
such that $\fg$ is of the type $X_N^{(r)}$).
Let $A$ be the Cartan matrix of $\fg$ and let $r^{\vee}$ be the dual
tier number of $\fg$, i.e. the tier number of the affine Lie algebra
$\fg(A^t,\tau)$. Recall that $\Delta$ (resp., $\Delta^{\vee}$)
is invariant under the shift by $r\delta$ (resp., $r^{\vee}K$)

\subsubsection{}
We normalize the invariant form $(-,-)$ by the condition that
$(\alpha,\alpha)=2r$, where $r$ is the tier number of $\fg$, and
$\alpha$ is a long simple root (\cite{Kbook}, Chapter 6); note that
$(\beta,\beta)$ is a positive rational number
for any real root $\beta$. One has
$\langle\mu,K\rangle=(\mu,\delta)$ for
all $\mu\in\fh^*$. Recall that $(\delta,\delta)=0$ and that
any imaginary root is an integral multiple
of $\delta$. In particular, $\lambda$ is non-critical iff
$(\lambda+\rho,\delta)\not=0$.

One has $(\rho,\delta)=h^{\vee}$, where
$h^{\vee}$ is the dual Coxeter number.
By above, $\lambda$ is critical iff
the level of $L(\lambda)$ is  $-h^{\vee}$
(the {\em critical level}).

\subsubsection{}\label{viraf}
A restricted $[\fg,\fg]$-module $N$ (see Subsection~\ref{symmKm})
of a non-critical level $k\not=-h^{\vee}$
admits an action of Virasoro algebra
given by Sugawara operators $L_n, n\in\mathbb{Z}$,
see~\cite{Kbook}, Section~12.8. In this case
the Casimir operator $\hat{\Omega}$ takes the form
$\hat{\Omega}=2r(k+h^{\vee})(D+L_0)$, and it acts on $L(\lambda)$ by
$(\lambda+2\rho,\lambda)\id$ (see~\cite{Kbook}, Chapter 2).
Notice that $\hat{\Omega}$ acts by different
scalars on $L(\nu)$ and on $L(\nu+s\delta)$, if $\nu$
is non-critical and $s\not=0$.

\subsubsection{}
Recall that $L(\lambda)\cong L(\lambda')$ as $[\fg,\fg]$-modules
iff $\lambda|_{\fh'}=\lambda'|_{\fh'}$ that is
$\lambda'-\lambda\in\mathbb{C}\delta$. Moreover,
if $\lambda'-\lambda=s\delta$, then $\Delta(\lambda)=\Delta(\lambda')$,
and, taking tensor products by
the one-dimensional module $L(s\delta)$, we obtain isomorphisms
$\Ext^1(L(\lambda),L(\nu))\iso\Ext^1(L(\lambda+s\delta),L(\nu+s\delta))$
for any $\nu\in\fh^*$.
In this subsection we consider non-critical
 weights $\lambda\in(\fh')^*$, and denote
by $L(\lambda)$ the corresponding $[\fg,\fg]$-module. By above,
$\Delta(\lambda), \Upsilon_{L(\lambda)}$
is well defined for $\lambda\in(\fh')^*$.
Notice that $\Ext^1_{[\fg,\fg]}(L(\lambda),L(\nu))\not=0$ implies
that for each $\lambda'\in\fh^*$ satisfying
$\lambda'|_{\fh'}=\lambda$ there exists a unique $\nu'\in\fh^*$ satisfying
$\nu'|_{\fh'}=\nu$ such that $\Ext^1_{\fg}(L(\lambda'),L(\nu'))\not=0$; such
$\nu'$ is determined by the condition
$(\lambda'+2\rho,\lambda')=(\nu'+2\rho,\nu')$.
It is easy to see that
$\Ext^1_{\fg}(L(\lambda'),L(\nu'))=\Ext^1_{[\fg,\fg]}(L(\lambda),L(\nu))$
if $\lambda\not=\nu$.

\subsubsection{}
\begin{rem}{Remkrat}
Let $\lambda\in\fh^*$ be such that
$k:=\langle\lambda,K\rangle\in\mathbb{Q}\setminus\{-h^{\vee}\}$.
Write $k+h^{\vee}=\frac{p}{q}$ with coprime
$p,q\in\mathbb{Z}, p\not=0, q>0$.

Since $\Delta^{\vee}$  is invariant under the shift by $r^{\vee}K$,
$\Delta(\lambda)^{\vee}$ is invariant under the shift by $qr^{\vee}K$.
In particular, $\mathbb{Q}\Delta(\lambda)^{\vee}+\mathbb{Q}K=
\mathbb{Q}\Delta(\lambda)^{\vee}$ so
$\mathbb{Q} \Delta(\lambda)+\mathbb{Q}\delta=\mathbb{Q}\Delta(\lambda)$
if $\lambda$ is non-critical and has rational level.
\end{rem}

\subsection{}\label{43}
Recall that (see the introduction) a weight $\lambda\in(\fh')^*$ of
non-critical level $k$ is weakly admissible iff
\begin{equation}\label{16}
\langle\lambda+\rho,\alpha^{\vee}\rangle\in\mathbb{Z}_{\geq 0}\
\text{ for all }\alpha\in\Delta_+(\lambda).
\end{equation}
Denote the set of such weights by $wAdm_k$.

Retain notation of Subsection~\ref{Pilambda}. By~\Rem{Remkrat},
for $\beta\in\Pi(\lambda)$ one has
$\beta\pm qr\delta\in \Delta(\lambda)$. Moreover,
$\beta+qr\delta\in \Delta(\lambda)_+$, so
$qr\delta-\beta=s_{\beta}(\beta+qr\delta)\in \Delta(\lambda)_+$, by
($\Pi 3$) of~\ref{Pilambda}. Hence $\beta\in\Pi(\lambda)$ forces
$qr\delta-\beta\in \Delta(\lambda)_+$.
If $\lambda$ is weakly admissible, then, by~\Rem{Remkrat},
$(\lambda+\rho,\beta),\ (\lambda+\rho,qr\delta-\beta)\geq 0$,
so $k+h^{\vee}=(\lambda+\rho,\delta)\geq 0$.
Hence there are no weakly admissible weights of level $k$ for
$k+h^{\vee}\in\mathbb{Q}_{<0}$.

\Thm{propLie1} implies the following corollary.

\subsubsection{}
\begin{cor}{corLie1}
(i) If any self-extension of $L(\lambda)$ with a diagonal action of $K$
splits, then
 $\mathbb{C} \Delta(\lambda)+\mathbb{C}\delta=\mathbb{C} \Delta$;
if, in addition, $k\in\mathbb{Q}\setminus\{-h^{\vee}\}$, then
$\lambda$ is rational.

(ii) A shifted-regular weight $\lambda$ is $k$-admissible
iff it is maximal in its $W(\lambda).$-orbit and
$\mathbb{C} \Delta(\lambda)+\mathbb{C}\delta=\mathbb{C} \Delta$.

(iii) For rational $k$ three sets: shifted-regular $k$-admissible weights,
shifted-regular admissible weights and KW-admissible weights coincide.
\end{cor}
\begin{proof}
Recall that $\fh'$ is spanned by $\{\alpha^{\vee}|\ \alpha\in\Pi\}$,
and that $\lambda$ is rational iff
$\mathbb{C}\Delta(\lambda)^{\vee}=\fh'$,
that is $\mathbb{C} \Delta(\lambda)=\mathbb{C} \Delta$.
Observe that $\xi\in\fh^*$ vanishes on $\fh'$
iff $\xi\in\Delta^{\perp}=\mathbb{C}\delta$.

For (i) assume that any self-extension of $L(\lambda)$
with a diagonal action of $K$
splits over $[\fg,\fg]$. In the light of~\Rem{rem2}, this means that
\begin{equation}\label{som}
\im \Upsilon_{L(\lambda)}\cap\{\mu|\ \langle\mu,K\rangle=0\}\subset
\Delta^{\perp}=\mathbb{C}\delta,
\end{equation}
By~\Prop{propalmin}, $\delta\in\im \Upsilon_{L(\lambda)}$.
Thus~(\ref{som}) is equivalent to
$\im \Upsilon_{L(\lambda)}\cap\{\mu|\ \langle\mu,K\rangle=0\}
=\mathbb{C}\delta$, which can be rewritten as
$(\im \Upsilon_{L(\lambda)})^{\perp}+\mathbb{C}\delta=\mathbb{C}\Delta$.
Hence any self-extension of $L(\lambda)$ with a diagonal action of $K$
splits over $[\fg,\fg]$ iff
\begin{equation}\label{som1}
(\im \Upsilon_{L(\lambda)})^{\perp}+\mathbb{C}\delta=\mathbb{C}\Delta.
\end{equation}
By~\Prop{propalmin}, $(\im\Upsilon_{L(\lambda)})^{\perp}\subset
\mathbb{C}\Delta(\lambda)$.
This gives (i).

Recall that $ k$-admissibility of $\lambda$
means that $\langle\lambda,K\rangle=k$, that $\lambda$ is weakly admissible,
that is $\lambda$ is maximal in its $W(\lambda).$-orbit, and that
any self-extension of $L(\lambda)$ with a diagonal action of $K$ splits
over $[\fg,\fg]$. By~\Thm{propLie1}, if $\lambda$ is shifted-regular
and maximal in its $W(\lambda).$-orbit,
then $(\im \Upsilon_{L(\lambda)})^{\perp}=\mathbb{C}\Delta(\lambda)$.
Now~(\ref{som1}) implies (ii) and
(iii) follows from (ii), ~\Rem{Remkrat}, and~\Cor{coradm1}.
\end{proof}

\subsection{Case of rational level}\label{polyh}
By~\Rem{Remkrat}, the set of weakly admissible weight of level $k$ is empty,
if $k+h^{\vee}\in\mathbb{Q}_{<0}$. Take $k+h^{\vee}\in\mathbb{Q}_{>0}$
and write $k+h^{\vee}=p/q$ for coprime $p,q\in\mathbb{Z}_{>0}$. Set
$$X_k:=\{\lambda\in (\fh')^*|\ \lambda\in wAdm_k\ \&\
\mathbb{C} \Delta(\lambda)^{\vee}=\fh'\}.$$
By~\Cor{corLie1}, the set of $k$-admissible weights is a subset of
$X_k$.

Recall that $\Delta(\lambda)^{\vee}_+:=\{\beta^{\vee}\in\Delta^{\vee}_+|\
\langle \lambda+\rho,\beta^{\vee}\rangle\in\mathbb{Z}\}$. Set
$$X_k(\Gamma):=\{\lambda\in wAdm_k|\
\Delta(\lambda)^{\vee}_+=\Gamma\}\ \text{ for }\Gamma\subset\Delta^{\vee}_+,\ \
B_k:=\{\Gamma\subset\Delta_+^{\vee}|\ \mathbb{C} \Gamma=\fh'\ \& \
X(\Gamma)\not=\emptyset\}.$$
One has
$$X_k=\coprod_{\Gamma\in B_k} X_k(\Gamma).$$

By~\Rem{Remkrat}, if $\lambda$ has level $k$, then
$\Delta(\lambda)$ is invariant under the shift by $rq\delta$ ($r\in\{1,2,3\}$).
Since $\Delta$ has finitely many orbits
modulo $\mathbb{Z}rq\delta$,  $\Delta$ has finitely many subsets
which are stable under the shift by $rq\delta$.
Hence there are finitely many possibilities for $\Delta(\lambda)$, so
the set $B_k$ is finite.

For each $\Gamma\in B_k$ define
the polyhedron $P(\Gamma)\subset (\fh')^*$ by
$$\cP^k(\Gamma):
=\{\lambda\in (\fh')^*|\ \langle\lambda,K\rangle=k\ \text{ and }\
\forall \beta^{\vee}\in\Gamma,\
\langle \lambda+\rho,\beta^{\vee}\rangle\in\mathbb{R}_{\geq 0}\}.$$
We will show that each $\cP^k(\Gamma)$ is compact.

We call $\lambda\in  \cP^k(\Gamma)$ {\em integral}
if $\Gamma=\Delta(\lambda)^{\vee}_+$. Clearly, if
$\lambda\in  \cP^k(\Gamma)$ is integral, then
$\langle \lambda+\rho,\beta^{\vee}\rangle\in
\mathbb{Z}_{\geq 0}$ for each $\beta^{\vee}\in\Gamma$. Moreover,
if $\lambda\in  X_k(\Gamma)$, then $\lambda\in  \cP^k(\Gamma)$ and is integral.
Thus $X_k(\Gamma)$ is the set of integral points of the polyhedron
$\cP^k(\Gamma)$. If $\lambda\in X_k(\Gamma)$ is shifted-regular, then
$\langle \lambda+\rho,\beta^{\vee}\rangle\in
\mathbb{Z}_{>0}$ for each $\beta^{\vee}\in\Gamma$.
Thus the set of KW-admissible weights of a level $k$
is a finite union of the interior integral points  of the polyhedra
$\cP^k(\Gamma)$, $\Gamma\in B_k$. The set of $k$-admissible weights
lies in the  union of integral points of these polyhedra
and contains all their interior integral points (see~\Cor{corLie1}).

For each $\Gamma\in B_k$ the set $\Gamma\coprod (-\Gamma)\subset\Delta^{\vee}$
is a root  subsystem. By~\Rem{Remkrat}, $\Gamma\coprod (-\Gamma)$
is invariant under the shift by $qr^{\vee}K$.
We denote by $\Pi(\Gamma)$ the corresponding set of simple roots, i.e.
$\Gamma=\Delta(\lambda)^{\vee}_+,
\Pi(\Gamma)=\Pi(\lambda)$ for $\lambda\in X_k(\Gamma)$,
see~\ref{Pilambda} for the notation. One has
$\cP^k(\Gamma)=\{\lambda\in (\fh')^*|\ \langle \lambda,K\rangle=k\ \&\
\forall \beta^{\vee}\in\Pi(\Gamma)\
\langle \lambda+\rho,\beta^{\vee}\rangle\in\mathbb{R}_{\geq 0}\}$.

Let us show that $\cP^k(\Gamma)$ is compact for any $\Gamma\in B_k$.
Indeed, by Subsection~\ref{43},
$qr^{\vee}K-\beta^{\vee}\in\Gamma$ if $\beta\in \Pi(\Gamma)$
(because $\Gamma=\Delta(\lambda)^{\vee}, \Pi(\Gamma)=\Pi(\lambda)$
for $\lambda\in X_k(\Gamma)$). Then for each $\beta\in \Pi(\Gamma)$
and any $\lambda'\in \cP^k(\Gamma)$ one has
$0\leq\langle \lambda'+\rho,qr^{\vee}K-\beta^{\vee}\rangle
=-\langle \lambda'+\rho,\beta^{\vee}\rangle
+pr^{\vee}$, that is
$\langle \lambda'+\rho,\beta^{\vee}\rangle\in [0; pr^{\vee}]$.
Since $\mathbb{C}\Pi(\Gamma)=\fh'$, $\cP^k(\Gamma)$ is compact.
For any $\lambda\in X_k(\Gamma)$ and for each $\beta\in \Pi(\Gamma)$ the value
$\langle \lambda+\rho,\beta^{\vee}\rangle$ is an integer
in the interval $[0; pr^{\vee}]$. Thus  $X_k(\Gamma)$ is a finite set.

\subsubsection{}
\begin{cor}{corkh}
For $k+h^{\vee}\in \mathbb{Q}_{\leq 0}$, the category $\Adm_k$ is empty.
For $k+h^{\vee}\in \mathbb{Q}_{>0}$,
the category $\Adm_k$ contains finitely many
irreducibles.
\end{cor}

\subsubsection{}
Consider the case, when $\Gamma\in B_k$ corresponds
to an irreducible root subsystem ($\Pi(\Gamma)$ is the set of simple roots
of an irreducible root system). We claim that the interior integral
points of each face of the polyhedron corresponding to
$X_k(\Gamma)$ are $k$-admissible weights (by interior points of a face we
mean the ones that  do not lie on the faces of codimension $2$).

Indeed, the polyhedron corresponding to
$X_k(\Gamma)$ is given by $\langle \lambda+\rho,\beta^{\vee}\rangle\in
\mathbb{Z}_{\geq 0}$ for each $\beta\in\Pi(\Gamma)$.
The faces are parameterized
by the elements of $\Pi(\Gamma)$. Consider the face corresponding to some
$\beta\in\Pi(\Gamma)$. The interior integral points of this face are $\lambda$s
such that $\langle \lambda+\rho,\beta^{\vee}\rangle=0,\
\langle \lambda+\rho,\gamma^{\vee}\rangle\in\mathbb{Z}_{>0}$
for each $\gamma\in\Pi(\Gamma)\setminus\{\beta\}$. By~\Lem{lemmik},
$\gamma\in \Pi(\Gamma)$ is $\lambda$-minimal for
$\lambda\in X_k(\Gamma)$ in the sense of~\ref{CC'} iff
$\langle \lambda+\rho,\gamma^{\vee}\rangle\in\mathbb{Z}_{>0}$.
In particular, if $\lambda$ is an interior integral point of the face
corresponding to $\beta$, then  each $\gamma\in\Pi(\Gamma)\setminus\{\beta\}$
is $\lambda$-minimal. By~\Prop{propalmin},
$\im\Upsilon_{L(\lambda)}\subset (\Pi(\Gamma)\setminus\{\beta\})^{\perp}$.
By~(\ref{som}), in order to show that $\lambda$ is $k$-admissible it is
enough to verify that
$((\Pi(\Gamma)\setminus\{\beta\})^{\perp}\cap
\delta^{\perp})\subset \mathbb{C}\delta$. One has
$(\Pi(\Gamma)\setminus\{\beta\})^{\perp}\cap
\delta^{\perp}=(\Pi(\Gamma)\setminus\{\beta\}\cup\{\delta\})^{\perp}$,
so it suffices to show that the span of
$(\Pi(\Gamma)\setminus\{\beta\})\cup\{\delta\}$ coincides with
the span of $\Delta$.
By above, $qr\delta=\sum_{\gamma\in \Pi(\Gamma)} a_{\gamma}\gamma$;
since $\Pi(\Gamma)$ is the set of simple roots of
an irreducible root subsystem, $a_{\gamma}\not=0$
for all $\gamma$. Hence $a_{\beta}\not=0$, so the span of
$(\Pi(\Gamma)\setminus\{\beta\})\cup\{\delta\}$ coincides with the span of
$\Gamma$, which is equal to the span of $\Delta$ as required.

\subsubsection{}\label{Lambda0}
Let us describe $X_k(\Gamma)$, where
$\Gamma\in B_k$ is such that the root system $\Gamma\cup(-\Gamma)$
is isomorphic to $\Delta^{\vee}$. Note that for a type $A_n^{(1)}$ all
$\Gamma\in B_k$ have this property.
Let $\Pi=\{\alpha_i\}_{i=0}^n$, where $\alpha_0$ is the affine root.
Let
$$M:=\{\nu\in\sum_{i=1}^n \mathbb{Q}\alpha_i|\
(\nu,\alpha_i)\in\mathbb{Z}\ \forall i=1,\ldots, n\}$$
and let $\tilde{W}$  be the semidirect product of the finite Weyl group
(generated by $\{s_{\alpha_i}\}_{i=1}^n$) and translations $t_{\alpha},
\alpha\in M$ (see~\cite{Kbook}, Chapter 6 or~\cite{KWmod2} for notation).
One has: $\tilde{W}=\tilde{W}^+\ltimes W$, where $W$ is the Weyl group of $\fg$
and $\tilde{W}^+(\Pi)=\Pi$.

Let $\Lambda_0\in\fh^*$ be such that
$\langle\Lambda_0,\alpha_0^{\vee}\rangle=1,
\langle\Lambda_0,\alpha_i^{\vee}\rangle=0$ for $i>0$.
For $u\in\mathbb{Z}_{>0}$ set
$$S_{(u)}:=
\{(u-1)K+\alpha_0^{\vee},\alpha_1^{\vee},\ldots,\alpha_n^{\vee}\}.$$

The proof of the following result (and, moreover, a complete description of
$B_k$ and $X_k(\Gamma)$ for $\Gamma\in B_k$)  is
similar to that Thm. 2.1 (resp. Thm. 2.2)  of~\cite{KWmod2}.
\subsubsection{}
\begin{thm}{}
Let $k+h^{\vee}=\frac{p}{q}$, where $p,q$ are coprime positive integers.
Let $\Gamma\in B_k$ be such that the root system $\Gamma\cup(-\Gamma)$
is isomorphic to $\Delta^{\vee}$. Then

(i) $\gcd (q,r^{\vee})=1$;

(ii) $\Pi(\Gamma)^{\vee}=y(S_{(q)})$ for some $y\in \tilde{W}$
such that $y(S_{(q)})\subset \Delta_+^{\vee}$;

(iii) $X_k(\Gamma)=\{y.(\lambda-\frac{p(q-1)}{q}\Lambda_0)|\
\langle\lambda+\rho,K\rangle=p \ \&\
\langle\lambda+\rho,\alpha^{\vee}\rangle\in \mathbb{Z}_{\geq 0}
\text{ for all }\alpha\in\Pi\}$.
\end{thm}

\subsubsection{}\begin{rem}{}
If $0<p<h^{\vee}$, then the polyhedra $\cP^k(\Gamma)$ do not have
interior integral points, since, by Thm. 2.1 of~\cite{KWmod2},
there are no KW-admissible weights of  level $\frac{p}{q}-h^{\vee}$
for such $p$.
\end{rem}

\subsubsection{Example}\label{exaaff}
Consider the example $\fg=\hat{\fsl}_2,\ k=-2+\frac{p}{q}$, where
$p,q$ are coprime positive integers. Let $\alpha,\delta-\alpha$ be the
simple roots. One has
$$B_k=\{\Gamma_r\}_{r=1}^q,\ \text{ where } \Gamma_r=\{(r-1)K+\alpha^{\vee};
(q-r+1)K-\alpha^{\vee}\};\ X_k(\Gamma_r)=\{\lambda_{r,s}\}_{s=0}^p,$$
where $\langle\lambda_{r,s}+\rho,(r-1)K+\alpha^{\vee}\rangle=s,\ \
 \langle\lambda_{r,s}+\rho,(q-r+1)K-\alpha^{\vee}\rangle=p-s$.
The interior integral points of $\Gamma_r$ are $\{\lambda_{r,s}\}_{s=1}^{p-1}$,
so the KW-admissible weights are $\lambda_{r,s}, r=1,\ldots, q;
s=1,\ldots, p-1$ (this set is empty for $p=1$). The face of
$\Gamma_r$ corresponding to $(r-1)\delta+\alpha$ (resp. to
$(q-r+1)\delta-\alpha$) is
$\lambda_{r,0}$ (resp. $\lambda_{r,p}$), and these points are interior
in this face. Thus $\lambda_{r,0},\lambda_{r,p}$ are $k$-admissible, and
the $k$-admissible weights are $\lambda_{r,s}, r=1,\ldots, q; s=0,\ldots, p$.

\subsection{Vacuum modules}\label{vac}
Retain notation of Subsection~\ref{genVerma}.
Let $\dot{\fg}$ be a simple finite-dimensional Lie algebra, and let
$\fg$ be its (non-twisted) affinization.
Let $l$ be the lacety of $\dot{\fg}$, i.e., the ratio of the lengths
squared of a long and a short root of $\dot{\fg}$.
Let $\dot{\Pi}$
be the set of simple  roots of $\dot{\fg}$, and let
$\theta$ be its highest root and $\theta'$ its highest short root.
The set of  simple roots of $\fg$ is
 $\{\alpha_0\}\cup \dot{\Pi}$, where $\alpha_0=\delta-\theta$.
We normalize the form $(-,-)$ in such a way that
$(\alpha_0,\alpha_0)=(\theta,\theta)=2$.
Let $\Delta$ (resp. $\dot{\Delta}$) be the set of roots
of $\dot{\fg}$ (resp. of $\fg$). Let $\rho$ (resp. $\dot{\rho}$)
be the Weyl vector for $\fg$ (resp. for $\dot{\fg}$).
Recall that $h^{\vee}=(\rho,\theta)+1$ is the dual Coxeter number
and $h=(\rho,\theta')+1$ is the  Coxeter number of $\dot{\fg}$,
and that both numbers are positive integers.

Retain notation of Subsections~\ref{genVerma} and~\ref{Lambda0}.
The space $\dot\Pi^{\perp}$ is the span of $\Lambda_0,\delta$.
The generalized Verma module $M_J(k\Lambda_0)$, where $J=\dot{\Pi}$
is called a {\em vacuum module} and is denoted by $V^k$;
the irreducible {\em vacuum module} $V_k$ is its quotient, i.e.
$V_k=L(k\Lambda_0)$.

By~\cite{GK}, $V^k$ is not irreducible iff $l(k+h^{\vee})$ is a non-negative
rational number, which is not the inverse of an integer. In particular,
if $k$ is irrational, then $V^k$ is irreducible. If $k$, hence
$k+h^{\vee}$, is rational, we write $k+h^{\vee}$ in minimal terms:
$$k+h^{\vee}=\frac{p}{q}, \ \text{ where }p,q\in\mathbb{Z},\ q>0,\
\gcd (p,q)=1.$$

\subsubsection{}
\begin{prop}{propvac1}
(i) $L(k\Lambda_0)$ is not weakly admissible iff
$k+h^{\vee}=\frac{p}{q}$ is rational with $p<h^{\vee}-1$ and $gcd(q,l)=1$,
or $p<h-1$ and $gcd(q,l)=l$.

(ii) $L(k\Lambda_0)$ is  KW-admissible iff $k$ is rational with
$p\geq h^{\vee}$ and $gcd(q,l)=1$
or $p\geq h$ and $gcd(q,l)=l$.

(iii) Any self-extension of $L(k\Lambda_0)$ with
a diagonal action of $K$ splits over $[\fg,\fg]$;

(iv) $L(k\Lambda_0)$ is not admissible for all irrational $k$ and for
$k+h^{\vee}=\frac{p}{q}$ rational with $p<h^{\vee}-1,\ \gcd(q,l)=1$, or
$p<h-1,\ \gcd(q,l)=l$.
\end{prop}
\begin{proof}
Retain notation of~\ref{Pilambda}. One has $\Pi(k\Lambda_0)=\dot{\Pi}$
 for irrational $k$.  For $k+h^{\vee}=\frac{p}{q}$ one has
$\Pi(k\Lambda_0)=\{q\delta-\theta\}\cup \dot{\Pi}$,
if $gcd(q,l)=1$, and
$\Pi(k\Lambda_0)=\{\frac{q}{l}\delta-\theta'\}\cup \dot{\Pi}$
 if $gcd(q,l)=l$. One has $\langle k\Lambda_0+\rho,\alpha^{\vee}\rangle=1$
for $\alpha\in \dot{\Pi}$,
$\langle k\Lambda_0+\rho,(q\delta-\theta)^{\vee}\rangle=p+1-h^{\vee}$, and
$\langle k\Lambda_0+\rho,(\frac{q}{l}\delta-\theta')^{\vee}\rangle=p+1-h$
if $gcd(q,l)=l$.
Now (i), (ii) follow from Subsection~\ref{noncrit},\ref{admKM} respectively.

Since the elements of $\dot{\Pi}$ are $k\Lambda_0$-minimal
in the sense of~\ref{CC'}, \Prop{propalmin} gives
\begin{equation}\label{vac3}
\im\Upsilon_{L(k\Lambda_0)}\subset \dot{\Pi}^{\perp}=\mathbb{C}\Lambda_0
+\mathbb{C}\delta
\end{equation}
for any $k$, and the above inclusion is equality for $k\not\in\mathbb{Q}$.
By~\Rem{rem2}, (iii) is equivalent to
$\im\Upsilon_{L(k\Lambda_0)}\cap \{\mu\in\fh^*|\ \langle\mu,K\rangle=0\}
=\mathbb{C}\delta$, which follows from the above inclusion.

For irrational $k$, from~\Rem{rem2} and the above equality
$\im\Upsilon_{L(k\Lambda_0)}=\mathbb{C}\Lambda_0+\mathbb{C}\delta$,
we conclude that $L(k\Lambda_0)$
admits  self-extensions which do not split over $[\fg,\fg]$,
so it is not admissible. Combining with (i) we obtain  (iv).
\end{proof}

\subsubsection{}
\begin{cor}{corvac1}
$L(k\Lambda_0)$ is $k$-admissible iff $k$ is irrational
or $k+h^{\vee}=\frac{p}{q}$ is rational with
 $p\geq h^{\vee}-1,\ \gcd(q,l)=1$, or
$p\geq h-1,\ \gcd(q,l)=l$.
\end{cor}

\subsubsection{}
\begin{rem}{remaffv}
By~\Rem{rem2}, $\Upsilon$ maps the self-extensions splitting over $[\fg,\fg]$
to $\mathbb{C}\delta$, so $L(k\Lambda_0)$ admits a non-splitting extension
over $[\fg,\fg]$ iff $\im\Upsilon_{L(k\Lambda_0)}\not=\mathbb{C}\delta$
which is equivalent to $\Lambda_0\in \im\Upsilon_{L(k\Lambda_0)}$, because
$\im\Upsilon_{L(k\Lambda_0)}\subset\mathbb{C}\delta+\mathbb{C}\Lambda_0$
by~(\ref{vac3}). From Subsection~\ref{genVerma},
$\Lambda_0\in \im\Upsilon_{L(k\Lambda_0)}$ iff
$M_J(k\Lambda_0)^1=M_J(k\Lambda_0)^2$, where $\{M_J(k\Lambda_0)^i\}$
is the Jantzen-type filtration constructed
for $\mu=\Lambda_0, \mu'=0$. This filtration is, by definition,
the Jantzen filtration $\{(V^k)^j\}$. We conclude that
\begin{equation}\label{vac1}
L(k\Lambda_0)\ \text{ admits a non-splitting extension
over } [\fg,\fg] \ \Longleftrightarrow\ (V^k)^1=(V^k)^2.
\end{equation}

For $\dot{\fg}=\fsl_2$ the module $L(k\Lambda_0)$ is admissible
iff $p\geq 2$, since in the sets of  admissible weights and
KW-admissible weights coincide (see Example~\ref{exsl23}).
We conjecture that, for $\dot{\fg}\not=\fsl_2$,
$L(k\Lambda_0)$ is admissible iff $k+h^{\vee}=\frac{p}{q}$,
where $p\geq h^{\vee}-1,\ \gcd(q,l)=1$,
$p\geq h-1,\ \gcd(q,l)=l$.
We verified this conjecture for simply-laced $\dot{\fg}$ using
formula~(\ref{vac1}) and the sum formula for the Jantzen filtration
given in~\cite{GK}.
\end{rem}

\section{Admissible modules for the Virasoro algebra}
\label{Vir}
In this section we prove~\Thm{thm06} and other results, stated in
Subsection~\ref{sect06}.

In this section $\cU$ is the universal enveloping algebra
of the Virasoro algebra.
Recall that $\fh$ is spanned by the central element $C$ and $L_0$;
$\fn_+$ (resp. $\fn_-$)
is spanned by $L_j$ (resp. $L_{-j}$) with $j\in\mathbb{Z}_{>0}$.
Note that the eigenvalues of $\ad L_0$ on $\fn_+$ are {\em negative}.

We consider only modules with a diagonal action of $C$.
We write $\mu\in\fh^*$ as $\mu=(h,c)$, where $\langle \mu,L_0\rangle=h,\
\langle \mu,C\rangle=c$ and we write a Verma module $M(\mu)$ as
$M(h,c)$. If $v$ is a weight vector of $M(h,c)$, its weight
is of the form $(h+j,c)$ for $j\in\mathbb{Z}_{\geq 0}$; the integer
$j$ is called the {\em level} of $v$;
we denote the corresponding weight space by $M(h,c)_{j}$
(instead of $M(h,c)_{(h+j,c)}$).

Let $\cH_c$ the category of $\cU$-modules with a fixed central charge $c$,
i.e. the category of $\cU/(C-c)$-modules.
We write $\Ext^1_c(N,N'):=\Ext^1_{\cH_c}(N,N')$.

We use the following parametrization of $c$:
\begin{equation}\label{ck}
c(k)=1-\frac{6(k+1)^2}{k+2},\ \ k\in\mathbb{C}\setminus\{-2\}.
\end{equation}

\subsection{Kac determinant}
Recall that
$$\dim M(h,c)_m=P(m),\ \text{ where }\
\sum_m P(m)q^{m}:=\prod_{j=1}^{\infty} (1-q^{j})^{-1}.$$
The Kac determinant $\det B_n$ is the Shapovalov determinant,
described in~\ref{shdet}, for the weight space $\cU(\fn_-)_n$.
The Kac determinant formula (see~\cite{KR},\cite{KWdet})
 can be written as follows
$$\det B_N(h,c(k))=\prod_{m,n\in \mathbb{Z}_{>0}}
(h-h_{m,n}(k))^{P(N-mn)},$$
where $h_{m,n}(k)$ for $m,n,k\in\mathbb{C}, k\not=-2$ is given by
$$h_{m,n}(k)=\frac{1}{4(k+2)}\bigl((m(k+2)-n)^2-
(k+1)^2\bigr).$$

Let $b(h,k)\in\mathbb{C}$ be such that $b(h,k)^2:=4(k+2)h+(k+1)^2$.
Observe that for $h=h_{m,n}(k)$, the straight line $x(k+2)-y=b(h,k)$
contains a point $(m,n)$ or $(-m,-n)$. Conversely, if
the straight line $x(k+2)-y=b(h,k)$ contains an integral point $(m,n)$,
then $h=h_{m,n}(k)=h_{-m,-n}(k)$.

\subsubsection{}
\begin{lem}{lemVir}
(i) The Verma module $M(h,c(k))$ is irreducible iff
the straight line $x(k+2)-y=b(h,k)$
does not contain integral points $(m,n)$ with $mn>0$.

(ii) A weight $(h,c(k))$ is weakly admissible
iff the straight line $x(k+2)-y=b(h,k)$
does not contain integral points $(m,n)$ with $mn<0$.
\end{lem}
\begin{proof}
The first statement follows from the fact that $h=h_{m,n}(k)$
iff one of the points $(m,n), (-m,-n)$
lies on the straight line $x(k+2)-y=b(h,k)$.
For (ii), observe that, by~\Lem{critwa}, $(h,c(k))$
is not weakly admissible iff there exist $m,n\in \mathbb{Z}_{>0}$ such that
$(m(k+2)-n)^2-(k+1)^2=4(k+2)(h-mn)$, which is equivalent to
$m(k+2)+n=\pm b(h,k)$.
\end{proof}

\subsubsection{}\label{corVir1}
If the straight line $x(k+2)-y=b$ does not contain integral points $(m,n)$
with $mn\not=0$, then $M(h,c(k))$ is irreducible and weakly admissible.

If $k$ is irrational,
the line $x(k+2)-y=b$ contains at most one integral point. Thus
$M(h,c(k))$ is irreducible or weakly admissible.

If $k+2\in\mathbb{Q}_{<0}$, there are two cases:
the line does not  contain integral points or
contains infinitely many integral points $(m,n)$ with $mn<0$,
so $M(h,c(k))$ is not weakly admissible.

If $k+2\in\mathbb{Q}_{>0}$, there are two cases:
the line does not  contain integral points or
contains infinitely many integral points $(m,n)$ with $mn>0$;
in this case $M(h,c(k))$ has an infinite length.

\subsubsection{}\label{vir2}
Assume  that $M(h,c(k))$ is reducible and
the weight $(h,c(k))$ is weakly admissible.
By~\ref{corVir1} , this happens if
$h=h_{m,n}(k)$ for some $m,n\in\mathbb{Z}_{>0}$ and either
$k$ is irrational,  or $k+2\in\mathbb{Q}_{>0}$
and the line $x(k+2)-y=b(h,k)$
does not contain integral points $(m',n')$ with $m'n'<0$.

Take $k+2\in\mathbb{Q}_{>0}$ and write $k+2=\frac{p}{q}$ in minimal terms.
Take $h=h_{m,n}(k)$ for $m,n\in\mathbb{Z}_{>0}$.
The equation $x(k+2)-y=b(h,k)$ can be rewritten as
$px-qy=pm-qn$. By above, the line $px-qy=pm-qn$
does not contain integral points $(m',n')$ with $m'n'<0$.
Since the integral points of this line are of the form
$(m+qj;n+pj), j\in\mathbb{Z}$, this holds iff the line contains
an integral point $(r,s)$ in the rectangle $0\leq r\leq q, 0\leq s\leq p$.
We conclude that  the set of weakly admissible weights $(h,c(k))$, such that
$M(h,c(k))$ is reducible, is $\{(h_{r,s}(k),c(k))\}$,
where $r,s\in\mathbb{Z},\ 0\leq r\leq q, 0\leq s\leq p$.
Notice that $h_{r,s}(k)=h_{q-r,p-s}(k)$.

\subsection{Jantzen-type filtration}
Define the map $\Upsilon_{L(\lambda)}$ as in Subsection~\ref{mapUps}.
The equality $\Ext^1_c(L(\lambda),L(\lambda))=0$ is equivalent to
$\mu\not\in \im\Upsilon_{L(\lambda)}$, where
$\mu\in\fh^*$ is defined by $\mu(L_0)=1,\ \mu(C)=0$.
Denote by $\{M(h,c)^i\}$ the Jantzen-type filtration introduced
in~\ref{Jan}, which is the image of the Jantzen filtration on
$M(\lambda+t\mu)$. We have $\langle C,\lambda+t\mu\rangle=c,\
\langle L_0,\lambda+t\mu\rangle=h+t$, hence
$$\det B_N(h+t,c)=\prod_{m,n\in \mathbb{Z}_{>0}}(h+t-h_{m,n}(k))^{P(N-mn)}.$$

\subsubsection{}\label{niceVir}
We call an integral point $(m,n)$ on the line
$x(k+2)-y=b(h,k)$ {\em minimal for $(h,c(k))$}
if the product $mn$ is positive and minimal among the
positive products: $mn>0$ and for any integral point $(m',n')$ on the line
$x(k+2)-y=b(h,k)$ one has $m'n'\leq 0$ or $m'n'\geq mn$.

\subsubsection{}
\begin{lem}{lemwaVir}
If for a weight $\lambda=(h,c)$ there is  exactly one minimal point
or exactly two minimal points of the form $(m,n),\ (-m,-n)$, then
$\Ext^1_c(L(\lambda),L(\lambda))=0$.
\end{lem}
\begin{proof}
In both cases the sum formula~(\ref{sumformula}) gives
$\sum_{j=1}^{\infty} \dim M(h,c)^j_{mn}=1$. Therefore
$M(h,c)^1\not=M(h,c)^2$, and, by~\Prop{corally},
$\mu\not\in \im\Upsilon_{L(\lambda)}$ as required.
\end{proof}

\subsection{Admissible weights}
Let us describe the $c$-admissible weights.

If the weight $(h,c)$ is $c$-admissible, then it is weakly admissible and
$M(h,c)$ is reducible. The weights with these properties are
described in~\ref{vir2}.

If $k$ is irrational, these weights are of the form
 $h=h_{m,n}(k)$ for $m,n\in\mathbb{Z}_{>0}$. The line
$x(k+2)-y=b(h,k)$ contains a unique integral point ($(m,n)$ or $(-m,-n)$),
so, by~\Lem{lemwaVir}, this weight is $c$-admissible.

Consider the case $k+2\in\mathbb{Q}_{>0}$. By~\ref{vir2},
the weights with the above properties are of the form
$(h_{m,n}(k),c(k))$ with $0\leq n\leq p, 0\leq m\leq q$.

If $0<n<p, 0<m<q$ and $(r,s)$ is minimal
for $(h_{m,n}(k),c(k))$, then $(r,s)\in\{(m,n), (m-q,n-p)\}$.
Since $mn\not=(m-q)(n-p)$, this minimal point is unique.
For $0<m<q$ (resp. for $0<n<p$)  the minimal point
for $(h_{m,p}(k),c(k))=(h_{q-m,0}(k),c(k))$
(resp. for $(h_{q,n}(k),c(k))=(h_{0,q-n}(k),c(k))$)
is $(m,p)$  (resp. $(q,n)$), and it is  unique.
For $(h_{0,0}(k),c(k))=(h_{q,p}(k),c(k))$ there are two
 minimal points: $(q,p)$ and $(-q,-p)$. Hence, by~\Lem{lemwaVir},
$(h_{m,n}(k),c(k))$ is $c$-admissible for $0\leq n\leq p, 0\leq m\leq q,
(m,n)\not=(0,p),(q,0)$.

For $(h_{0,p}(k),c(k))=(h_{q,0}(k),c(k))$ there are two
minimal points: $(q,2p)$ and $(-2q,-p)$.
By~\cite{Ast}, in this case the maximal proper submodule
of $M(h,c)$ is generated by a singular vector $v$ of level $2pq$.
Since $h_{2p,q}(k)=h_{p,2q}(k)$,
the sum formula implies that $\sum_{j=1}^{\infty}\dim  M(h,c)^j_{2pq}=2$.
By above, $\dim  M(h,c)^1_{2pq}=1$ so $\dim  M(h,c)^2_{2pq}=1$, that is
$v\in M(h,c)^2$. Thus $M(h,c)^1=M(h,c)^2$, so
$(h_{2p,q}(k),c(k))$ is not $c$-admissible.

\subsubsection{}
For $k=\frac{p}{q}-2$, where $p,q\in\mathbb{Z}_{>0},\ gcd(p,q)=1$,
write
$$c^{p,q}:=c(k)=1-\frac{6(p-q)^2}{pq},\ \
h^{p,q}_{r,s}:=h_{r,s}(k)=\frac{(pr-qs)^2-(p-q)^2}{4pq}.$$

\subsubsection{}
\begin{cor}{corVir}
(i) The weakly admissible weights $(h,c(k))$ with the property that
$M(h,c(k))$ is reducible are of the form $h=h_{r,s}(k)$,
where $r,s$ are integers, and

if $k$ is not rational, then  $r,s>0$;

if $k+2=\frac{p}{q}$  is rational, then
$0\leq r\leq q, 0\leq s\leq p$.

There are no such weights, if $k+2\in\mathbb{Q}_{<0}$.

(ii) All these weakly admissible weights, except  for the weights
$(h_{p,0}^{p,q},c^{p,q})=(h_{0,q}^{p,q},c^{p,q})$ are $c$-admissible.
\end{cor}
\subsubsection{}
\begin{rem}{}
Notice that $h_{1,1}(k)=0$ for all $k$. Hence the weight
$(0,c(k))$ is $c$-admissible iff $k+2\not\in\mathbb{Q}_{<0}$.
\end{rem}

\subsubsection{Minimal models}\label{mmVir}
Recall that $L(h,c)$ is called a {\em minimal model} if
$c=c^{p,q}$ and $h=h^{p,q}_{r,s}$,
where  $p,q$ are  coprime positive integers $\geq 2$, and $r,s$
are positive  integers with $r<q,\ s<p$.

Consider the set of integral points of the rectangle
$0\leq r\leq q, 0\leq s\leq p$.
Recall that $h^{p,q}_{r,s}=h^{p,q}_{q-r,p-s}$ so the symmetrical
points of this rectangle with respect to the rotation by $180^{\circ}$
give the same value of $h$.
We see that the $h^{p,q}_{r,s}$, corresponding to weakly admissible weights
with $c=c^{p,q}$, where $p$ and $q$ are coprime positive integers,
are parameterized by the integral points in the rectangle; the
$h^{p,q}_{r,s}$, corresponding to $c$-admissible weights, are parameterized
by the integral points of this rectangle, except for
$(0,p), (q,0)$ (which give the same $h$);
the $h^{p,q}_{r,s}$, corresponding to minimal models,
are parameterized by the inner integral points of this rectangle, namely
$1\leq r\leq q-1, 1\leq s\leq p-1$.

\subsection{Self-extensions of $L(h,c)$}
Using~\Lem{lemwaVir} and the description of the Jantzen
filtration given in~\cite{Ast}, it is not hard to show that
$\Ext^1_c(L(h,c),L(h,c))=\mathbb{C}$ in the following cases:

(a) $M(h,c)$ is irreducible;

(b) $b(h,k),\frac{b(h,k)}{k+2}\in\mathbb{Z}$ and $k+2\not=\pm 1, b(h,k)\not=0$.

(c)  $c=1, 25$ and $h\not=0$.

In all other cases $\Ext^1_c(L(h,c),L(h,c))=0$.

\section{Admissible modules for the Neveu-Schwarz superalgebra}
In this section $\cU$ is the universal enveloping algebra
of the Neveu-Schwarz superalgebra.
Recall that $\fh$ is spanned by the central element $C$ and $h:=L_0$;
$\fn_+$ (resp. $\fn_-$)
is spanned by $L_j$ (resp. $L_{-j}$) with $j\in\frac{1}{2}\mathbb{Z}_{>0}$.
Note that the eigenvalues of $L_0$ on $\fn_+$ are {\em negative}.

We consider only modules with a diagonal action of $C$.
We use the same notation for Verma modules as in Sect.~\ref{Vir}.
Let $\cH_c$ be the category of $\cU$-modules with a fixed central charge $c$,
i.e. the category of $\cU/(C-c)$-modules.
We write $\Ext^1_c(N,N'):=\Ext^1_{\cH_c}(N,N')$.

We use the following parameterization of $c$:
$$c(k)=\frac{3}{2}-\frac{12(k+1)^2}{2k+3},\ \
k\in\mathbb{C}\setminus\{-\frac{3}{2}\}.$$

\subsection{Kac determinant}
One has $\dim M(h,c)_N=P(N)$, where the partition function $P(N)$
is given by $\sum_{n\in \frac{1}{2}\mathbb{Z}} P(N)q^N=\prod_{j=1}^{\infty}
(1-q^j)^{-1}(1+q^{j-1/2})$.

The Kac determinant formula  for
$N\in\frac{1}{2}\mathbb{Z}_{\geq 0}$ can be written as follows
(see~\cite{KWdet}):

$$\det B_N(h,c(k))=\prod_{\scriptsize{\begin{array}{c}
m,n\in \mathbb{Z}_{>0}\\ m\equiv n\mod 2\end{array}}}
(h-h_{m,n}(k))^{P(N-\frac{mn}{2})},$$
where $h_{m,n}(k)$ for $m,n,k\in\mathbb{C}, k\not=-\frac{3}{2}$ is given by
$$h_{m,n}(k)=\frac{1}{2(2k+3)}\bigl((m(k+\frac{3}{2})-\frac{n}{2})^2-
(k+1)^2\bigr).$$

We call a point $(m,n)\in\mathbb{C}^2$  {\em nice} if
$m,n\in \mathbb{Z},\ m\equiv n\mod 2$.

Let $b(h,k)\in\mathbb{C}$ be such that
$b(h,k)^2:=4(2(2k+3)h+(k+1)^2)$. For $h=h_{m,n}(k)$,
the straight line $x(2k+3)-y=b(h,k)$
contains a point $(m,n)$ or $(-m,-n)$. Conversely, if
the straight line $x(2k+3)-y=b(h,k)$ contains a point $(m,n)$,
then $h=h_{m,n}(k)=h_{-m,-n}(k)$.
The following lemma is similar to~\Lem{lemVir}.

\subsubsection{}
\begin{lem}{lemNS}
(i) The Verma module $M(h,c(k))$ is irreducible iff
the straight line $x(2k+3)-y=b(h,k)$
does not contain nice points with $mn>0$.

(ii) A weight $(h,c(k))$ is weakly admissible
iff the straight line $x(2k+3)-y=b(h,k)$
does not contain nice points $(m,n)$ with $mn<0$.
\end{lem}

\subsubsection{}\label{corns1}
If the straight line $x(2k+3)-y=b$ does not contain nice points $(m,n)$
with $mn\not=0$, then $M(h,c(k))$ is irreducible and weakly admissible.

If $k$ is irrational,
the line $x(2k+3)-y=b$ contains at most one integral point. Thus
$M(h,c(k))$ is irreducible or weakly admissible.

If $k+\frac{3}{2}\in\mathbb{Q}_{<0}$, there are two cases:
the line does not  contain nice points or
contains infinitely many nice points $(m,n)$ with $mn<0$,
so $M(h,c(k))$ is not weakly admissible.

If $k+\frac{3}{2}\in\mathbb{Q}_{>0}$, there are two cases:
the line does not  contain integral points or
contains infinitely many nice points $(m,n)$ with $mn>0$; in this
case $M(h,c(k))$ has an infinite length.

\subsubsection{}\label{ns2}
Assume  that $M(h,c(k))$ is reducible and
the weight $(h,c(k))$ is weakly admissible.
By~\ref{corns1} , this happens if
$h=h_{m,n}(k)$ for some nice $(m,n)$ with $mn>0$, and either
$k$ is irrational,  or $k+\frac{3}{2}\in\mathbb{Q}_{>0}$
and the line $x(2k+3)-y=b(h,k)$
does not contain nice points $(m',n')$ with $m'n'<0$.

Take $k+\frac{3}{2}\in\mathbb{Q}_{>0}, h=h_{m,n}(k)$ and write
\begin{equation}\label{pqNS}
2k+3=\frac{p}{q},\ \ p,q\in\mathbb{Z}\setminus\{0\},\ q>0, \
p\equiv q\mod 2,\ gcd(\frac{p-q}{2},p)=1.
\end{equation}
Take $h=h_{m,n}(k)$ for nice $(m,n)$ with $mn>0$.
The equation $x(2k+3)-y=b(h,k)$ can be rewritten as
$px-qy=pm-qn$. The nice points of this line are of the form
$(m+qj;n+pj),\ j\in\mathbb{Z}$.
By above, the line $px-qy=pm-qn$
does not contain nice points $(m',n')$ with $m'n'<0$.
This is equivalent to the condition that this line contains
a nice point $(r,s)$ in the rectangle $0\leq r\leq q, 0\leq s\leq p$.
We conclude that  the set of weakly admissible weights $(h,c(k))$, such that
$M(h,c(k))$ is reducible, is $\{(h_{r,s}(k),c(k))\}$,
where $0\leq r\leq q, 0\leq s\leq p,\ r\equiv s\mod 2$.
Notice that $h_{r,s}(k)=h_{q-r,p-s}(k)$.

\subsection{Jantzen-type filtration}
Define the map $\Upsilon_{L(\lambda)}$ as in Subsection~\ref{mapUps}.
The formula $\Ext^1_c(L(\lambda),L(\lambda))=0$ means that
$\mu\not\in \im\Upsilon_{L(\lambda)}$, where
$\mu\in\fh^*$ is defined by $\mu(L_0)=1,\ \mu(C)=0$.
Denote by $\{M(h,c)^i\}$ the Jantzen-type filtration introduced
in~\ref{Jan}, which is the image of the Jantzen filtration on
$M(\lambda+t\mu)$. We have $\langle C,\lambda+t\mu\rangle=c,\
\langle L_0,\lambda+t\mu\rangle=h+t$, hence
$$\det B_N(h+t,c)=\prod_{\scriptsize{\begin{array}{c}
m,n\in \mathbb{Z}_{>0}\\ m\equiv n\mod 2\end{array}}}
(h+t-h_{m,n}(k))^{P(N-\frac{mn}{2})}.$$

\subsubsection{}\label{nicens}
Call a nice point $(m,n)$ belonging to one of the lines $x(2k+3)-y=\pm b(h,k)$
{\em minimal for $(h,c(k))$} if the product $mn$ is positive and
minimal among the positive products: $mn>0$ and for any nice point
$(m',n')$ on the line
$x(2k+3)-y=b(h,k)$ one has $m'n'\leq 0$ or $m'n'\geq mn$.

\subsubsection{}
\begin{lem}{lemwaNS}
If for a weight $\lambda=(h,c)$ there is exactly one minimal point
or exactly two minimal points of the form $(m,n),\ (-m,-n)$, then
$\Ext^1_c(L(\lambda),L(\lambda))=0$.
\end{lem}

The proof is the same as in~\Lem{lemwaVir}.

\subsection{Admissible weights}
Let us describe the $c$-admissible weights.

If the weight $(h,c)$ is $c$-admissible, then it is weakly admissible and
$M(h,c)$ is reducible. The weights with these properties are
described in~\ref{ns2}.

If $k$ is irrational, these weights are of the form
 $h=h_{m,n}(k)$ for nice $(m,n)$ with $m,n>0$. The line
$x(2k+3)-y=b(h,k)$ contains a unique integral point ($(m,n)$ or $(-m,-n)$),
so, by~\Lem{lemwaNS}, this weight is $c$-admissible.

Consider the case $k+\frac{3}{2}\in\mathbb{Q}_{>0}$. By~\ref{ns2},
the weights with the above properties are of the form
$(h_{m,n}(k),c(k))$ with $0\leq n\leq p, 0\leq m\leq q, m\equiv n\mod 2$.
If $0<m<q, 0<n<p$ and
$(r,s)$ is minimal for $(h_{m,n}(k),c(k))$, then
$(r,s)\in\{(m,n), (m-q,n-p)\}$.
One has $mn=(m-q)(n-p)$ iff $(m,n)=(\frac{q}{2},\frac{p}{2})$,
which is impossible, since $(\frac{q}{2},\frac{p}{2})$
is not nice. Hence $(h_{m,n}(k),c(k))$ has a unique minimal point.
For $0<m<q$ (resp. for $0<n<p$)  the minimal point
for $(h_{m,p}(k),c(k))=(h_{q-m,0}(k),c(k))$
(resp. for $(h_{q,n}(k),c(k))=(h_{0,q-n}(k),c(k))$)
is $(m,p)$  (resp. $(q,n)$), and it is  unique.
For $(h_{0,0}(k),c(k))=(h_{q,p}(k),c(k))$ there are two
 minimal points: $(q,p)$ and $(-q,-p)$. Hence, by~\Lem{lemwaNS},
$(h_{m,n}(k),c(k))$ is $c$-admissible for nice $(m,n)$, where
$0\leq n\leq p, 0\leq m\leq q,\ (m,n)\not=(0,p),(q,0)$.

By~\cite{Ast},\cite{IK1},\cite{IK2}, the Verma modules over the Neveu-Schwarz
algebra have no subsingular vectors. For the weight
$(h_{q,0}^{(p,q)},c^{(p,q)})=(h_{0,p}^{(p,q)},c^{(p,q)})$,
where $p,q$  are even, there are two
 minimal points: $(q,2p)$ and $(-2q,-p)$. The absence of
subsingular vectors implies that the maximal proper submodule
of $M(h_{q,0}^{(p,q)},c^{(p,q)})$ is generated by a singular vector of
level $pq$ and the sum  formula implies that this vector belongs
to $M(h_{q,0}^{(p,q)},c^{(p,q)})^2$. Thus
$M(h_{q,0}^{(p,q)},c^{(p,q)})^1=M(h_{q,0}^{(p,q)},c^{(p,q)})^2$
and so $L(h_{p,0}^{(p,q)},c^{(p,q)})$ is not $c$-admissible.

\subsection{}
\begin{cor}{corNS}
(i) The weakly admissible weights $(h,c(k))$ such that
$M(h,c(k))$ is reducible are of the form $h=h_{r,s}(k)$, where
$r$ and $s$ are integers and

if $k$ is not rational, then  $r,s>0, r\equiv s\mod 2$;

if $k+\frac{3}{2}\in\mathbb{Q}_{>0}$ and $p,q$
are chosen as in~(\ref{pqNS}), then $0\leq r\leq q, 0\leq s\leq p$
and $r\equiv s\mod 2$.

There are no such weights, if $k+\frac{3}{2}\in\mathbb{Q}_{<0}$.

(ii) All these weakly admissible
weights, except for the weights $h_{q,0}(k)=h_{0,p}(k)$
(when $p,q$ are even) are $c$-admissible.
\end{cor}

\subsubsection{}
\begin{rem}{}
Notice that $h_{1,1}(k)=0$ for all $k$. Hence the weight
$(0,c(k))$ is $c$-admissible iff $k+\frac{3}{2}\not\in\mathbb{Q}_{<0}$.
\end{rem}

\subsubsection{}
\begin{rem}{}
If $k$ is rational and $p,q$ are chosen as in~(\ref{pqNS}), then
$$c^{(p,q)}:=c(k)=\frac{3}{2}(1-\frac{2(p-q)^2}{pq}),\ \ \
 h_{r,s}^{(p,q)}:=h_{r,s}(k)=\frac{(pr-qs)^2-(p-q)^2}{8pq}.$$
Notice that $c^{(p,q)}<3/2$ and that $h_{r,s}^{(p,q)}=h_{q-r,p-s}^{(p,q)}$.
\end{rem}

\subsection{Minimal models}\label{minmodNS}
Recall that the minimal series are the modules
$L(h_{r,s}^{(p,q)},c^{(p,q)})$,
where  $p,q\in\mathbb{Z}_{\geq 2}$, $p\equiv q \mod 2,\
\gcd(\frac{p-q}{2},p)=1$, and $r,s$ are integers,
$0<r<q, 0<s<p$, $r\equiv s\mod 2$.

Consider the rectangle $0\leq r\leq q, 0\leq s\leq p$.
Recall that $h^{(p,q)}_{r,s}=h^{(p,q)}_{q-r,p-s}$ so the symmetrical
points of this rectangle with respect to the rotation by $180^{\circ}$
give the same value of $h$.
We see that the $h^{(p,q)}_{r,s}$, corresponding to weakly admissible weights
with $c=c^{(p,q)}$, where $p$ and $q$ are positive integers of the same parity
and $\gcd(\frac{p-q}{2},p)=1$,
are parameterized by the integral points $(r,s)$ such that
$r\equiv s\mod 2$ in this rectangle; the $h^{(p,q)}_{r,s}$,
corresponding to $c$-admissible weights, are parameterized by
these points except for $(0,p), (q,0)$ (which give the same $h$);
the $h^{(p,q)}_{r,s}$, corresponding to minimal models, are parameterized by
these points except  for the ones on the boundary of the rectangle.

\section{Vertex algebras}\label{sectva}
Recall that a vertex algebra $V$ is a vector space with a vacuum vector
$\vac$,  and a linear map $V\to (\End V)[[z,z^{-1}]]$,
$a\mapsto a(z)=\sum_{n\in\mathbb{Z}} a_{(n)}z^{-n-1}$, such that
$a(z)v\in V((z))$ for any $v\in V$, subject to the vacuum axiom
$\vac(z)=\id_V$, and the Borcherds identity~\cite{Kbook2}.

A vertex algebra $V$ is called {\em graded} if it is endowed with
 a  diagonalizable operator $L_0$ with non-negative
real eigenvalues such that $\Delta(a_{(j)}b)=\Delta(a)+\Delta(b)-j-1$,
where $\Delta(a)$ stands for the $L_0$-eigenvalue of $a$.
If $\Gamma$ is the additive subgroup of $\mathbb{R}$
generated by the eigenvalues of $L_0$, we call $V$ {\em $\Gamma$-graded}.
Given an $L_0$-eigenvector $a\in V$ we write the corresponding field
as $a(z)=\sum_{j\in\mathbb{Z}-\Delta(a)} a_jz^{-j-\Delta(a)}$.
Then one has the following commutator formula~\cite{Kbook2}, 4.6.3
\begin{equation}\label{commab}
[a_{m},b_{n}]=\sum_{j\in\mathbb{Z}_{\geq 0}}  \binom{m+\Delta(a)}{j}
(a_{(j)}b)_{m+n}.
\end{equation}
One also has
\begin{equation}\label{Ta}
(Ta)_n=-(n+\Delta(a))a_n,
\end{equation}
where $T\in\End(V)$ is the translation operator defined by $Ta:=a_{(-2)}\vac$,
and one has the following identity
\begin{equation}\label{borch1}\begin{array}{l}
(a_{(-1)}b)_n=\sum_{j=-1}^{-\infty} a_{j-\Delta(a)+1} b_{n-j-\Delta(b)}
+(-1)^{p(a)p(b)}\sum_{j=0}^{\infty} b_{n-j-\Delta(b)}a_{j-\Delta(a)+1},
\end{array}
\end{equation}
where $p(a)$ stands for the parity of $a$. Recall that the
identities~(\ref{commab}) and~(\ref{borch1}) together
are equivalent to the Borcherds identity (\cite{Kbook2} Section~4.8).

Recall that $V$ is called {\em $C_2$-cofinite} if $V_{(-2)}V$
has a finite codimension in $V$.

\subsection{Universal enveloping algebra}\label{UV}
The universal enveloping algebra of $V$,
introduced in~\cite{FZ}, can be defined as follows.

Let $Lie V$ be the quotient of the vector space with the basis
consisting of the formal symbols
$a_{m}$, where $a\in V_{\Delta(a)},m\in\mathbb{Z}-\Delta(a)$
subject to the following linearity
relations $(a+\gamma b)_{m}=a_{m}+\gamma b_{m}$
for $a\in V_{\Delta(a)}, b\in V_{\Delta(b)},
\gamma\in\mathbb{C}$ and relation~(\ref{Ta}). Then the
commutator formula~(\ref{commab}) induces on $Lie V$ a well defined
Lie superalgebra structure. Note that $Lie V$  is
a $\Gamma$-graded Lie superalgebra, where
$\deg a_j=j$.

Let $\tilde{\cU}$ be the universal enveloping algebra of the Lie superalgebra
$Lie V$ subject to the relation
$\vac_{(-1)}=1$. Extend the $\Gamma$-grading to $\tilde{\cU}$:
$\tilde{\cU}=\oplus_N\tilde{\cU}_N$, and define a
system of fundamental neighborhoods of zero in $\tilde{\cU}_N$ by
$\tilde{\cU}^s_N:=\{\sum_{\gamma\geq s} u_{-\gamma}u_{N+\gamma}|\
u_{\gamma}\in \tilde{\cU}_{\gamma}\}$.
Denote by $\hat{\tilde{\cU}}_N$ the completion of $\tilde{\cU}_N$;
the direct sum
$\hat{\tilde{\cU}}:=\oplus_{N\in\Gamma}\hat{\tilde{\cU}}_N$
is a $\Gamma$-graded complete topological algebra.
The {\em universal enveloping algebra} $\cU(V)$ of $V$
is the quotient of $\hat{\tilde{\cU}}$ by the  relations~(\ref{borch1}).
Since these relations are homogeneous, $\cU(V)$ inherits
 the $\Gamma$-grading.

\subsubsection{}
For example, if $V=V^k$ (resp. $V=\Vir^c$ or $V=\NS^c$)
is the affine vertex algebra (resp.  Virasoro or
Neveu-Schwarz vertex algebras), then
$\cU(V)$ is a completion of the universal enveloping algebra
$\cU(\fhg)/(K-k)$ (resp. of $\cU(\Vir)/(C-c)$ or of
 $\cU(\NS)/(C-c)$).

\subsection{Modules over vertex algebras}
Recall that a representation of a vertex algebra $V$ in a vector space $M$ is
a linear map $V\to (End M)[[z,z^{-1}]]$,
$a\mapsto a^M(z)=\sum_{n\in-\Delta(a)+\mathbb{Z}} a_n^M z^{-n-\Delta(a)}$,
such that $a(z)v\in M((z))$ for each $a\in V,v\in M$,
$\vac^M(z)=\id_M$, and the Borcherds identity holds.
Note that this is the same as a continuous representation
of the topological algebra $\cU(V)$ in $M$, endowed
with the discrete topology.

\subsection{Toric subalgebras}\label{toric}
For an associative superalgebra $B$ we call
$\fh\subset B$ {\em a toric (Lie) subalgebra}  if

(T1) $\fh$ is an even commutative Lie algebra:
$\forall a,b\in\fh\ p(a)=0\ \&\ [a,b]:=ab-ba=0$;

(T2) for any $a\in\fh$ the map $b\mapsto [a,b]$ is a diagonalizable
endomorphism of $B$.

The following proposition is similar to a part of Thm. 2.6 in~\cite{M}.

\subsubsection{}
\begin{prop}{hlocfin}
Let $V$ be a graded vertex algebra such that its vertex
subalgebra $W:=\oplus_{x\in\mathbb{Z}} V_x$
is $C_2$-cofinite. Let $\fh$ be a toric subalgebra of $\cU(W)$ such that

(T3) $\fh$ is spanned by elements of the form $a_0$ with $a\in W$.

Then every  $V$-module is $\fh$-locally finite.
\end{prop}

\begin{proof}
Notice that any continuous $\cU(V)$-module, viewed as a $\cU(W)$-module,
is continuous. Thus it is enough to prove that any $\cU(W)$-module is
$\fh$-locally finite. Hence we may (and will) assume that $W=V$.

Denote by $U'$ the unital commutative associative subalgebra of $\cU(V)$
generated by $\fh$. Let $M'$ be the $\fh$-locally finite part of $M$, i.e.
$M':=\{v\in M| \dim (U' v)<\infty\}$.
Since $\fh$ is commutative, $M'$ is the direct sum of
generalized $\fh$-weight spaces:
$$M'=\oplus_{\nu\in\fh^*} M'_{\nu},\ \  M'_{\nu}:=\{v\in M|\ \
\forall h\in\fh\
\exists N: (h-\langle \nu,h\rangle)^Nv=0\}.$$
From (T2) it follows that $M'$ is a $\cU(V)$-submodule of $M$.
Clearly, $M/M'$ is a continuous $\cU(V)$-module and
its $\fh$-locally finite part is trivial: for any non-zero
$v\in M/M'$ the space $U' v$ is infinite-dimensional.

Suppose that $M/M'\not=0$.
Fix a non-zero $v\in M/M'$ and let $M''$ be a cyclic $\cU(V)$-submodule
generated by $v$. For $m\in\mathbb{Z}$ introduce
$$M''(m):=\cU(V)_m v,$$
where $\cU(V)_m$ is the subspace of elements of degree $m$ in $\cU(V)$.
By Lemma 2.4 of~\cite{M}, $C_2$-cofiniteness of $V$ forces
$M''(m)=0$ for $m<<0$. Let $m\in\mathbb{Z}$
be minimal such that $M''(m)\not=0$.  Then for any $a\in V$
one has $a_{j}M''(m)=0$ for $j<0$. By~\cite{Zh}, there is a
well defined action of the Zhu algebra $A(V)$ on $M''(m)$:
for $a\in V$ its image in $A(V)$ acts on $M''(m)$ as $a_0$.
The $C_2$-cofiniteness of $V$ means that
$A(V)$ is finite-dimensional so the subalgebra of $\End(M''(m))$
generated by $a_0, a\in V$, is finite-dimensional.
By (T3), $\fh$ is spanned by elements of the form $a_0$ for some $a\in V$.
 Hence $\dim (U' v)<\infty$, a contradiction.
As a result, $M/M'=0$ as required.
\end{proof}

\subsubsection{}
\begin{rem}{remconf}
Recall that $\omega\in W$ is called  a {\em conformal vector},
if the field $\omega(z)$
is the Virasoro field $L(z)=\sum_{n\in\mathbb{Z}} L_nz^{-n-2}$,
(i.e. the $L_n$ satisfy the Virasoro commutation relations),
where  $L_{-1}=T$ and $L_0$ is a diagonalisable
operator (in this case $L_0$ defines a grading of $V$).

The examples of toric subalgebras of $\cU(W)$
satisfying (T3) can be constructed as follows.
First, if $W$ contains a conformal vector and
$\sum_{n\in\mathbb{Z}} L_nz^{-n-2}$ is the corresponding
Virasoro field, then  $\mathbb{C}L_0\subset\cU(W)$
is toric.

Second, let $a^1,\ldots, a^r\in W_1=V_1$ be even vectors such
that every $a^j_0$
is a diagonalizable endomorphism of $W$ and $a^j_0a^i=0$ for
all $i,j=1,\ldots,r$. Then $\{a^j_0\}_{j=1}^r$ span
a toric subalgebra of $\cU(W)$
satisfying (T3). These  follow from the formula
$[a_{0},b_{n}]=(a_{0}b)_{n}$.

Finally, if $\omega\in W$ is a conformal vector and $a^1,\ldots, a^r\in V_1$
are as in the second example, then  $\{L_0; a^j_{0}\}_{j=1}^r$
span a toric subalgebra of $\cU(W)$
satisfying (T3).
\end{rem}

\section{Admissible modules for the minimal $W$-algebras}
\label{Wa}
In this section $\fg=\fn_-\oplus\fh\oplus\fn_+$
is a finite-dimensional simple Lie algebra with a triangular decomposition,
$\theta$ is its maximal root, $h^{\vee}$ is its dual Coxeter number
 and $k\not=-h^{\vee}$ is a scalar. In this section we study self-extensions
of irreducible representations $\bL(\nu)$ of the vertex algebra
$\cW:=\cW^k(\fg,e_{-\theta})$. Recall that $\cW$ is a
$\frac{1}{2}\mathbb{Z}_{\geq 0}$-graded vertex algebra, called a {\em minimal
$W$-algebra}, constructed in~\cite{KRW},\cite{KWdet}. The results of this
section extend without difficulty to the case
of the Lie superalgebra $\fg=\osp(1,n)$.
The Virasoro and Neveu-Schwarz vertex algebras
 are particular cases of $W^k(\fg,e_{-\theta})$:
Virasoro case corresponds to $\fg=\fsl_2$ and Neveu-Schwarz case
to $\fg=\mathfrak{osp}(1,2)$.

The main results of this section are Theorems~\ref{thmadmW},\ \ref{thmmuW}.

\subsection{Structure of $\cW$}
Denote by $\Delta_+$ the set of positive roots of $\fg$. Define
$\rho$ and the bilinear form  $(-,-)$ on $\fh^*$ as in~\ref{symmKm};
we  normalize the form by the condition $(\theta,\theta)=2$.

\subsubsection{Notation}\label{notW}
Let $e\in\fg_{\theta}, h\in\fh, f\in\fg_{-\theta}$ be an $\fsl_2$-triple, and
let $x:=h/2$.
Let $\fg_j:=\{a\in\fg|\ [x,a]=ja\}$ be the $j$th eigenspace of $\ad x$.
One has
$$\fg=\fg_{-1}\oplus\fg_{-1/2}\oplus\fg_0\oplus\fg_{1/2}\oplus\fg_1,\ \
\fg_1=\mathbb{C}e,\ \fg_{-1}=\mathbb{C}f,\ \
\fg_{\pm 1/2}\subset \fn_{\pm}.$$

For a subspace $\fm\subset \fg$ we set $\fm^f:=\{a\in\fm|\ [f,a]=0\}$.
The centralizer of the triple $\{e,x,f\}$ is
$\fg^{\natural}:=\{a\in\fg_0|\ (x|a)=0\}=\fg_0^f$. One has
$$\fg_0=\fg^{\natural}\oplus
\mathbb{C}x, \ \
\fg^f=\mathbb{C}f\oplus \fg_{-1/2}\oplus \fg^{\natural},\ \
\fh=\fh^f\oplus\mathbb{C}x,\ \
\fh^f=\fg^{\natural}\cap\fh.$$

Then $\fg^{\natural}$ is a reductive Lie algebra,
$\fh^f$  is a Cartan subalgebra of $\fg^{\natural}$
and $\fg^{\natural}=(\fg_0\cap \fn_-)\oplus
\fh^f\oplus(\fg_0\cap \fn_+)$ is a triangular decomposition.
Let $P^{\natural}$ be a weight lattice for $\fg^{\natural}$
and $Q_+^{\natural}$ be the positive part of its root lattice
with respect to the above triangular decomposition. Note that $\fh^{\natural}$
acts diagonally on $\fg^f$ and the weights lie in $P^{\natural}$.

\subsubsection{Generators of $\cW$}
The structure of the vertex algebra
$\cW$ is explicitly described in~\cite{KWdet}.
Recall that $\cW$ contains a conformal vector $\omega$, for which
$\omega(z)$ is a Virasoro field $L(z)=\sum_n L_nz^{-n-2}$,
such that the action of $L_0$ endows  $\cW$ with a
$\frac{1}{2}\mathbb{Z}_{\geq 0}$-grading
$\cW=\oplus_{i\in \mathbb{Z}_{\geq 0}} \cW_{i/2}$.
Recall that a vector $a\in\cW$ and the corresponding
field $a(z)$ is called {\em primary of conformal
weight $\Delta$} if $L_0a=\Delta a$ and $L_ja=0$ for $j>0$.

The vertex algebra $\cW$ has primary fields $J^a(z)$ of conformal
weight $1$ (resp. $3/2$), labeled by $a\in \fg^{\natural}$
 (resp. $a\in\fg_{-1/2}$). We view the Fourier
coefficients $J^a_n$ as elements of the enveloping
algebra $\cU(\cW)$, see Subsection~\ref{UV}. One has
\begin{equation}\label{J0}
[J^{a}_{0}\, ,J^{b}_n]=J^{[a,b]}_n\ \text{ for }
a\in\fg^{\natural}, b\in\fg^{\natural}\cup\fg_{-1/2}.
\end{equation}

\subsubsection{The algebra $\cU(\cW)$}\label{UW}
Let $\{a_1,\ldots, a_s\}$ be a basis of
$\fg_{-1/2}\oplus \fg^{\natural}$ consisting of root vectors;
set $J^i(z):=J^{a_i}(z)$ for $i=1,\ldots,s$ and $J^0(z):=L(z)$.
Let $I:=\{ (i,n)|\ i=0,\ldots,s;\  n\in\mathbb{Z}-\Delta(J^i)\}$
and fix a total order on $I$ in such a way that $(i,m)<(j,n)$ for $m<n$.

By~\cite{KWdet}, Thm. 5.1, the vertex algebra $\cW$ is strongly generated by
the fields $J^{i}(z), i=0,\ldots,s$.
This means that the universal enveloping algebra $\cU(\cW)$,
has a topological PBW-basis, which consists of
the monomials of the form
$\prod_{(i,m)} (J_m^i)^{k_{m,i}}$, where $(i,m)\in I,
k_{m,i}\in\mathbb{Z}_{>0}$ with $k_{m,i}=1$ if $a_i$ is odd,
and the factors
are ordered with respect to the total order on $I$.

\subsubsection{The algebra $\fh_{\cW}$}
Write
$$\begin{array}{l}I=I_-\coprod I_0\coprod I_+,\ \text{
where }\\
I_{-}:=\{(i,n)\in I|\ n<0\text{ or }\ n=0\ \& \ a_i\in\fn_{-}\cap \fg_0\},\\
I_{+}:=\{(i,n)\in I|\ n>0\text{ or }\ n=0\ \&\  a_i\in\fn_+\cap \fg_0\},\\
I_0:=\{(i,0)|\ i=0\ \text{ or }a_i\in\fh^f\}.
\end{array}$$
Set
$$\fh_{\cW}:=\spn \{J^i_m, (i,m)\in I_0\}=\mathbb{C}L_0\oplus \{ J^a_0|\ a\in
\fh^f\}\subset \cU(\cW).$$

By~(\ref{J0}), the elements of $\fh_{\cW}$ commute  and
\begin{equation}\label{relW}
 [L_0,J^{a}_n]=-nJ^{a}_n,\ \
[J^{a}_0, J^{b}_n]=\langle a, \wt b\rangle J^{b}_n\ \text{ if }
 a\in\fh^f,\  b\in\fg^{\natural}\cup\fg_{-1/2},
\end{equation}
where $\wt b\in\fh^*$ stands for the weight of $b$.
We identify $\fh_{\cW}$ with $\mathbb{C}L_0\oplus \fh^f$
via the map $J^a_0\mapsto a$.
By above,  $\fh_{\cW}$ acts semisimply on $\cW$: for every
$a\in\fg^{\natural}+\fg_{-1/2}$ of weight $\beta\in\fh^*$ the weight of
 $J^a_n$ is
\begin{equation}\label{Jann}
\wt J^a_n=\mu\in\fh_{\cW}^*=(\mathbb{C}L_0\oplus \fh^f)^*,\ \text{ where }
\langle\mu, L_0\rangle=n,\ \ \mu|_{\fh^f}=\beta.
\end{equation}

We consider the adjoint action of $\fh_{\cW}$ on $\cU(\cW)$ given by
$h.u:=hu-uh,\ \ h\in \fh_{\cW}, u\in \cU(\cW)$;
this action is semisimple and for each $\nu\in\fh^*_{\cW}$
we denote by $\cU_{\nu}$ the corresponding weight space.
Note that $\fh_{\cW}$ is a toric Lie subalgebra of $\cU(\cW)$,
see Subsection~~\ref{toric}.

\subsubsection{}\label{order}
Let $Q_+\subset \fh_{\cW}$ be the semigroup generated by the weights of
$J^i_n$ with $(i,n)\in I_+$. One has
$$\nu\in Q_+\ \Rightarrow\ \
\langle \nu,L_0\rangle\in \frac{1}{2}\mathbb{Z}_{<0},\
\nu|_{\fh^f}\in P^{\natural}\ \text{ or }
\ \langle\nu, L_0\rangle=0,\ \nu|_{\fh^f}\in Q^{\natural}_+.$$
Note that $-Q_+\cap Q_+=\{0\}$.
Introduce a partial order on $\fh^*_{\cW}$ by
setting $\nu'\geq \nu$ iff $\nu'-\nu\in Q_+$.

\subsubsection{Triangular decomposition}\label{triW}
Since $\cW$ is strongly generated by $J^i(z)$, the subalgebra generated
by $\{J_m^i, (m,i)\in I_0\}$ is the symmetric algebra  $\cS(\fh_{\cW})$.

Fix a PBW-basis in $\cU(\cW)$ as in Subsection~\ref{UW}.
Denote by $U_{\pm}$ the span of elements of the PBW-basis
which are product of elements $\{J_m^i,(m,i)\in I_{\pm}\}$.
Notice that $U_{\pm}$
are closed subspaces of the topological algebra $\cU(\cW)$
and $U_+\cap U_-=\mathbb{C}$.

Denote by $U_+'$ (resp. $U_-'$) the closure of the left (resp. right)
ideal generated by the monomials
$\{J_m^i,(m,i)\in I_+\}$ (resp.  $\{J_m^i,(m,i)\in I_-\}$).
From the existence of PBW-basis, we conclude that
 the multiplication map induces embeddings
$U_-\otimes \cS(\fh_{\cW})\to \cU(\cW),
\ \cS(\fh_{\cW})\otimes U_+\to \cU(\cW)$ and that
\begin{equation}\label{eqtriW}
\cU(\cW)=(U_-\otimes \cS(\fh_{\cW}))\oplus U_+'=U'_-\oplus
(\cS(\fh_{\cW})\otimes U_+).\end{equation}

Clearly, $\cS(\fh_{\cW}), U_{\pm}$ are  $\fh_{\cW}$-submodules of $\cU(\cW)$;
we set $U_{\pm;\nu}:=\cU(\cW)_{\nu}\cap U_{\pm}$. One has
$$\Omega(U_+)\subset Q_+,\ \ \Omega(U_-)\subset -Q_+,\ \
\Omega(\cS(\fh_{\cW}))=\{0\}.$$

\subsubsection{}
\begin{lem}{U+nu}
One has $\dim U_{+;\nu}=\dim U_{-;-\nu}<\infty$ for any $\nu\in\fh^*_{\cW}$.
\end{lem}
\begin{proof}
Retain notation of~\ref{UW}.
The weight spaces $U_{\pm;\nu}$
are spanned by subsets of the PBW basis. In order to show that
$\dim U_{+;\nu}=\dim U_{-;-\nu}$ we construct an involution
on the PBW-basis which interchanges the vectors of weight $\nu$
with the vectors of weight $-\nu$ and the vectors lying in
$U_+$ with the vectors lying in $U_-$.
Introduce an involution
on the set $I$ as follows: for $a_i\in\fg_{\mu}$, where
$\mu\in\fh^*\setminus\{0\}$ set
$(i,n)\mapsto (i',-n)$, where $a_{i'}\in \fg_{-\mu}$ (such $i'$ is unique,
since $\fg_{\mu}$ is one-dimensional for $\mu\not=0$);
otherwise (if $i=0$ or $a_i\in\fh$) set $(i,n)\mapsto (i,-n)$.
Then the involution maps $I_+$ onto $I_-$.
Define the corresponding involution $\sigma$ on the set
$\{J^i_n,\ (i,n)\in I\}$ and extend $\sigma$
to the PBW basis of $\cU(\cW)$. Then $\sigma$ is
a required involution and this establishes the equality
$\dim U_{+;\nu}=\dim U_{-;-\nu}$.

It remains to verify that $\dim U_{+;\nu}<\infty$.
A monomial of the form $\prod (J_m^i)^{k_{m,i}}$ lying in $U_+$
can be written as $y_0y_+$, where $y_0$ is of the form
$\prod (J_0^i)^{k_{i}},\ a_i\in\fg_0\cap\fn_+$ and
$y_+$  is of the form $\prod (J_m^i)^{k_{m,i}}$
with all $m>0$.
Suppose that the monomial $y_0y_+$ has weight $\nu$.
Then  $\sum mk_{m,i}=-\langle \nu,L_0\rangle$.
Thus there are finitely many possibilities for $y_+$.
For a given monomial $y_+$ the weight of $y_0$ is fixed (it is $\nu-\wt y$).
Assign to the monomial $y_0=\prod (J_0^i)^{k_{i}}$ the element
$\prod a_i^{k_{i}}\in \cU(\fg_0\cap \fn_+)$.
Clearly, the images of the monomials of weight $\mu\in\fh^*_{\cW}$
form a PBW basis in $\cU(\fg_0\cap \fn_+)_{\mu'}$,
where  $\mu'\in (\fh^f)^*$ is the restriction of
$\mu\in\fh^*_{\cW}=(\fh^f\oplus\mathbb{C}L_0)^*$
to $\fh^f$. By~\ref{notW},
$(\fg_0\cap \fn_-)\oplus \fh^f\oplus (\fg_0\cap \fn_+)$
is a triangular decomposition of a reductive Lie algebra. As a result,
$\cU(\fg_0\cap \fn_+)_{\mu'}$ is finite-dimensional for any
$\mu'\in (\fh^f)^*$. Thus for a given $y_+$ there are finitely
many possibilities for $y_0$. Hence $U_{+;\nu}$ is finite-dimensional.
\end{proof}

Thus, the pair $(\cU(\cW),\fh_{\cW})$ satisfies the conditions
(U1)--(U5) of Subsection~\ref{assm} with (U2) replaced by~(\ref{eqtriW}).
In fact, it follows from~\cite{KWmod2} that (U6) holds as well,
but we will not need it.

\subsection{The category $\tilde{\CO}(\cW)$}
Let $\bN$ be a $\cU(\cW)$-module.
For $\nu\in\fh^*_{\cW}$ define $\bN_{\nu}$ and the set of generalized
weights as usual:
\begin{equation}\label{bNnu}
\bN_{\nu}:=\{w\in \bN|\ \forall u\in \fh_{\cW}\ (u-\langle \nu, u\rangle)^nw=0
\text{ for } n>>0\},\
 \end{equation}
$\Omega(\bN):=\{\nu\in\fh^*_{\cW}|\ \bN_{\nu}\not=0\}$.
Clearly, $\cU(\cW)_{\mu}\bN_{\nu}\subset \bN_{\mu+\nu}$.

Let $\tilde{\CO}(\cW)$ be the full subcategory of  the category
of $\cU(\cW)$-modules such that

(O1)  $\bN=\oplus_{\nu\in \Omega(\bN)} \bN_{\nu}$;

(O2) $\exists \nu_1,\ldots,\nu_m\in \fh^*_{\cW}\
\text{ s.t. } \forall \nu\in \Omega(\bN)\ \ \exists i\text{ for which }
\nu\leq \nu_i$.

Notice that any module in $\tilde{\CO}(\cW)$
is a continuous $\cU(\cW)$-module and that $\tilde{\CO}(\cW)$ is closed
with respect to  the extensions. This category is an analogue of the
category $\tilde{\CO}$ of modules over a Kac-Moody algebra.

\subsubsection{Verma modules over $\cW$}\label{verW}
For $\nu\in \fh^*_{\cW}$
a Verma module $\bM(\nu)$ can be defined as follows:
extend $\nu$ to an algebra homomorphism $\nu:\cS(\fh_{\cW})\to\mathbb{C}$
and let $\Ker\nu$ be its kernel; then
$$\bM(\nu)=\cU(\cW)/\bigl(\cU(\cW)\Ker\nu+U'_+\bigr),$$
where $U'_+$ is the left ideal introduced in Subsection~\ref{triW}.
A Verma module is a cyclic $\cU(\cW)$-module generated by $v$
of weight $\nu$ such that $U'_+v=0$. Using~(\ref{eqtriW})
we can identify  $\bM(\nu)$ with $U_-$ as vector spaces:
the preimage of the weight space $U_{-,\mu}$ is the weight space
$\bM(\nu)_{\nu+\mu}$; in particular,
$$\Omega(\bM(\nu))\subset (\nu-Q_+),\
\ \dim \bM(\nu)_{\nu}=1.$$
By~\Lem{U+nu}, the weight spaces of $\bM(\nu)$ are finite-dimensional so
$\bM(\nu)$ lies in $\tilde{\CO}(\cW)$.
The module $\bM(\nu)$ has a unique irreducible quotient which we denote by
$\bL(\nu)$; this is an irreducible module with highest weight $\nu$.
 Any irreducible module in $\tilde{\CO}(\cW)$
is an irreducible  highest module.

\subsection{Functor $H$}
Let $\fhg$ be the affinization of $\fg$, $\fhh=\fh\oplus\mathbb{C}K\oplus
\mathbb{C}D$ be the Cartan subalgebra of $\fhg$; define $\hat{\rho}\in\fhh^*$
for $\fhg$ as usual (see Subsection~\ref{symmKm}).
Let $\delta$ be the minimal imaginary root and let $\alpha_0=\delta-\theta$.
We denote by $s_0$ the reflection with respect to
$\alpha_0$ ($s_0\in W\subset GL(\fhh^*)$).
For $\lambda\in\fhh^*$ we denote by $M(\lambda)$
the $\fhg$-module with highest
weight $\lambda$ and by $L(\lambda)$ its irreducible quotient.
Consider the category $\CO$ for the affine Kac-Moody algebra $\fhg$
and let $\CO_k$ be the full subcategory of the category $\CO$
with the objects $V$ such that $K|_V=k\cdot\id$. Set
$$\fhh^*_k:=\{\lambda\in\fhh^*|\ \langle\lambda, K\rangle=k\}.$$

For $N\in\Obj\CO$ we denote by $N^{\#}$ its graded dual, see~\ref{dual};
one has $L(\lambda)^{\#}\cong L(\lambda)$.

\subsubsection{}\label{redH}
Reduction functors $V\mapsto H^i(V), i\in\mathbb{Z}$, from the category
$\CO_k$ to the category of continuous $\cW$-modules were
introduced in~\cite{KRW}; in Subsection~\ref{Hi} we describe these functor for
$\fg=\fsl_2$. The following properties
are proven in~\cite{KRW},\cite{KWdet},\cite{Ar}:

(H1) for any $V\in \Obj\CO_k$ one has $H^i(V)=0$ for $i\not=0$;
the functor $V\mapsto H^0(V)$ is exact;

(H2) for  $\lambda\in\fhh^*_k$ one has
$H^0(M(\lambda))=\bM(\lambda_{\cW})$, where $\lambda_{\cW}\in\fh^*_{\cW}$
is given by
\begin{equation}\label{lambdaW}
\lambda_{\cW}|_{\fh^f}=\lambda|_{\fh^f},\ \
\lambda_{\cW} (L_0)=\frac{(\lambda+2\hat{\rho},
\lambda)}{2(k+h^{\vee})}-
\langle\lambda,x+D\rangle;
\end{equation}

(H3) for  $\lambda\in\fhh^*_k$ the continuous $\cW$-module
$H^0(L(\lambda))$ is irreducible ($\cong\bL(\lambda_{\cW})$) if
$(\lambda,\alpha_0)\not\in\mathbb{Z}_{\geq 0}$ and
$H^0(L(\lambda))=0$ if
$(\lambda,\alpha_0)\in\mathbb{Z}_{\geq 0}$;

(H4) for any $N\in \Obj\CO_k$ and any irreducible
module $L\in\CO_k$ such that $H^0(L)\not=0$ one has
$$\ \ \ \ \ \ \ [H^0(N): H^0(L)]=\sum_{\lambda: H^0(L(\lambda))=H^0(L)}
[N: L(\lambda)];$$

(H5) for  $\lambda\in\fhh^*_k$
any non-zero submodule of $H^0(M(\lambda)^{\#})$ intersects non-trivially
the weight space $H^0(M(\lambda)^{\#})_{\lambda_{\cW}}$
(\cite{Ar}, Thm. 6.6.2).

From (H2) it follows that any Verma  $\cW$-module
is the image of $M(\lambda)$ with $\lambda\in\fhh^*_k$.
From this and~\Lem{lamW} below,
all simple objects in $\tilde{\CO}(\cW)$ are of the form
$H^0(L(\lambda))$ for $\lambda\in\fh^*_k$.

\subsubsection{}\label{COO}
Let $\CO_k'$ be the full category of $\fhg$-modules $N$ admitting finite
filtrations with the subquotients belonging to the category $\CO_k$
and with the condition $K|_N=k\cdot\id$.

Extend the functors $N\mapsto H^i(N)$ from $\CO_k$
to $\CO_k'$ (define the differential $d$  by the same formula).
The property (H1) ensures that $H^i(N)=0$ for $i\not=0$ and that
the functor $N\mapsto H^0(N)$ is exact.

\subsubsection{}
\begin{lem}{lamW}
For $\lambda,\lambda'\in\fhh^*_k$ one has
$$H^0(M(\lambda))\cong H^0(M(\lambda')) \ \Longleftrightarrow\
\lambda'-\lambda\in\mathbb{C}\delta\ \text{ or }
\lambda'-s_0.\lambda\in\mathbb{C}\delta.$$
\end{lem}
\begin{proof}
From (H2)  one sees that $H^0(M(\lambda))\cong H^0(M(\lambda'))$
is equivalent to the conditions
\begin{equation}\label{shum}
(\lambda'-\lambda)|_{\fh^f}=0\ \text{ and }\ \
\frac{(\lambda'-\lambda,\lambda'+\lambda+2\hat{\rho})}{2(k+h^{\vee})}-
\langle \lambda'-\lambda,x+D\rangle=0.
\end{equation}
Write $\lambda'-\lambda$ in the form
$$\lambda'-\lambda=a\delta+a_0\alpha_0+\mu,\
\text{ where } a,a_0\in\mathbb{C},\text{ and }
\langle\mu,K\rangle=\langle\mu,D\rangle=\langle\mu,x\rangle=0.$$
Then $(\lambda'-\lambda)|_{\fh^f}=0$ means that
$\mu|_{\fh^f}=0$, hence $\mu=0$.
Now the second condition of~(\ref{shum}) takes the form
$$\frac{(2\lambda+2\hat{\rho}+a\delta+a_0\alpha_0,
a\delta+a_0\alpha_0)}{2(k+h^{\vee})}-
\langle a\delta+a_0\alpha_0,x+D\rangle=0,$$
which is equivalent to $(2\lambda+2\hat{\rho}+a_0\alpha_0,a_0\alpha_0)=0$.
The assertion follows.
\end{proof}

\subsubsection{}
Define $W(\lambda)\subset\hat{W}$ as in Subsection~\ref{Wlambda}.

\begin{cor}{corlamW}
If $N$ is a subquotient of $M(\lambda)$, then any irreducible subquotient
of $H^0(N)$ is isomorphic to $H^0(L(w.\lambda))$ for some $w\in W(\lambda)$,
and
$[H^0(N): H^0(L(w.\lambda))]=[N: L(w\lambda)]$, provided that
$H^0(L(w.\lambda))\not=0$.
\end{cor}
\begin{proof}
By (H4), $[H^0(N): H^0(L)]=\sum_{\lambda': H^0(L(\lambda'))=H^0(L)}
[N: L(\lambda')]$. Since $N$ is a subquotient of $M(\lambda)$,
$[N: L(\lambda')]\not=0$ forces $\lambda'=w\lambda$ for some $w\in W(\lambda)$.
One has
$$[H^0(N): H^0(L(w.\lambda))]=\sum_{y\in W(\lambda): H^0(L(y.\lambda))=
H^0(L(w.\lambda))}[N: L(y.\lambda)].$$
By~\Lem{lamW}, $H^0(L(y.\lambda))=H^0(L(w.\lambda))$ forces
$y.\lambda=w.(\lambda+a\delta)$ or $y.\lambda=s_0w.(\lambda+a\delta)$
for some $a\in\mathbb{C}$. Recall that the value
$(\lambda'+\hat{\rho},\lambda'+\hat{\rho})$ is invariant of a $W.$-orbit.
Notice that $(\lambda+\hat{\rho}+a\delta,\lambda+\hat{\rho}+a\delta)=
(\lambda+\hat{\rho},\lambda+\hat{\rho})$ forces $a=0$ or
$(\lambda+\hat{\rho},\delta)=0$. Therefore
$H^0(L(y.\lambda))=H^0(L(w.\lambda))$ forces
$y.\lambda=w.\lambda$ or $y.\lambda=s_0w.\lambda$.
If $(w.\lambda,\alpha_0)\not\in\mathbb{Z}$, then $s_0w\not\in W(\lambda)$.
Since $H^0(L(w.\lambda))\not=0$,
$(w.\lambda,\alpha_0)\not\in\mathbb{Z}_{\geq 0}$.
If $(w.\lambda,\alpha_0)\in\mathbb{Z}_{\leq -2}$, then
$(s_0w.\lambda,\alpha_0)\in\mathbb{Z}_{\geq 0}$, so
$H^0(L(s_0w.\lambda))=0$. Finally, if $(w.\lambda,\alpha_0)=-1$,
then $s_0w.\lambda=w.\lambda$. Hence
$H^0(L(y.\lambda))=H^0(L(w.\lambda))$ forces $y.\lambda=w.\lambda$
and so $[H^0(N): H^0(L(w.\lambda))]=[N: L(w\lambda)]$ as required.
\end{proof}

\subsection{Admissible modules}
Define the admissible weights for $\cW$ in the same way as before. That is
$\lambda\in\fh^*_{\cW}$ is called {\em weakly admissible}
if $[\bM(\lambda'):\bL(\lambda)]\not=0$ forces
$\lambda'=\lambda$, and $\lambda$ is  called {\em admissible}
if it is  weakly admissible and $\Ext^1(\bL(\lambda),\bL(\lambda))=0$.

\subsubsection{}
\begin{defn}{}
Let $\Adm$ be the full category of  $\cW$-modules $\bN$ which
are locally finite $\fh_{\cW}$-modules and are
such that every irreducible subquotient of $\bN$ is admissible.
\end{defn}

\subsection{Main results}
One of our main results is~\Thm{thmmuW}, which relates self-extensions
of a $\cW$-module $\bL(\lambda_{\cW})$ to self-extensions
of the $\fhg$-module $L(\lambda)$.
Using  this theorem we obtain important information
on admissible weights for $\cW$ for rational $k$ in
Subsection~\ref{admisW}. The main results here are the following.

Recall that the set of rational weakly admissible weights  $X_k$ for $\fhg$
is the union of the sets of all integral points of all the polyhedra
$\mathcal{P}^k(\Gamma), \ \Gamma\in B_k$, whereas that of
KW-admissible weights is the union of the sets of
all integral points in the interiors of all the polyhedra
$\mathcal{P}^k(\Gamma), \ \Gamma\in B_k$. In~\Prop{propadmW} (iii)
we show that the set of admissible weights in $\fh^*_{\cW}$ lies
in the image of $X_k$ under the map $\lambda\mapsto \lambda_{\cW}$.
On the other hand, by~\Cor{corpolyh}, the  set of
admissible weights in $\fh^*_{\cW}$
contains the image of the set of  KW-admissible weights.

Another main result is the following theorem.

\subsubsection{}
\begin{thm}{thmadmW}
(i) If $k+h^{\vee}\in\mathbb{Q}_{>0}$ the category $\Adm$
is semisimple with finitely many irreducibles and, in particular,
is a subcategory of $\tilde{\CO}(\cW)$.

(ii) If $k+h^{\vee}\in\mathbb{Q}_{<0}$ the category $\Adm$
is empty.
\end{thm}

Our strategy of the proof of~\Thm{thmadmW} is as follows.
In~\Prop{propadmW} (iii) we show that if a $\cW$-module
$\bL(\lambda')$ is admissible for $\lambda'\in\fh^*_{\cW}$,
then there exists  weakly admissible rational
$\lambda\in\fhh^*_k$ such that $\lambda'=\lambda_{\cW}$.
The set of weakly admissible rational weights of level $k$
is described in Subsection~\ref{polyh} (it is denoted by $X_k$);
this set is finite for $k+h^{\vee}\in\mathbb{Q}_{>0}$ and is empty for
$k+h^{\vee}\in\mathbb{Q}_{<0}$.
Therefore the category $\Adm$  is empty for $k+h^{\vee}\in\mathbb{Q}_{<0}$
and has finitely many irreducibles for $k+h^{\vee}\in\mathbb{Q}_{>0}$.
The formula $\Ext^1(\bL',\bL)=0$ for non-isomorphic
weakly admissible $\cW$-modules $\bL$ and $\bL'$
is established in~\Cor{lemlalaW}.
Finally, using~\Lem{lemcH}, we obtain the semisimplicity.

\subsection{$\Ext^1(\bL,\bL')=0$}
In this subsection we show that  $\Ext^1(\bL,\bL')=0$ if
 $\bL\not\cong\bL'$ are irreducible  weakly admissible  $\cW$-modules.

\subsubsection{}
\begin{lem}{Wdual}
Let $\bN\in\tilde{\CO}(\cW)$ be a module with finite-dimensional
generalized weight spaces $\bN_{\nu}$. Suppose that for some
$\nu\in\fh^*_{\cW}$ one has

(i) $\Omega(\bN)\subset \{\nu'\in\fh^*_{\cW}|\
\nu'\leq \nu\}$;

(ii) each submodule of $\bN$ intersects  $\bN_{\nu}$ non-trivially
and $\dim \bN_{\nu}=1$.

Then $\bN$ is isomorphic to a submodule of $H^0(M(\lambda)^{\#})$
for $\lambda\in\fh^*_k$  such that $\lambda_{\cW}=\nu$.
\end{lem}
\begin{proof}
For a continuous $\cW$-module $M=\oplus_{\nu\in\Omega(M)} M_{\nu}$
with finite-dimensional generalized weight spaces $M_{\nu}$, we
view $M^{\#}:=\oplus M_{\nu}^*$
as a right $\cU(\cW)$-module via the action $(fu)(v):=f(uv)$, $f\in M^{\#},
u\in\cU, v\in M$. Take $\bN$ satisfying the assumptions (i), (ii)
and let $\lambda\in\fh^*_k$ be such that $\lambda_{\cW}=\nu$.
Set
$$\bM':=H^0(M(\lambda)^{\#}).$$
View $\bN^{\#}, (\bM')^{\#}$ as right $\cU(\cW)$-modules.
Let us  show that $\bN^{\#}$ is a quotient of $(\bM')^{\#}$.

View $\bN^*_{\lambda_{\cW}}$ as a subspace of $\bN^{\#}$.
The assumption (ii) implies that $\bN^*_{\lambda_{\cW}}$ is one-dimensional
and it generates $\bN^{\#}$.
For each weight element $a\in\cU(\cW)$ with $\wt a<0$ one has
$\bN_{\lambda_{\cW}}^*a=0$, because $a\bN\cap \bN_{\lambda_{\cW}}=0$
by (i).
Therefore $\bN_{\lambda_{\cW}}^*U'_-=0$, where $U'_-$ is the right ideal
introduced in Subsection~\ref{triW}.
Extend $\lambda_{\cW}\in\fh_{\cW}^*$ to an algebra
homomorphism $\lambda_{\cW}:\cS(\fh_{\cW})\to\mathbb{C}$
and let $\Ker\lambda_{\cW}$ be its kernel.
By above, $\bN_{\lambda_{\cW}}^*$ is annihilated by
a right ideal $J:=U'_-+\Ker\lambda_{\cW}\cdot\cU(\cW)$.
Thus $\bN^{\#}$ is a quotient of a right cyclic $\cU(\cW)$-module
$\cU(\cW)/J$. We will show below that
\begin{equation}\label{rigU}
(\bM')^{\#}\cong \cU(\cW)/J\ \text{ as right
$\cU(\cW)$-modules}.
\end{equation}
 This implies that $\bN$ is isomorphic to
a submodule of $\bM'$, as required.

It remains to verify~(\ref{rigU}).
By~(\ref{eqtriW}), we can identify $\cU(\cW)/J$ with
$U_+$ as vector spaces; for every $\mu\in\fh^*_{\cW}$
the preimage of a weight space
$U_{+,\mu}$ is the weight space of weight $\lambda_{\cW}-\mu$
(since the weight of $U_{+,\mu}$ with respect to the {\em right}
adjoint action $u.h:=uh-hu$ is $-\mu$). Thus
$$\dim (\cU(\cW)/J)_{\lambda_{\cW}-\mu}=\dim U_{+,\mu}
\overset
{\text{~\ref{U+nu}}}{=}
\dim  U_{-,-\mu}=
\dim \bM(\lambda_{\cW})_{\lambda_{\cW}-\mu}<\infty.$$
As usual, define the character of a $\cW$-module $M$ as
 $\ch_{\cW} M:=\sum_{\mu\in\fh^*_{\cW}}\dim M_{\mu} e^{\mu}$.
 From (H4) it follows that
\begin{equation}\label{chuhk}
\ch_{\cW} H^0(M(\lambda))=\sum_{\lambda'\in\fhh^*} [M(\lambda): L(\lambda')]
\cdot\ch_{\cW} H^0(L(\lambda'))=
\ch_{\cW} H^0(M(\lambda)^{\#}),
\end{equation}
since $[M(\lambda): L(\lambda')]=[M(\lambda)^{\#}: L(\lambda')]$
for any $\lambda'\in\fhh$.
As a result,  for each $\mu\in\fh^*_{\cW}$
\begin{equation}\label{chuhkk}
\dim\bM'_{\mu}=\dim \bM(\lambda_{\cW})_{\mu}=
\dim(\cU(\cW)/J)_{\mu}<\infty.\end{equation}
Observe that $\bM'$ satisfies the assumption (i),
by~(\ref{chuhk}), and  the assumption (ii),
by (H5).
Thus, by above, $(\bM')^{\#}$ is isomorphic to a quotient of $\cU(\cW)/J$,
and, by~(\ref{chuhkk}), $(\bM')^{\#}\cong \cU(\cW)/J$ as required.
\end{proof}

The following corollary of~\ref{verW} and \ref{Wdual}
is an analogue of~\Lem{proplala}.

\subsubsection{}\begin{cor}{lemlalaW}
 (i) If $0\to \bL(\lambda_{\cW})\to \bN\to \bL(\lambda'_{\cW})\to 0$
is a non-splitting extension  and $\lambda_{\cW}\not=\lambda'_{\cW}$,
then either $\bN$ is isomorphic to a quotient of
$\bM(\lambda'_{\cW})=H^0(M(\lambda'))$ or
$\bN$ is isomorphic to a submodule of $H^0(M(\lambda)^{\#})$.

(ii) If $\bL\not\cong\bL'$ are irreducible  weakly admissible  $\cW$-modules,
then $\Ext^1(\bL,\bL')=0$.
\end{cor}

\subsection{Self-extensions of irreducible $\cW$-modules}
\label{UpsWL}
Retain notation of~\ref{extlala}. For a quotient
$\bM'$ of $\bM(\lambda)$ introduce the natural map
 $\Upsilon_{\bM'}:\Ext^1_{\cW}({\bM}',{\bM}')\to \fh_{\cW}^*$
similarly to $\Upsilon_{M'}$ in~\ref{extlala}.
As in~\ref{extlala},
if $N'$ is an extension of $\bM'$ by $\bM'$ (i.e., $\bN'/\bM'\cong \bM'$)
denote by $\Upsilon_{\bM'}(\bN')$ the one-dimensional subspace
of $\fh^*_{\cW}$ spanned by the image of an exact sequence
$0\to \bM'\to\bN'\to\bM'\to 0$. For $\Upsilon_{\bM'}$ the
the  properties ($\Upsilon 1$)-($\Upsilon 4$) of Subsection~\ref{corUps}
can be deduces along the same lines, using~(\ref{eqtriW}).

Recall that $\fhh^*_0=\{\mu\in\fhh^*|\ \langle\mu,K\rangle=0\}$.
Note that if $M'\in \CO_k$ is a quotient of $M(\lambda)$ and
$N'\in\CO_k'$ (see Subsection~\ref{COO} for the notation) is an extension of
$M'$ by $M'$, then $\Upsilon_{M'}(N')\subset \fhh^*_0$.

\subsubsection{}
\begin{thm}{thmmuW}
Fix $\lambda\in\fhh^*_k$ such that $H^0(L(\lambda))\not=0$.
Introduce the linear map
$$\begin{array}{l}
\phi_{\lambda}: \fhh^*_0\to \fh^*_{\cW}:\ \
\phi_{\lambda}(\mu)|_{\fh^f}:=\mu|_{\fh^f}, \ \
\langle \phi_{\lambda}(\mu), L_0\rangle:=
\frac{(\mu,\lambda+\hat{\rho})}{k+h^{\vee}}-
\langle \mu,x+D\rangle.
\end{array}$$
 Then

(i) If $N\in\CO_k'$ is a self-extension of the $\hat{\fg}$-module
$L(\lambda)$, then
$$\ \ \ \
\Upsilon_{H^0(L(\lambda))} (H^0(N))=
\phi_{\lambda}\bigl(\Upsilon_{L(\lambda)}(N)\bigr).$$

In particular,
 $\ \phi_{\lambda}\bigl(\im\Upsilon_{L(\lambda)}\cap \fhh^*_0\bigr)\subset
\im\Upsilon_{H^0(L(\lambda))}$.

(ii) If $(\lambda,\alpha_0)\not\in\mathbb{Z}$, then
$\ \phi_{\lambda}\bigl(\im\Upsilon_{L(\lambda)}\cap \fhh^*_0\bigr)
=\im\Upsilon_{H^0(L(\lambda))}$.

(iii)
$\begin{array}{lll}
\text{ If }(\lambda+\hat{\rho},\alpha_0)\not=0, & \text{ then } &
\ker\phi_{\lambda}=\mathbb{C}\delta,\ \im\phi_{\lambda}=\fh^*_{\cW};\\
\text{ if }(\lambda+\hat{\rho},\alpha_0)=0, & \text{ then } &
\ker\phi_{\lambda}=\mathbb{C}\delta+\mathbb{C}\alpha_0,\
\dim\im\phi_{\lambda}=\dim\fh^*_{\cW}-1.
\end{array}$
\end{thm}

\begin{rem}{}
The assumption $(\lambda,\alpha_0)\not\in\mathbb{Z}$
in (ii) gives $H^0(L)\not=0$ for
any subquotient $L$ of $M(\lambda)$. \Exa{exaW} shows that
$\im\phi_{\lambda}\not=\im\Upsilon_{H^0(L(\lambda))}$ in general.
\end{rem}

\begin{proof}
Set
$$M:=M(\lambda),\ L:=L(\lambda), \ \bM:=H^0(M(\lambda)),
\ \bL:=H^0(L(\lambda)),\ \bN:=H^0(N).$$
Consider an exact sequence $0\to L\to N\to L\to 0$
and let $\Upsilon_{L}:\Ext^1(L,L)\to\fhh^*$
maps  this  exact sequence to $\mu$.
Let $0\to\bL\to \bN\to \bL\to 0$ be the image of this sequence under $H^0$.
For (i) let us verify that
\begin{equation}\label{upsups}
\Upsilon_{\bL}\bigl(0\to\bL\to \bN\to \bL\to 0\bigr)=\phi_{\lambda}(\mu)
\end{equation}

The space $N_{\lambda}$ has a basis $v,v'$ such that
$$hv'=\langle \lambda,h\rangle v', \ \ hv=\langle \lambda,h\rangle v+
\langle \mu,h\rangle v'\ \text{ for any } h\in\fhh.$$
By the assumption $(\lambda,\alpha_0)\not\in\mathbb{Z}_{\geq 0}$.
Let $\vac$ be the vacuum vector of
$F^{ch}\otimes F^{ne}$, see~\cite{KWdet} for notation.
Using the fact that $\Omega(N)\subset \Omega(M(\lambda))$,
it is easy to show by an explicit computation or to deduce
from~\cite{Ar}, 4.7.1, 4.8.1, that the images of $v\otimes\vac,\
v'\otimes\vac$ lie in $\bN$ and are linearly independent.
Since $\bN/\bL\cong \bL$ one has $\dim\bN_{\lambda_{\cW}}=2$
and thus $v\otimes\vac,\ v'\otimes\vac$ form a basis of
$\bN_{\lambda_{\cW}}$.

The explicit formulas for $J^a$ and $L(z)$ given in~\cite{KWdet} imply that
for $h\in\fh^f$ one has
$$\begin{array}{l}
J^{h}(0)(v'\otimes\vac)= \langle \lambda,h\rangle v'\otimes\vac,\ \
J^{h}(0)(v\otimes\vac)=\langle \lambda,h\rangle v\otimes\vac
+\langle \mu,h\rangle v'\otimes \vac;\\
L_0(v''\otimes \vac)=\bigl(\frac{\hat{\Omega}}{2(k+h^{\vee})}-
(x+D)\bigr)v''\otimes \vac, \text{ for any }
v''\in N_{\lambda},
\end{array}$$
where $\hat{\Omega}$ is the Casimir operator of $\fhg$.
One has
$$\hat{\Omega}v'=(\lambda,\lambda+2\hat{\rho})v',\ \
\hat{\Omega}v=(\lambda,\lambda+2\hat{\rho})v+
2(\mu,\lambda+\hat{\rho})v'$$
and thus
$$\begin{array}{l}
L_0(v'\otimes \vac)=av'\otimes \vac,\
L_0 (v\otimes \vac)=av\otimes \vac+bv'\otimes \vac,
\text{ where }\\
a:=\frac{(\lambda,\lambda+2\hat{\rho})}{2(k+h^{\vee})}
-\langle \lambda, x+D\rangle,\ \
b:=\frac{(\lambda+\hat{\rho},\mu)}{k+h^{\vee}}-\langle \mu, x+D\rangle.
\end{array}$$
This establishes~(\ref{upsups}) and proves (i).

For (ii) fix $\lambda\in\fhh^*_k$ such that
$(\lambda,\alpha_0)\not\in\mathbb{Z}$.
Observe that for any subquotient $L'$ of  $M^{\#}$ one has $H^0(L')\not=0$
since, by~\ref{Wlambda}, one has $(w.\lambda,\alpha_0)\not\in\mathbb{Z}$
for any $w\in W(\lambda)$.

Suppose that $\mu\not\in\im\Upsilon_L$ and
$\phi_{\lambda}(\mu)\in \im\Upsilon_{\bL}$. Let $N$ be an extension
of $M$ by $M$ such that $\Upsilon_{M}(N)=\mathbb{C}\mu$.
Set $\bN:=H^0(N)$. By (i), $\Upsilon_{\bM}(\bN)=\mathbb{C}\phi_{\lambda}(\mu)$.
By the property $(\Upsilon 3$), the assumption
$\phi_{\lambda}(\mu)\in \im\Upsilon_{\bL}$ implies the existence
of $\bN'\subset \bN$ such that  $\bN/\bN'$ is an extension
of $\bL$ by $\bL$. Notice that $\bN'_{\lambda_{\cW}}=0$ since
$\dim\bN_{\lambda_{\cW}}=2$  and $\dim\bL_{\lambda_{\cW}}=1$.

In the light of~\Lem{lemcri}, $N$ has a subquotient $E\not\cong L$
which  is isomorphic to a submodule of $M^{\#}$.
Clearly, $H^0(E)$ is isomorphic to a submodule
of $H^0(M^{\#})$ so, by~\ref{redH}, any submodule
of $H^0(E)$ contains a vector of weight $\lambda_{\cW}$.
By the above observation, $H^0(E)\not\cong \bL$. Since $H^0$ is exact, $\bN$
has submodules $\bN_1\subset \bN_2$ such that $\bN_2/\bN_1\cong H^0(E)$.

Since $\bN'_{\lambda_{\cW}}=0$,
each $v\in \bN_2\cap \bN'$ generates a submodule which intersects trivially
$\bN_{\lambda_{\cW}}$. Thus the image of $\bN_2\cap \bN'$
in $\bN_2/\bN_1\cong H^0(E)$ is zero. Therefore $\bN_2\cap \bN'\subset \bN_1$
and so $\bN_2/\bN_1$ is a quotient
of ${\bN}_2/({\bN}_2\cap {\bN}'')\subset {\bN}/{\bN'}$. Hence
$H^0(E)$   is a subquotient of ${\bN}/{\bN'}$.
However, since ${\bN}/{\bN'}$
is an extension of ${\bL}$ by ${\bL}$, the only subquotients of
${\bN}/{\bN'}$ are $\bL$ and ${\bN}/{\bN'}$ itself.
By above, $H^0(E)\not\cong \bL$. By~(\ref{chuhk})
$\dim H^0(M^{\#})_{\lambda_{\cW}}=1$ so
$\dim H^0(E)_{\lambda_{\cW}}=1=\dim \bL_{\lambda_{\cW}}$
and thus $H^0(E)$ is not an extension of ${\bL}$ by ${\bL}$,
a contradiction. This establishes (ii). The proof of
(iii) is straightforward.
\end{proof}

\subsection{Admissible weights}\label{admisW}
In this subsection we study the admissible weights in $\fh^*_{\cW}$.

Denote by $\varphi_{\cW}:\fhh^*_k\to \fh^*_{\cW}$ the map given by
$\lambda\mapsto \lambda_{\cW}$. Recall that $\varphi_{\cW}$
is surjective and that $\varphi_{\cW}(\lambda)=\varphi_{\cW}(\lambda')$
iff $\lambda'\in\{\lambda+\mathbb{C}\delta, s_0.\lambda+\mathbb{C}\delta\}$,
by~\Lem{lamW}. Let $\hat{\Delta}$ be the set of roots of $\fhg$ and
$wAdm_k$ be the set of weakly admissible weights  in $\fhh^*_k$.
By~\Cor{corLie1}, $k\text{-}Adm\subset \{\lambda\in wAdm_k|\
\mathbb{C} \hat{\Delta}(\lambda)+\mathbb{C}\delta=\mathbb{C}\hat{\Delta}\}$;
moreover,  by~\Rem{Remkrat}, the  $k$-admissible weights for rational $k$ are
the  rational weakly admissible weights in $\fhh^*_k$.
In Subsection~\ref{pfpropadmW} we prove the following proposition.

\subsubsection{}
\begin{prop}{propadmW}
(i) If $(\lambda,\alpha_0)\not\in\mathbb{Z}$, then
$\lambda_{\cW}\in Adm_{\cW}$ iff $\lambda\in k\text{-}Adm$.

(ii) $\ \ Adm_{\cW}\subset \varphi_{\cW}\bigl(\{\lambda\in wAdm_k|\
\mathbb{C} \Delta(\lambda)+\mathbb{C}\delta=\mathbb{C}\hat{\Delta}\}
\bigr)$.

(iii) If $k$ is rational, then $\ Adm_{\cW}\subset \varphi_{\cW}\bigl(
\{\lambda\in wAdm_k\ \&\ \lambda\text{ is rational}\}\bigr)$.

(iv) The set $Adm_{\cW}$ is empty if $k+h^{\vee}\in\mathbb{Q}_{<0}$
and is finite if $k+h^{\vee}\in\mathbb{Q}_{>0}$.
\end{prop}

\subsubsection{}\label{polyhW}
Let  $k$ be rational. Retain notation of~\ref{polyh} and recall that the
set of rational weakly admissible weights of level $k$
is the set of integral points of
a finite union of polyhedra $X_{\Gamma}$. The image of
each polyhedron in $\fhh^*/\mathbb{C}\delta$ is finite.
The set of $k$-admissible weights is a subset of
integral points of these polyhedra; it contains the set of
 interior integral points
of each polyhedra. We see that $Adm_{\cW}$ lies in the image (under
$\varphi_{\cW}$) of the set of integral points of the polyhedra $X_{\Gamma}$.

Let $\Gamma$ be such that $\alpha_0\not\in\Gamma$. Set
$s_0\Gamma:=\{(s_0\alpha)^{\vee}|\ \alpha^{\vee}\in \Gamma\}$. Then
$X_{s_0\Gamma}=s_0X_{\Gamma}$ and so
$\varphi_{\cW}(X_{\Gamma})=\varphi_{\cW}(X_{s_0\Gamma})$; each point from
$\varphi_{\cW}(X_{\Gamma})\in\fh^*_{\cW}$ has exactly two preimages
in $\fhh^*_k/\mathbb{C}\delta$ (one in $X_{\Gamma}$ and another one
in $X_{s_0\Gamma}$). By~\Prop{propadmW} (i),
admissible weights in $\varphi_{\cW}(X_{\Gamma})$ are the images
of $k$-admissible weights in $X_{\Gamma}$.

Let $\alpha_0\in\Gamma$.  In~\Prop{extW0} we will show that
$Adm_{\cW}$ contains the image of the set of interior
integral points of the polyhedron $X(\Gamma)$ and
the image of the set of interior points of
the face $(\lambda+\hat{\rho},\alpha_0)=0$ of the polyhedron $X_{\Gamma}$.

\subsubsection{}
\begin{cor}{corvacW}
The irreducible vacuum module over the vertex algebra
$\cW=\cW^k(\fg,e_{-\theta})$
is admissible for the following values of $k$: $k\not\in\mathbb{Q}$,
or $k+h^{\vee}=\frac{p}{q}>0, \gcd(p,q)=1$,
where $p\geq h^{\vee}-1$ and $\gcd(q,l)=1$, or
$p\geq h-1$ and $\gcd(q,l)=l$. The irreducible vacuum module is
not admissible for other values of $k$, except for, possibly, the case
$k=-2$ and $\fg\not=C_n$.
\end{cor}
\begin{proof}
The module $\bL(0)$ is the irreducible vacuum module and
$0=\varphi_{\cW}(k\Lambda_0)$.

Consider the case when $k$ is an integer.
If $k+h^{\vee}<0$, then $Adm_{\cW}$ is empty and so $0\in\fh^*_{\cW}$
is not admissible. One has
$\langle k\Lambda_0+\hat{\rho},\alpha_0^{\vee}\rangle=k+1$,
so $k\Lambda_0$ is weakly admissible iff $k\geq -1$.
In the light of~\Prop{extW0}, $0=\varphi_{\cW}(k\Lambda_0)$ is admissible
for $k\geq -1$.

Consider the case $k\in\mathbb{Z}_{\leq -2}$. Assume that
 $0=\varphi_{\cW}(k\Lambda_0)$ is admissible. Then, by~\Prop{propadmW} (ii),
$k\Lambda_0$ or $s_0.k\Lambda_0$ is weakly admissible.
Since $k\Lambda_0$ is not weakly admissible,
 $s_0.k\Lambda_0$ is weakly admissible, that is
for each $\beta\in \Pi$ one has
$$0\leq \langle s_0.k\Lambda_0+\hat{\rho},\beta^{\vee}\rangle=
1-(k+1)\langle  \alpha_0,\beta^{\vee}\rangle.$$
For $k<-2$ the above inequality does not hold if
$\langle  \alpha_0,\beta^{\vee}\rangle\not=0$, so
$0=\varphi_{\cW}(k\Lambda_0)$ is not weakly admissible. Hence
$0=\varphi_{\cW}(k\Lambda_0)$ is not admissible for $k<-2$.
For $k=-2$ the above inequalities hold iff
$\langle  \alpha_0,\beta^{\vee}\rangle\geq -1$
for all $\beta\in \Pi$ and this holds
iff the $0$th column of Cartan matrix of $\hat{\fg}$ contains
only $0,-1,2$; this means that the corresponding Dynkin
diagram does not have arrows going from the $0$th vertex.
Thus $s_0.(-2\Lambda_0)$ is weakly admissible iff $\fg$ is not of the type
$C_n$. Hence $0=\varphi_{\cW}(-2\Lambda_0)$ is not admissible for $C_n$.

If $k\not\in\mathbb{Z}$, then $0\in\fh^*_{\cW}$ is  admissible
iff $k\Lambda_0$ is $k$-admissible, that is iff $k$ is irrational,
or $k+h^{\vee}=\frac{p}{q}\in\mathbb{Q}_{>0}\setminus\mathbb{Z}$,
where $p\geq h^{\vee}-1$ and $\gcd(q,l)=1$, or
$p\geq h-1$ and $\gcd(q,l)=l$, see~\Cor{corvac1}. Note that the case
$q=1,\ p\geq h^{\vee}-1$ produces $k\in\mathbb{Z}_{\geq -1}$.
\end{proof}

\subsubsection{}
Set
$$\begin{array}{l}
Y:=\{\lambda\in\fhh^*_k|\ (\lambda,\alpha_0)\in\mathbb{Z}\},\ \
Y_0:=\{\lambda\in\fhh^*_k|\ (\lambda,\alpha_0)=0\},\\
\ol{Y}:=\fhh^*_k\setminus Y,\ \
\ol{\varphi_{\cW}(Y)}:=\fh^*_{\cW}\setminus\varphi_{\cW}(Y),
\end{array}$$
and observe that
$$\ol{\varphi_{\cW}(Y)}=\varphi_{\cW}(\ol{Y}),\ \
\varphi_{\cW}^{-1}\bigl(\varphi_{\cW}(Y)\bigr)=Y,\ \
\varphi_{\cW}^{-1}\bigl(\varphi_{\cW}(Y_0)\bigr)=Y_0.$$

We start the proof of~\Prop{propadmW} from the following description of
the set of weakly admissible weights in $\fh^*_{\cW}$, which we denote by
$wAdm_{\cW}$.

\subsubsection{}
\begin{lem}{lemwadm}
One has
$$wAdm_{\cW}=\varphi_{\cW}(wAdm_k),\ \ \
\varphi_{\cW}^{-1}\bigl(wAdm_{\cW}\cap \varphi_{\cW}(\ol{Y})\bigr)\subset
wAdm_k.$$
\end{lem}
\begin{proof}
Consider the case when  $\lambda\in \ol{Y}$.
Let us show that $\lambda\in\fhh^*_k$ is weakly admissible
iff $\lambda_{\cW}\in\fh_{\cW}^*$ is weakly admissible.

Take $\lambda\in \ol{Y}$, which is not weakly admissible.
Then $\Hom_{\fg}(M(\lambda),M(\lambda+m\beta))\not=0$ for some
real root $\beta$. Therefore
$\Hom_{\cW}(\bM(\lambda_{\cW}),\bM(\lambda+m\beta)_{\cW})\not=0$.
Since $(\lambda,\alpha_0)\not\in\mathbb{Z}$
one has $m\beta, m\beta+(\lambda-s_0.\lambda)\not\in\mathbb{C}\delta$ so
$\lambda_{\cW}\not=(\lambda+m\beta)_{\cW}$.
Hence $\lambda_{\cW}$ is not weakly admissible.

Assume that $\lambda_{\cW}$ is not weakly admissible that is
$[\bM(\lambda'_{\cW}):\bL(\lambda_{\cW})]\not=0$
for some $\lambda'$ with $\lambda'_{\cW}\not=\lambda_{\cW}$. Recall that
$\bM(\lambda'_{\cW})=H^0(M(\lambda'))=H^0(M(s_0.\lambda'))$.
From (H4) we obtain
$$\begin{array}{l}
\sum_{a\in\mathbb{C}} [M(\lambda'):L(\lambda-a\delta)]+
[M(\lambda'):L(s_0.\lambda-a\delta)]\not=0,\\
\sum_{a\in\mathbb{C}} [M(s_0.\lambda'):L(\lambda-a\delta)]+
[M(s_0.\lambda'):L(s_0.\lambda-a\delta)]\not=0.
\end{array}$$
If $[M(\lambda'):L(s_0.\lambda-a_1\delta)]\not=0$
and $[M(s_0.\lambda'):L(s_0.\lambda-a_2\delta)]\not=0$
for some $a_1,a_2$, then $\lambda'-s_0.\lambda,\ s_0.\lambda'-s_0.\lambda
\in\hat{Q}+\mathbb{C}\delta$ so $\lambda-s_0.\lambda\in \hat{Q}+
\mathbb{C}\delta$, which contradicts to
$(\lambda,\alpha_0)\not\in\mathbb{Z}$. As a result,
$[M(\lambda'):L(\lambda-a\delta)]\not=0$
or $[M(s_0.\lambda'):L(\lambda-a\delta)]\not=0$
for some $a$, then
$[M(\lambda''+a\delta):L(\lambda)]\not=0$ for
$\lambda''\in\{\lambda',s_0.\lambda'\}$.
Notice that $(\lambda''+a\delta)_{\cW}=\lambda'_{\cW}\not=\lambda_{\cW}$
so $\lambda''+a\delta\not=\lambda$.
Thus $\lambda$ is not weakly admissible.

Consider the case when $\lambda\in Y$.
Take a weakly admissible $\lambda\in \fhh^*_k$ with
$\alpha_0\in\Delta(\lambda)$.
We claim that $\lambda_{\cW}$ is weakly admissible.
Indeed, otherwise $[\bM(\lambda'_{\cW}):\bL(\lambda_{\cW})]\not=0$
for some $\lambda'$ with $\lambda'_{\cW}\not=\lambda_{\cW}$.
Since $H^0(L(\lambda+a\delta))=0$ for all $a\in\mathbb{C}$,
(H4) give
$[M(\lambda'):L(s_0.\lambda+a\delta)]>0$ for some $a\in\mathbb{C}$
that is $[M(\lambda'-a\delta):L(s_0.\lambda)]>0$.
Since $\lambda$ is weakly admissible, it is maximal in its
$W(\lambda).$-orbit so $\lambda'-a\delta=\lambda$ or
$\lambda'-a\delta=s_0.\lambda$. This contradicts
to the assumption $\lambda'_{\cW}\not=\lambda_{\cW}$.

Now suppose that $\lambda_{\cW}$ is weakly admissible and
the orbit $W(\lambda).\lambda$ has a maximal element
$\lambda'$. Then $\Hom_{\fg}(M(\lambda),M(\lambda'))\not=0$ so
$\Hom_{\cW}(\bM(\lambda_{\cW}),\bM(\lambda'_{\cW}))\not=0$
that is $\lambda_{\cW}=\lambda'_{\cW}$, because
$\lambda_{\cW}$ is weakly admissible. Since
$\lambda'$ is maximal in its $W(\lambda).\lambda$-orbit, $\lambda'$
is weakly admissible and so the preimage of $\lambda_{\cW}$
in $\fhh^*_k$ contains a weakly admissible element.

Finally, suppose $\lambda_{\cW}$ is weakly admissible
and the orbit $W(\lambda).\lambda$ does not have
a maximal element. Since $\lambda_{\cW}=(s_0.\lambda)_{\cW}$, we may
(and will) assume that $s_0.\lambda\leq \lambda$. By the assumption,
$s_{\beta}.\lambda>\lambda$ for some real root $\beta$ so
$\Hom_{\fg}(M(\lambda),M(s_{\beta}.\lambda))\not=0$.
Therefore $\Hom_{\cW}(\bM(\lambda_{\cW}),\bM((s_{\beta}.\lambda)_{\cW}))\not=0$
so $\lambda_{\cW}=(s_{\beta}.\lambda)_{\cW}$, because
$\lambda_{\cW}$ is weakly admissible.
This implies $s_{\beta}.\lambda-\lambda\in\mathbb{C}\delta$,
which is impossible, or
$s_{\beta}.\lambda-s_0.\lambda\in\mathbb{C}\delta$.
Since $(s_{\beta}.\lambda+\hat{\rho},s_{\beta}.\lambda+\hat{\rho})=
(s_0.\lambda+\hat{\rho},s_0.\lambda+\hat{\rho})$ and
$(s_{\beta}.\lambda+\hat{\rho},\delta)\not=0$, the formula
$s_{\beta}.\lambda-s_0.\lambda\in\mathbb{C}\delta$ forces
$s_{\beta}.\lambda=s_0.\lambda$, which contradicts to
$s_{\beta}.\lambda>\lambda\geq s_0.\lambda$.
\end{proof}

\subsubsection{}
Retain notation of~\ref{CC'}.

\begin{lem}{lemWY}
$$\begin{array}{l}
Adm_{\cW}\cap \varphi_{\cW}(Y\setminus Y_0)\subset \varphi_{\cW}\bigl(\{
\lambda\in Y\setminus Y_0\cap wAdm_k|\
\mathbb{C}\,C(\lambda)+\mathbb{C}\delta=\mathbb{C}\hat{\Delta}\}
\bigr),\\
Adm_{\cW}\cap \varphi_{\cW}(Y_0)\subset \varphi_{\cW}\bigl(\{
\lambda\in Y_0\cap wAdm_k|\ \mathbb{C}\,C(\lambda)+\mathbb{C}\delta
\supset(\mathbb{C}\hat{\Delta}\cap\alpha_0^{\perp})\}
\bigr).
\end{array}
$$
\end{lem}
\begin{proof}
Take $\nu\in Adm_{\cW}\cap \varphi_{\cW}(Y)$.
By~\Lem{lemwadm}, $\nu=\lambda_{\cW}$ for a weakly admissible
$\lambda\in Y$. Then
$(\lambda+\hat{\rho},\alpha_0)\in\mathbb{Z}_{\geq 0}$ and so
$\bL(\lambda_{\cW})=H^0(L(s_0.\lambda))$.

Consider the case when $\lambda\in Y\setminus Y_0$ that is
$(\lambda+\hat{\rho},\alpha_0)\in\mathbb{Z}_{>0}$.
Since $H^0(L(s_0.\lambda))$ is admissible, this module
does not admits self-extensions and so
$\im\Upsilon_{L(s_0.\lambda)}\cap \fhh^*_0=\mathbb{C}\delta$,
by~\Thm{thmmuW}.
By~\Prop{propalmin}, $\im\Upsilon_{L(s_0.\lambda)}$
contain $C(s_0.\lambda)^{\perp}$ so $C(s_0.\lambda)^{\perp}\cap \fhh^*_0=
\mathbb{C}\delta$
that is
$\mathbb{C}\,C(s_0.\lambda)+\mathbb{C}\delta=\mathbb{C}\hat{\Delta}$.
It is easy to see that $C(s_0.\lambda')\setminus\{\alpha_0\}
=s_0(C(\lambda')\setminus\{\alpha_0\})$ for any $\lambda'$.
The condition $(\lambda+\hat{\rho},\alpha_0)\in\mathbb{Z}_{>0}$
gives $\alpha_0\in C(\lambda),
\alpha_0\not\in C(s_0.\lambda)$ that is
$C(\lambda)=s_0(C(s_0.\lambda))\cup\{\alpha_0\}$.
Since $\mathbb{C}\,C(s_0.\lambda)+\mathbb{C}\delta=\mathbb{C}\hat{\Delta}$
we get $\mathbb{C}\,C(\lambda)+\mathbb{C}\delta=\mathbb{C}\hat{\Delta}$
as required.

Consider the case when $\lambda\in Y_0$ that is
$s_0.\lambda=\lambda$. Using~\Thm{thmmuW} we obtain
$\im\Upsilon_{L(\lambda)}\cap \fhh^*_0\subset\spn\{\delta,
\alpha_0\}$. By~\Prop{propalmin}, $\im\Upsilon_{L(\lambda)}$
contains $C(\lambda)^{\perp}$ so $\im\Upsilon_{L(\lambda)}\cap \fhh^*_0$
contains $\bigl(C(\lambda)\cup\{\delta\}\bigr)^{\perp}$. Thus
$\mathbb{C} C(\lambda)+\mathbb{C}\delta\supset
(\mathbb{C}\hat{\Delta}\cap\alpha_0^{\perp})$ as required.
\end{proof}

\subsubsection{Proof of~\Prop{propadmW}}\label{pfpropadmW}
We rewrite (i) as follows
\begin{equation}\label{propad1}
\lambda_{\cW}\in Adm_{\cW}\cap \varphi_{\cW}(\ol{Y})\ \Longleftrightarrow\
\lambda\in k\text{-}Adm\cap \ol{Y}.
\end{equation}
Take $\lambda\in \ol{Y}$.
Combining~\Lem{lemwadm} and~\Thm{thmmuW} we conclude that
$\lambda_{\cW}\in Adm_{\cW}$ iff $\lambda$ is weakly admissible and
$\im\Upsilon_{L(\lambda)}\cap \fhh^*_0=\mathbb{C}\delta$.
Thus $\lambda_{\cW}\in Adm_{\cW}$ iff $\lambda\in k\text{-}Adm$.
This establishes~(\ref{propad1}) and (i).

Rewrite (ii) in the form
\begin{equation}\label{QQ}
\lambda'\in Adm_{\cW}\ \Longrightarrow\
\exists \lambda\in\varphi_{\cW}^{-1}(\lambda')\ \text{ s.t. }\ \lambda\in
wAdm_k\ \&\
\mathbb{C} \Delta(\lambda)+\mathbb{C}\delta=\mathbb{C}\hat{\Delta}.
\end{equation}
If $\lambda'\in Adm_{\cW}\setminus\varphi_{\cW}(Y)$ the existence of
$\lambda$  follows from (i)
and~\Cor{corLie1}; if
$\lambda'\in Adm_{\cW}\cap\varphi_{\cW}(Y\setminus Y_0)$
this follows from~\Lem{lemWY} (because $C(\lambda)\subset\Delta(\lambda)$).
Take $\lambda'\in Adm_{\cW}\cap\varphi_{\cW}(Y_0)$. By~\Lem{lemWY},
$\lambda'=\varphi_{\cW}(\lambda)$, where
$\lambda\in Y_0\cap wAdm_k$ is such that
$\mathbb{C}\Delta(\lambda)+\mathbb{C}\delta$ contains
$\mathbb{C}\hat{\Delta}\cap\alpha_0^{\perp}$. Note that for
$\lambda\in Y_0$ one has
$\alpha_0\in\Delta(\lambda),\alpha_0\not\in C(\lambda)$.
Thus $\mathbb{C}\Delta(\lambda)+\mathbb{C}\delta$ contains
$(\mathbb{C}\hat{\Delta}\cap\alpha_0^{\perp})+\mathbb{C}\alpha_0=
\mathbb{C}\hat{\Delta}$. This establishes (ii).

Finally, (iii) follows from (ii) and~\Rem{Remkrat}, and (iv)
follows from (ii),~\Rem{Remkrat}, and~\Cor{corkh}.
\qed

This completes the proof of~\Thm{thmadmW}. Now let us continue
our study of admissible weights in $\fh^*_{\cW}$.

\subsubsection{}
\begin{prop}{extW0}
If a rational weakly admissible $\lambda\in\fhh^*_k$ is such that
$\langle\lambda+\hat{\rho},\beta^{\vee}\rangle>0$ for each
$\beta\in\Pi(\lambda)\setminus\{\alpha_0\}$,
then $\lambda_{\cW}$ is admissible.
\end{prop}
\begin{proof}
Set
$B:=\hat{\Delta}^{re}_+\setminus\{\alpha_0\}\times
\mathbb{Z}_{>0}$,
and let $B_{>i}\subset B$ (resp., $B_{\geq i}$)
be the set of pairs $(\alpha,m)\in B$ such that
$\langle\alpha,x\rangle>i$ (resp., $\langle\alpha,x\rangle\geq i$).
The determinant formula in~\cite{KWdet} (Thm. 7.2) for $\fg$ being
Lie algebra  can be rewritten
in the following form:
for $\nu\in Q_+$ one has
$$\det\! _{\nu}(\lambda_{\cW})=(k+h^{\vee})^M\prod_{(\alpha,m)\in B_{\geq 0}}
\phi_{\alpha,m}(\lambda_{\cW})^
{\dim \cM((s_{\alpha}.\lambda)_{\cW})_{\lambda_{\cW}-\nu}}
\!\!\prod_{(\alpha,m)\in B_{>0}}\phi'_{\alpha,m}(\lambda_{\cW})^
{\dim \cM((s_{\alpha}s_0.\lambda)_{\cW})_{\lambda_{\cW}-\nu}},$$
if $s_0.\lambda\not=\lambda$, and
$$\det\! _{\nu}(\lambda_{\cW})=(k+h^{\vee})^M\prod_{(\alpha,m)\in B_{\geq 0}}
\phi_{\alpha,m}(\lambda_{\cW})^
{\dim \cM((s_{\alpha}.\lambda)_{\cW})_{\lambda_{\cW}-\nu}},$$
if $s_0.\lambda=\lambda$, where $M\geq 0$, and
$\phi_{\alpha,m}(\lambda_{\cW})=0$ iff
$(\lambda+\hat{\rho},\alpha)=\frac{m}{2}(\alpha,\alpha)$,
$\phi'_{\alpha,m}(\lambda_{\cW})=0$ iff
$(s_0.\lambda+\hat{\rho},\alpha)=\frac{m}{2}(\alpha,\alpha)$.

One has $\phi'_{\alpha,m}(\lambda_{\cW})=0$ iff
$(\lambda+\hat{\rho},s_0\alpha)=\frac{m}{2}(s_0\alpha,s_0\alpha)$ that is
$\phi_{s_0\alpha,m}(\lambda_{\cW})=0$. Furthermore,
$(s_{\alpha}s_0.\lambda)_{\cW}=(s_0s_{\alpha}s_0.\lambda)_{\cW}=
(s_{s_0\alpha}\lambda)_{\cW}$,
and $(\alpha,m)\in B_{>0}$ iff $(s_0\alpha,m)\in B_{<0}$.
Hence the determinant formula
for $s_0.\lambda\not=\lambda$ can be rewritten as
$$\det\! _{\nu}(\lambda_{\cW})=(k+h^{\vee})^M
\prod_{\alpha\in \hat{\Delta}^{re}_+\setminus\{\alpha_0\}}\prod_{m=1}^{\infty}
\phi_{\alpha,m}(\lambda_{\cW})^
{\dim \cM((s_{\alpha}.\lambda)_{\cW})_{\lambda_{\cW}-\nu}}.$$

Let us make the following identifications
$$\fh^*=\{\xi\in\fhh^*|\ (\xi,\delta)=(\xi,\Lambda_0)=0\},\ \
(\fh^f)^*=\{\xi\in\fhh^*|\ (\xi,\delta)=(\xi,\Lambda_0)=(\xi,\theta)=0\}.$$
For $\lambda\in\fhh^*_k$ we write
$$\lambda=k\Lambda_0+\langle \lambda,x\rangle \theta+\lambda^{\#}.$$
From~\cite{KWdet} (Thm. 7.2) one has
\begin{equation}\label{phit}
\begin{array}{ll}
\phi_{\alpha,m}(\lambda_{\cW}+t\mu)=\phi_{\alpha,m}(\lambda_{\cW})+t
(\mu|_{\fh^f},\alpha), &\text{ if }\langle\alpha,x\rangle=0,\\
\phi_{\alpha,m}(\lambda_{\cW}+t\mu)=\phi_{\alpha,m}(\lambda_{\cW})+t
a(\mu,\lambda,\alpha)\mod t^2, &\text{ if }
\langle\alpha,x\rangle=\pm\frac{1}{2},\pm 1,\\
a(\mu,\lambda,\alpha):=
\langle \mu,L_0\rangle+\frac{((\lambda+\rho,\theta)
+k+h^{\vee})(\mu|_{\fh^f},\alpha)+(\mu|_{\fh^f},\lambda^{\#}+\rho^{\#})}
{(k+h^{\vee})}, &
\end{array}\end{equation}
where we view $\mu|_{\fh^f}\in (\fh^f)^*$ as a vector in $\fhh^*$
using the above identification.

Let $\lambda\in\fhh^*_k$ satisfy the assumptions  of our proposition.
If $\alpha_0\not\in \Pi(\lambda)$, then $\lambda$ is weakly admissible,
rational and regular, so $\lambda$ is KW-admissible and
$(\lambda,\alpha_0)\not\in \mathbb{Z}$. In this case
$\lambda_{\cW}$ is admissible by~\Prop{propadmW}.

From now on we assume that $\alpha_0\in \Pi(\lambda)$. By~\Lem{lemwadm},
$\lambda_{\cW}$ is weakly admissible, so it is enough to verify that
$\Ext^1(\bL(\lambda_{\cW}),\bL(\lambda_{\cW})=0$.

For each $\beta\in\Pi(\lambda)\setminus\{\alpha_0\}$ one has
$(\beta,\alpha_0)=-(\beta,\theta)\leq 0$ (by~\ref{Pilambda})
and so $\langle\beta,x\rangle\geq 0$ that is
$\langle\beta,x\rangle=0,\frac{1}{2},1$, by~\ref{notW}.
Write
$$\Pi(\lambda)=\{\alpha_0\}\cup\Pi(\lambda)_0\cup\Pi(\lambda)_{\frac{1}{2}}\cup
\Pi(\lambda)_1,\text{ where } \Pi(\lambda)_j=\{\beta\in\Pi(\lambda)|\
\langle\beta,x\rangle=j\}.$$

We will use the following fact (see~\cite{KT}):
if  a non-critical weight
$\lambda'\in'\fhh^*$ is maximal in its $W.$-orbit,
then $\Stab_{W(\lambda')}(\lambda'+\hat{\rho})=
\Stab_{\hat{W}}(\lambda'+\hat{\rho})$
is a finite Coxeter group
generated by the reflections $s_{\alpha}, \alpha\in \Pi(\lambda')$ such that
$(\lambda'+\hat{\rho},\alpha)=0$, and
$[M(y.\lambda'):L(w.\lambda')]\not=0$ implies the existence of
$z\in \Stab_{W}(\lambda'+\hat{\rho})$ such that $y\leq_{\lambda'} wz$, where
$\leq_{\lambda'}$ stands for the Bruhat order in the Coxeter
group $W(\lambda')$.

Assume that
\begin{equation}\label{pila}
\beta\in\Pi(\lambda)\setminus\{\alpha_0\},\
\alpha\in \hat{\Delta}^{re}_+\setminus\{\alpha_0\}\ \text{ s.t. }
[\bM((s_{\alpha}.\lambda)_{\cW}):\bL((s_{\beta}.\lambda)_{\cW})]\not=0.
\end{equation}
Then, by~\Cor{corlamW}, $[M(s_{\alpha}.\lambda):L(w.\lambda)]\not=0$, where
$w\in W(\lambda)$ is such that $(w.\lambda)_{\cW}=(s_{\beta}.\lambda)_{\cW}$,
that is $w.\lambda=s_{\beta}.\lambda$ or $w.\lambda=s_0s_{\beta}.\lambda$.
By above, $s_{\alpha}\leq_{\lambda} wz$, where
$w\in\{s_{\beta},s_0s_{\beta}\}$,
$z\in \Stab_{W(\lambda)}(\lambda+\hat{\rho})$,
and $\Stab_{W(\lambda)}(\lambda+\hat{\rho})$ is generated by $s_0$
if $(\lambda+\hat{\rho},\alpha_0)=0$ and is trivial otherwise.
Since $\beta,\alpha_0\in\Pi(\lambda)$,
$s_{\alpha}\leq_{\lambda} s_{\beta}$ forces $\alpha=\beta$;
$s_{\alpha}\leq_{\lambda} s_{\beta}s_0$ or
$s_{\alpha}\leq_{\lambda} s_0s_{\beta}$ forces $\alpha=\beta,\alpha_0$;
$s_{\alpha}\leq_{\lambda} s_0s_{\beta}s_0$ forces
$\alpha=\beta,\alpha_0,s_0\beta$. Therefore~(\ref{pila}) implies
$\alpha=\beta$ if $(\lambda+\hat{\rho},\alpha_0)\not=0$,
 and $\alpha\in\{\beta,s_0\beta\}$ if $(\lambda+\hat{\rho},\alpha_0)=0$.
Since $\langle \beta,x\rangle\geq 0$, one has
$\langle s_0\beta,x\rangle\geq 0$ iff $s_0\beta=\beta$.

Take $\mu\in\fh^*_{\cW}$ and let $\{\bM(\lambda_{\cW})^{\mu;j}\}$
be the corresponding Jantzen-type filtration, see Subsection~\ref{Jan}.
Combining the above determinant formulas and the fact that~(\ref{pila})
implies $\alpha=\beta$ if $(\lambda+\hat{\rho},\alpha_0)\not=0$
or $(\lambda+\hat{\rho},\alpha_0)=0,\ (\alpha,m)\in B_{\geq 0}$,
we obtain that
\begin{equation}\label{syumj}
\sum_{j=1}^{\infty} [\bM(\lambda_{\cW})^{\mu;j}:
\bL((s_{\beta}.\lambda)_{\cW})]=1,
\end{equation}
if $\beta\in\Pi(\lambda)_0,\ (\mu|_{\fh^f},\beta)\not=0$ or
$\beta\in\Pi(\lambda)_{j},\ a(\mu,\lambda,\beta)\not=0, j>0$.

Take $\mu\in\im\Upsilon_{\bL(\lambda_{\cW})}$ and
set $\xi:=\mu|_{\fh^f}$. Let us show that $\xi=0$.
By~\Prop{corally}, $\bM(\lambda_{\cW})^{\mu;1}=\bM(\lambda_{\cW})^{\mu;2}$ so
$\sum_{j=1}^{\infty} [\bM(\lambda_{\cW})^{\mu;j}:
\bL((s_{\beta}.\lambda)_{\cW})]\not=1$. Using~(\ref{syumj}) and~(\ref{phit}),
we conclude that $(\xi,\beta)=0$ for all $\beta\in\Pi(\lambda)_0$
and $a(\mu,\lambda,\beta)=0$ for each $\beta\in\Pi(\lambda)_{j}, j>0$.
By~(\ref{phit}),
$$a(\mu,\lambda,\alpha)=a(\mu,\lambda,\theta)+(\frac{(\lambda+\rho,\theta)}
{k+h^{\vee}}+1)(\xi,\alpha)$$
Hence
\begin{equation}\label{xibeta}
\begin{array}{l}
\forall\beta\in\Pi(\lambda)_0\ \ (\xi,\beta)=0\\
\forall\beta\in\Pi(\lambda)_j, j>0\ \
 a(\mu,\lambda,\theta)+(\frac{(\lambda+\rho,\theta)}
{k+h^{\vee}}+1)(\xi,\beta)=0.
\end{array}
\end{equation}

We claim that $(\lambda+\rho,\theta)+k+h^{\vee}\not=0$. Indeed,
$$(\lambda+\rho,\theta)+k+h^{\vee}
=(\lambda+\hat{\rho},\delta+\theta)=2(k+h^{\vee})-
(\lambda+\hat{\rho},\alpha_0).$$
If $(\lambda+\hat{\rho},\alpha_0)\not=0$, then, by above,
$\lambda+\hat{\rho}$ has a trivial stabilizer in $\hat{W}$ so
$(\lambda+\hat{\rho},\delta+\theta)\not=0$. On the other hand, if
$(\lambda+\hat{\rho},\alpha_0)=0$, then
$(\lambda+\rho,\theta)+k+h^{\vee}=2(k+h^{\vee})\not=0$.

Combining the inequality $(\lambda+\rho,\theta)+k+h^{\vee}\not=0$
and~(\ref{xibeta}), we conclude that for $a:=-
\frac{a(\mu,\lambda,\theta)(k+h^{\vee})}{(\lambda+\rho,\theta)+k+h^{\vee}}$
one has $(\xi,\beta)=a$ for all   $\beta\in\Pi(\lambda)_{\frac{1}{2}}$.
Let us show that $\xi-a\theta$ is orthogonal to
$\Pi(\lambda)\setminus\{\alpha_0\}$.
Indeed, since $(\xi,\beta)=a$ for all   $\beta\in\Pi(\lambda)_{\frac{1}{2}}$,
$\xi-a\theta$ is orthogonal to $\Pi(\lambda)_{\frac{1}{2}}$.
Moreover, by~(\ref{xibeta}), $\xi$ is orthogonal to $\Pi(\lambda)_0$ so
$\xi-a\theta$ is orthogonal to $\Pi(\lambda)_0$. Thus
$\xi-a\theta$ is orthogonal to $\Pi(\lambda)\setminus\{\alpha_0\}$, provided
that $\Pi(\lambda)_1$ is empty. Consider the case when
$\Pi(\lambda)_1\not=\emptyset$. The root
 $\beta\in\Pi(\lambda)_1$ is of the form $m\delta+\theta$; one has
$(\xi,\beta)=0$ since $\xi\in(\fh^f)^*$. Hence
$a(\mu,\lambda,\beta)=a(\mu,\lambda,\theta)=0$ so $a=0$, and thus
$\xi=\xi-a\theta$ is orthogonal to $\beta$. Hence
$\xi-a\theta$ is orthogonal to $\Pi(\lambda)\setminus\{\alpha_0\}$
as required.

Since $\alpha_0=\delta-\theta\in\Delta(\lambda)$ and $k+h^{\vee}=\frac{p}{q}$
is rational, $(q-1)\delta+\theta\in \Delta(\lambda)$ so
\begin{equation}\label{kukun}
(q-1)\delta+\theta=x_0(\delta-\theta)+\sum_{\beta\in
\Pi(\lambda)\setminus\{\alpha_0\}} x_{\beta}\beta
\end{equation}
for some $x_0,x_{\beta}\in\mathbb{Z}_{\geq 0}$. Thus
$(q-1-x_0)\delta+(1+x_0)\theta=\sum x_{\beta}\beta$
and so, by above, $\xi-a\theta$ is orthogonal to
$(q-1-x_0)\delta+(1+x_0)\theta$. Since $\xi-a\theta$ lies in
$\fh^*$, it is orthogonal to $\delta$. Therefore $\xi-a\theta$
is orthogonal to $\theta$ (because $1+x_0\not=0$) and thus
is orthogonal to $\alpha_0$. Hence $\xi-a\theta$
is orthogonal to $\Pi(\lambda)$. The rationality of $\lambda$
gives $\mathbb{C}\Pi(\lambda)=\mathbb{C}\hat{\Delta}$ so
$\xi-a\theta\in\mathbb{C}\delta$. Taking into account that
$\xi\in(\fh^f)^*$ is orthogonal to $\theta$ and $\Lambda_0$,
we conclude that $\mu|_{\fh^f}=\xi=0$.

From~(\ref{kukun}) it follows that
$\Pi(\lambda)\not=\Pi(\lambda)_0\cup\{\alpha_0\}$, so, by~(\ref{xibeta}),
$a(\mu,\lambda,\beta)=0$ for some $\beta$. Substituting $\xi=0$
to~(\ref{phit}), we obtain $\langle \mu, L_0\rangle=0$. By above,
$\mu|_{\fh^f}=0$, so $\mu=0$, that is
$\Ext^1(\bL(\lambda_{\cW}),\bL(\lambda_{\cW})=0$ as required.
\end{proof}

\subsubsection{}
\begin{cor}{corpolyh}
If $\lambda\in\fhh^*_k$ is KW-admissible weight,
then $\lambda_{\cW}$ is admissible.
\end{cor}

\subsubsection{}
Recall that the Virasoro vertex algebra is isomorphic to $\cW^k(\fsl_2,f)$,
where $f$ is a non-zero nilpotent element of $\fsl_2$, and that its central
charge $c(k)$ is given by formula~(\ref{ck}). Here we show how to recover some
results of Section~\ref{Vir} using the theory of W-algebras.
Retain notations of~\ref{exaaff}. One has
$\varphi(\lambda_{r,s})=(h^{p,q}_{r,s},c^{p,q})$.
Recall that $B=\{\Gamma_r\}_{r=1}^q,\
X_{\Gamma_r}=\{\lambda_{r,s}\}_{s=0}^p$.

Note that $\alpha_0^{\vee}\in\Gamma_r$
only for $r=q$. Thus, by above, $\varphi(\lambda_{r,s})$
is admissible if $r=1,\ldots,q-1, s=0,\ldots, p$.
For $r=q$ the KW-admissible weights are $\lambda_{q,s}$ with
$s=1,\ldots,p-1$. The face $(\lambda+\rho,\alpha_0)=0$
of the polyhedron $X_{\Gamma_q}$ is $\lambda_{q,p}$.
By~\Prop{extW0}, $\varphi(\lambda_{q,s})$
is admissible if $s=1,\ldots, p-1$ and if $(r,s)=(q,p)$.

Summarizing, we get that $\varphi(\lambda_{r,s})=(h^{p,q}_{r,s},c^{p,q})$
is admissible if $r=1,\ldots,q-1, s=0,\ldots, p$ or $r=q, s=1,\ldots, p$.
Taking into account that $h^{p,q}_{r,s}=h^{p,q}_{q-r,p-s}$,
we see that the admissible weights lie in the set
$\{(h^{p,q}_{r,s},c^{p,q})\}_{r=0,\ldots,q, s=0,\ldots, p}$ and that
$(h^{p,q}_{r,s},c^{p,q})$ is admissible
if $r=0,\ldots,q, s=0,\ldots, p$, and $(r,s)\not=(0,p), (q,0)$.
By~\Cor{corVir}, the latter set is exactly the set
of admissible weights. Thus, we recover this corollary, except for
the proof that points $(0,p), (q,0)$ are not admissible.
One can treat similarly
the Neveu-Schwarz algebra by taking
${\osp}(1,2)$, recovering thereby~\Cor{corNS}.

\subsection{Example}\label{exavir}
We give an example of $N\in \CO_k$ with the image $H^0(N)$
which does not lie in $\CO$-category for $\cW$; in fact the image $H^0(N)$
is a non-splitting extension of a Verma $\cW$-module by itself.

Take $\fg:=\fsl_2$  and  $k\not=-2$.
In this case $\cW$ is a Virasoro algebra $\Vir^c=\Vir/(C-c)$ for
$c=1-\frac{6(k+1)^2}{k+2}$.
Let $N\in \CO_k$ be a non-splitting extension of
$M(k\omega)$ by $M(k\omega-\alpha_0)$, where $\omega\in \fhh$ is given by
$$\langle \omega,D\rangle=(\omega,\alpha_0)=0, \ (\omega,\alpha)=1;$$
it is easy to see that such extension exists
if $M(k\omega-\alpha_0)$ is irreducible:
it is a submodule of $M(k\omega-\Lambda_0)\otimes L(\Lambda_0)$
generated by $v\otimes w$, where $v$ is a highest weight vector
of $M(k\omega-\Lambda_0)$ and $w\in  L(\Lambda_0)_{s_0\Lambda_0}$.

We show below that $H^0(N)$ is a non-splitting extension of
a Verma module $\bM((k\omega)_{\cW})$ over $\Vir^c$ by itself
(such extension is unique up to an isomorphism).
Since $L_0$ does not act diagonally on $H^0(N)$,
$H^0(N)$ does not lie in $\CO$-category for the Virasoro algebra.
Notice that $\bM((k\omega)_{\cW})$ is irreducible if $k$ is not rational.

\subsubsection{Definition of $H^i$ in the case $\fg=\fsl_2$}\label{Hi}
Denote by $f,h,e$ the standard basis of $\fsl_2$ .
Then $\hat{\fsl}_2$ is generated by $f,h=\alpha^{\vee},e, f_0:=et^{-1},\
e_0:=ft,\ \alpha^{\vee}_0=[e_0,f_0], K, D$.

Define a Clifford algebra generated by the odd elements
$\psi_n,\psi_n^*$ with $n\in\mathbb{Z}$ subject to the relations
$[\psi_n,\psi_m]=[\psi_n^*,\psi_m^*]=0,\ [\psi_n,\psi_m^*]=\delta_{m,n}$.
Take a vacuum module $F^{ch}$ over this algebra generated by $\vac$
such that $\psi_n\vac=0$ for $n\geq 0$, $\psi_n^*\vac=0$ for $n>0$.
We define a $\mathbb{Z}$-grading on $F^{ch}=\sum (F^{ch})^i$
by the assignment $\deg \vac=0,\ \deg \psi_n=-1,\ \deg \psi^*_n=1$.
For a $\fhg$-module $V$ we set $C(V):=V\otimes F^{ch}$ and
define a $\mathbb{Z}$-grading by the assignment
$C(V)^i:=V\otimes (F^{ch})^i$.
Finally, we define $d: C(V)\to C(V)$ by
$$d=\sum_{n\in\mathbb{Z}} et^{-n}\otimes \psi_n^*+\psi_1^*.$$
Then $d$ is odd, $d^2=0$ and $d(C(V)^i)\subset C(V)^{i+1}$. By definition,
$$H^i(V):=\Ker d\cap C(V)^i/\im d\cap C(V)^i.$$
Recall that for $V\in \CO_k$ one has $H^i(V)=0$ for $i\not=0$.

Define a semisimple action of $\fhh$ on $F^{ch}$ by assigning
to $\psi_n$ the weight $\alpha+n\delta$, to $\psi_n^*$ the weight
$-\alpha+n\delta$ and to $\vac$ the zero weight.
This induces an action of $\fhh$ on the tensor product
$C(V)=V\otimes F^{ch}$. If $\fhh$ acts locally finitely on $V$, then
$\fhh$ acts locally finitely on $C(V)$ and
$$d(C(V)_{\nu})\subset C(V)_{\nu}+C(V)_{\nu+\alpha_0}.$$

Let $\hat{\Omega}$ be the Casimir operator on $V\in\CO_k$; define
the action of $\hat{\Omega}$ on $C(V)$ by
$$\hat{\Omega}(v\otimes u'):=\hat{\Omega}v\otimes u, \ v\in V, u\in F^{ch}.$$
The action of $L_0$ on $H(V)$ is given by
$$L_0=\frac{\hat{\Omega}}{2(k+2)}-(\frac{\alpha^{\vee}}{2}+D).$$

\subsubsection{}
Set $\bN:=H^0(N)$. Since  $N$ is an extension of
$M(k\omega)$ by $M(k\omega-\alpha_0)$, $\bN$ is an extension of
$M((k\omega)_{\cW})$ by itself. Let us show that this extension
is non-splitting.

Let $E$ be the highest weight space of $\bN$, i.e.
$E:=\bN_{(k\omega)_{\cW}}$. Then $E$ is two-dimensional.
Set $X:=\sum_j C(N)_{k\omega-j\alpha_0}$ and note that
$d(X)\subset X$. It is easy to see that $E$ is the image of
$X$ in $\bN=H^0(N)=\Ker d/\im d$, i.e. $E=(\Ker d\cap X)/(\im d\cap X)$.

Denote by $v$ the highest weight generator of $M(k\omega)\subset N$,
and by $v'$ a vector in $N$ satisfying  $e_0v'=v$.
One readily sees  that $X$ is spanned by
$$x_n:=f_0^nv\otimes \vac,\ \ x'_n:=f_0^nv'\otimes \vac,\ \
y_n:=f_0^nv\otimes \psi_{-1}\vac,\ \ y'_n:=f_0^nv'\otimes \psi_{-1}\vac,\
n\geq 0.$$
One has
$$d(x_n)=d(x_n')=0,\ \ d(y_n)=-x_n-x_{n+1},\ d(y'_n)=-x'_n-x'_{n+1}.$$
Thus $E=(\Ker d\cap X)/(\im d\cap X)$ is spanned by the images
$\ol{x}_0,\ol{x}'_0$.

The Casimir operator $\hat{\Omega}$
acts on $M(k\omega)\subset N$ by $a\cdot\id$ for some $a\in\mathbb{C}$ and
one has $\hat{\Omega}v'=av'+f_0v$. Therefore
$$L_0\ol{x}_0=b\ol{x}_0,\ \ \ L_0\ol{x}'_0=b\ol{x}'_0+\frac{\ol{x}_1}
{2(k+2)}=b\ol{x}'_0-\frac{\ol{x}_0}
{2(k+2)}, \text{ where } b:=\frac{a}{2(k+2)}-\frac{1}{2}.$$
Hence $L_0$ does not act semisimply on the highest weight space of $\bN$
and thus $\bN$ is a non-splitting extension
of $M((k\omega)_{\cW})$ by itself.

\subsection{Example}
\label{exaW}
Below we give an example of $\bN\in \tilde{\CO}(\cW)$
which does not lie in $H^0(\CO'_k)$, see~\ref{COO} for notation.
This example is based on the following observation:
if $M(\lambda)$ is an irreducible Verma module
and  $(\lambda+\hat{\rho},\alpha_0)=0$,
then $\bM:=H^0(M(\lambda))$  is an irreducible Verma module over $\cW$
and $Ext^1(\bM,\bM)\iso \fh^*_{\cW}\not=\im\phi_{\lambda}$,
by~\Thm{thmmuW} (iii).

\subsubsection{}
Fix $\lambda\in\fhh^*_k$ such that
$(\lambda+\hat{\rho},\alpha_0)=0$ and
$(\lambda+\hat{\rho}+\mu/2,\mu)\not=0$
for all non-zero $\mu$ in the root lattice of $\fhg$
(such $\lambda$ exists for
$k\not\in\mathbb{Q}$); this condition ensures
that the Casimir operator have different eigenvalues on $M(\lambda)$
and on $M(\lambda+\mu)$. In particular,  $M(\lambda+s\delta)$ is irreducible.
Moreover, if an indecomposable module $N\in \CO_k'$ has
an irreducible subquotient isomorphic to $M(\lambda+s\delta)$,
then all irreducible subquotients of $N$ are isomorphic.

Set $\bM:=H^0(M(\lambda))$. Since $\bM$ is a Verma module
one has $\im\Upsilon_{\bM}=\fh^*_{\cW}$.
Notice that $\im\phi_{\lambda}\not=\fh^*_{\cW}$.
We claim that  if $\bN$ is an extension of $\bM$ by $\bM$,
then $\bN\in H^0(\CO'_k)$ iff
$\Upsilon_{\bM}(\bN)\subset \im\phi_{\lambda}$.
In particular, any extension which does not correspond
to an element in $\im\phi_{\lambda}$
does not lie in $H^0(\CO'_k)$.

Let  $\bN$ be a non-splitting extension of $\bM$ by $\bM$
and $\bN\in H^0(\CO'_k)$. Then $\bN=H^0(N)$ for some indecomposable module
$N\in\CO_k'$. Combining~\Lem{lamW} and the assumption
$(\lambda+\hat{\rho},\alpha_0)=0$, we conclude that
$H^0(L(\lambda'))\cong \bM$ iff $\lambda'=\lambda+s\delta$.
Since $N\in\CO_k'$, $N$ has a finite filtration with the factors belonging
to $\CO_k$. Since $H^0(N)=\bN$,
$N$ has a subquotient of the form $M(\lambda+s\delta)$.
Then, by above, all irreducible subquotients
of $N$ are isomorphic to $M(\lambda+s\delta)$ and
$[N: M(\lambda+s\delta)]=2$.
Hence $N$ is an extension of $M(\lambda+s\delta)$ by itself
and, by~\Thm{thmmuW} (i),   $\Upsilon_{\bM}(\bN)=
\phi_{\lambda+s\delta}\bigl(\Upsilon_{M(\lambda+s\delta)}(N)\bigr)$.
One readily sees that $\phi_{\lambda+s\delta}=\phi_{\lambda}$
so $\Upsilon_{\bM}(\bN)\in \im\phi_{\lambda}$. This establishes the claim.

\subsection{Example}\label{exL}
Let us give an example of an admissible $\cW$-module $\bL$, which
does not admit a $\mathbb{R}_{\geq 0}$-grading
$\bL=\oplus_{j\geq 0}\bL_j$, compatible
with the grading on $\cW$ (i.e., $\cW_n\bL_j\subset \bL_{j-n}$)
such that all $\bL_j$ are finite-dimensional.

\subsubsection{Description of $\bL$}
Let $\fg:=\fsl(4)$. Set $\Pi=\{\alpha_1,\alpha_2,\alpha_3\}$,
fix the standard bases $\{\alpha_i^{\vee}\}_{i=1}^3$
in $\fh$ and $\{e_{\beta}\}_{\beta\in \Delta}$ in $\fn_-\oplus\fn_+$.
Then $\theta=\alpha_1+\alpha_2+\alpha_3,\
x=(\alpha_1^{\vee}+\alpha_2^{\vee}+\alpha_3^{\vee})/2$ and
$\fg^{\natural}\cong \fgl(3)$ has the following triangular
decomposition $\mathbb{C}e_{-\alpha_2}\oplus(\mathbb{C}\alpha_2^{\vee}
+\mathbb{C}(\alpha_1^{\vee}-\alpha_3^{\vee}))\oplus
\mathbb{C}e_{\alpha_2}$.

Set
$$\Pi':=\{2\delta-\theta;3\delta+\alpha_1;4\delta+\alpha_2;4\delta+\alpha_3\}
\subset \hat\Delta_+$$
and take $\lambda\in\fhh^*$ such that
$(\lambda+\hat{\rho},\beta)=1$ for any $\beta\in \Pi'$.
It is easy to check that $k:=\langle \lambda,K\rangle=4/13-4$,
that $\lambda$ is a KW-admissible weight
with $\Pi(\lambda)=\Pi'$.
Set $M:=M(\lambda)$ and let $M'$ be the maximal proper submodule of $M$.
 Note that
\begin{equation}\label{betaxD}
\langle \beta,x+D\rangle>0\ \text{ for any }\beta\in \Pi'.
\end{equation}

By~\Cor{corpolyh}, $\bL:=H^0(L(\lambda))\cong H^0(M)/H^0(M')$
is an admissible $W^k(\fg,e_{-\theta})$-module.

\subsubsection{}\label{prL}
Let $v\in\bM$ (resp. $\ol{v}\in\bL$) be a highest weight vector.
Set
$$y:=J^{e_{-\alpha_2}}_0\in\cU(\cW).$$
Observe that $v$ generates a free
module over $\mathbb{C}[y]$. Let us show that
$\ol{v}\in\bL$ also generates a free module over $\mathbb{C}[y]$.
Indeed, for any  $\lambda'\in\fhh^*$ one has
$$[\bM': \bL(\lambda'_{\cW})]=[M':L(\lambda')]\ \text{ by (H4)};\ \
\langle\lambda'_{\cW}-\lambda_{\cW}, L_0\rangle=
\langle\lambda-\lambda',x+D\rangle, \ \text{ by (H2)}.$$
Note that $[M':L(\lambda')]>0$ implies that $\lambda-\lambda'$
is a non-negative integral linear combination of elements in $\Pi'$ so,
by~(\ref{betaxD}), $\langle\lambda-\lambda',x+D\rangle>0$.
Therefore
$$[\bM': \bL(\lambda'_{\cW})]>0\ \Longrightarrow\
\langle\lambda'_{\cW}-\lambda_{\cW}, L_0\rangle>0.$$
As a result, for any $\mu\in\Omega(\bM')$ one has
$\langle\mu-\lambda_{\cW}, L_0\rangle>0$.
Since $[L_0,y]=0$, this means that $\mathbb{C}[y]v$ intersects $\bM'$
trivially and thus $\ol{v}\in\bL$ generates a free
module over $\mathbb{C}[y]$.

\subsubsection{}
Suppose that $\bL$ has such a grading.
Let us show that $\dim \bL_i=\infty$ for some $i$.

Write $\ol{v}=\sum_{i=1}^m v_i$,
where $v_i$ are homogeneous vectors.
By above, $\ol{v}$ generates a free module over $\mathbb{C}[y]$.
If some $v_i$ generates a free module over $\mathbb{C}[y]$,
then its homogeneous component is infinite-dimensional
(since $y$ has zero degree). If this is not the case, then
for each $v_i$ there exists a non-zero polynomial $P_i(y)$ such that
$P_i(y)v_i=0$. But then $\prod P_i(y)v=0$, a contradiction.

\section{A conjecture on simple $W$-algebras}
Recall that to any finite-dimensional simple Lie algebra
$\fg$, a nilpotent element $f$ of $\fg$ and $k\in\mathbb{C}$ one associates
the $W$-algebra $\cW^k(\fg,f)$, which is a
$\frac{1}{2}\mathbb{Z}_{\geq 0}$-graded vertex algebra~\cite{KRW},
\cite{KWdet}. Provided that $k\not=-h^{\vee}$, the grading is the
eigenspace decomposition with respect $L_0$, the $0$th coefficient
of a Virasoro filed $L(z)$, and the $0$th eigenspace of $L_0$ is
$\mathbb{C}\vac$. It follows that the vertex algebra
$\cW^k(\fg,f)$ has a unique simple quotient, denoted by $\cW_k(\fg,f)$.

One can define highest weight modules and Verma modules over
$\cW^k(\fg,f)$~\cite{KWdet} and, in the same way as in Section~\ref{Wa}
for $\cW=\cW^k(\fg,e_{-\theta})$, one  can define weakly admissible
and admissible (irreducible  highest weight)  modules  over $\cW^k(\fg,f)$.

\subsection{Conjecture}\label{conj9}
Suppose that the vertex algebra $\cW_k(\fg,f)$ satisfies the $C_2$ condition.
Then any irreducible $\cW_k(\fg,f)$-module is obtained by pushing down from
an admissible $\cW^k(\fg,f)$-module.

\subsection{}\label{92}
Suppose that $\cW_k(\fg,e_{-\theta})$ satisfies the $C_2$ condition.
Let $\bL$ be an irreducible $\cW_k(\fg,e_{-\theta})$-module; then it is, of
course, a $\cW$-module. It follows from Lemma 2.4 of~\cite{M} that
the eigenvalues of $L_0$  in $\bL$ are bounded from the below
and each eigenspace is finite-dimensional.
Retain notation of Subsection~\ref{UW}.
The elements $\{J^a_0, a\in ((\fn_+\cap\fg_0)+\fh^f)\}$
span a Lie subalgebra $\fp\subset\cU(\cW)$ which is
isomorphic to $(\fn_+\cap\fg_0)+\fh^f$. These elements
preserve the eigenspaces of $L_0$ in $\bL$.
Since the Lie algebra $(\fn_+\cap\fg_0)+\fh^f$ is solvable,
each eigenspace contains a $\fp$-eigenvector. Let
$v$ be a $\fp$-eigenvector of the minimal $L_0$-eigenvalue.
Then  $J^b_iv=0$ for each
$b$ and $i>0$, since $v$ has the minimal $L_0$-eigenvalue,
$J^a_0v=0$ if $a\in (\fn_+\cap\fg_0)$, and $\fh_{\cW}=\spn\{J^a_0,
a\in\fh^f\}+\mathbb{C}L_0$ acts diagonally on $v$. Hence $v$
is a highest weight vector and so $\bL$ is a highest weight module.
So, if $\cW_k(\fg,e_{-\theta})$ satisfies the $C_2$ condition, then
any its irreducible module is a highest weight module.
If, in addition, this module is admissible (i.e. Conjecture~\ref{conj9}
holds), it follows from~\Thm{thmadmW} that $\cW_k(\fg,e_{-\theta})$
is a regular vertex algebra.

We conjecture that a theorem, similar to~\Thm{thmadmW}
can be established for any $W$-algebra $\cW^k(\fg,f)$. Then
Conjecture~\ref{conj9} would imply that any $W$-algebra $\cW_k(\fg,f)$
satisfying $C_2$ condition, is regular.

%%%%%%%%%%%%%%%%  biblio.tex


\begin{thebibliography}{MMM}
%\bibitem[ABD]{ABD} T.~Abe, G.~Buhl and C.~Dong, {\em Rationality,
%regularity, and $C_2$-cofiniteness}, Trans. Amer. math. Soc. {\bf 356}
%(2004), no. 8, 3391--3402.


\bibitem[AM]{AM} D.~Adamovi\'c, A.~Milas {Vertex operator algebras associated
to modular invariant representations for $A_1^{(1)}$},
Math. Res. Lett. {\bf 2} (1995), 563--575.


\bibitem[Ar]{Ar} T.~Arakawa, {\em Representation theory of superconformal
algebras and the Kac-Roan-Wakimoto conjecture}, Duke Math. J.,
{\bf 130} (2005), 435--478.



\bibitem[Ast]{Ast} A.~Astashkevich, {\em On the structure of Verma modules
over Virasoro and Neveu-Schwarz algebras}, Commun. in Math. physics,
{\bf 186} (1997), 531--562.

\bibitem[BB]{BB} A.~Beilinson, J.~Bernstein, {\em A proof
of Jantzen conjectures}, I.~M.~Gelfand Seminar, 1--50,
Adv. Soviet Math. 16, Part 1, AMS, Providence, RI, 1993.

\bibitem[BPZ]{BPZ} A.~A.~Belavin, A.~M.~Polyakov, A.~B.~Zamolodchikov,
{\em Infinite conformal symmetry in two-dimensional quantum field theory},
Nuclear Physics B {\bf 241} (1984), 333-380.

\bibitem[DGK]{DGK} V.~V.~Deodhar, O.~Gabber, V.~G.~Kac,
{\em Structures of some categories of representations of
infinite-dimensional Lie algebras}, Adv. Math., {\bf 45} (1982), 92--116.

\bibitem[DLM]{DLM1} Ch.~Dong. H.~Li, G.~Mason,
{\em Regularity of rational vertex operator algebras}, Adv. Math.,
{\bf 132} (1997), 148-166.

%\bibitem[DLM2]{DLM2} Ch.~Dong. H.~Li, G.~Mason,
%{\em Twisted represenations of vertex operator algebras}, Math. Ann. {\bf
%310}, (1998), 571-600.

%\bibitem[FF]{FF} B.~Feigin, E.~Frenkel, {\em Quantization
%of Drinfeld-Sokolov reduction}, Phys. Lett., {\bf B 246}
%(1990),  75--81.

%\bibitem[FFu]{FFu} B.~Feigin, D.~Fuchs, ??


\bibitem[FZ]{FZ} I.~B.~~Frenkel, Y.~Zhu, {\em Vertex operator algebras
associated to representations of affine and Virasoro algebras},
Duke Math. J. {\bf 66} (1992),  123--168.

%\bibitem[G]{Gq} M.~Gorelik,
%{\em Shapovalov determinants of $Q$-type Lie superalgebras},
%IMRP (2006), Art. Id. 96895, 1-71.


\bibitem[GK]{GK} M.~Gorelik, V.~G.~Kac,
{\em On simplicity of vacuum modules},
Adv. Math. {\bf 211}, (2007), 621--677.


\bibitem[IK1]{IK2} K.~Iohara, Y.~Koga, {\em Fusion algebras for N=1
superconformal field theories through coinvariants I:
${\osp}(1|2)\hat{ }$-symmetry}, Journal f\"ur die reine und angewandte
Mathematik {\bf 531}, (2001), 1--34.


\bibitem[IK2]{IK1} K.~Iohara, Y.~Koga, {\em Representation Theory of
Neveu-Schwarz and Ramond algebra I: Verma modules}, Adv. in Math.
{\bf 178 } (2003), 1--65.




\bibitem[Jan]{Jan} J.-C.~Jantzen, {\em Kontravariante Formen
   auf induzierten Darstellungen halbeinfacher Lie-Algebren}, Math. Ann.
{\bf 226} (1977), 53--65.

\bibitem[J]{Tony} A.~Joseph, {\em Sur l'annulateur d'un module de Verma},
in NATO Adv. Sci. Inst. Ser. C Math. Phys. Sci., {\bf 514}
Representation Theories and Algebraic Geometry, (Montreal, PQ, 1997),
237--300, Kluwer Acad. Publ., Dordrecht,  1998.


\bibitem[K1]{K77} V.~G.~Kac, {\em Lie superalgebras},
Adv. in Math., {\bf 26} (1977), 8--96.

%\bibitem[K2]{K78} V.~G.~Kac, {\em Infinite-dimensional algebras, Dedekind's
%$\eta$-function, classical M\"obius function and the Very Strange  Formula},
%Adv. in Math., {\bf 30} (1978), 85--136.

\bibitem[K2]{Kbook} V.~G.~Kac, {\em Infinite-dimensional Lie algebras},
Cambridge University Press, 1990.

\bibitem[K3]{Kbook2} V.~G.~Kac,
{\em Vertex algebras for beginners}, University Lecture Series,
Vol. 10, AMS, 1998.




\bibitem[KK]{KK} V.~G.~Kac, D.~Kazhdan, {\em Structure of representations
with highest weight of infinite-dimensional Lie algebras},
Adv. Math.,  {\bf 34} (1979), 97-108.

\bibitem[KR]{KR} V.~G.~Kac, A.~K.~Raina, {\em Bombay lectures
on highest weight representations of infinite dimensional
Lie algebras}, Singapore: World Sci., 1987.

\bibitem[KRW]{KRW} V.~G.~Kac, S.-S.~Roan and M.~Wakimoto,
{\em Quantum reduction for affine superalgebras}, Commun. Math. Phys.
{\bf 241} (2003), 307--342.



\bibitem[KW1]{KWmod} V.~G.~Kac, M.~Wakimoto, {\em Modular invariant
representations of infinite-dimensional Lie algebras and superalgebras},
Proc. Natl. Acad. Sci., USA, {\bf 85} (1998), 4956--4960.


\bibitem[KW2]{KWmod2} V.~G.~Kac, M.~Wakimoto, {\em Classification of
modular invariant representations of  affine Lie algebras},
Adv. Series in Math. Physics, {\bf 7}, World Scientific, 1989.




\bibitem[KW3]{KWdet} V.~G.~Kac, M.~Wakimoto, {\em Quantum reduction
and representation theory of superconformal algebras}, Adv. in Math.,
{\bf 185}, (2004), 400--458.


\bibitem[KW4]{KW4} V.~G.~Kac, M.~Wakimoto,
{\em On rationality of $\cW$-algebras}, Transform. Groups, {\bf 13}, (2008),
671--713.

\bibitem[KWa]{KWa} V.~G.~Kac, W.~Wang, {\em Vertex operator
superalgebras and their representations}, in Mathematical Aspects
of Conformal and Topological Field Theories and Quantum Groups,
Contemp. Math. {\bf 175}, AMS, Providence, Rhode Island, 1994, 161--192.

\bibitem[KT]{KT}  M.~Kashiwara, T.~Tanisaki,
{\em Characters of irreducible modules with non-critical highest weights
over affine Lie algebras}, Representations and Quantizations,
Shanghai (1998), 275--296, China High. Educ. Press, Beijing, 2000.


%\bibitem[KL]{KL} D.~Kazhdan, G.~Lusztig, {\em Representations of Coxeter
%groups and Hecke algebras}, Invent. Math., {\bf 53} (1979), 165--184.

%\bibitem[KT1]{KT} M.~Kashiwara, T.~Tanisaki {\em Kazhdan-Lusztig
%conjecture for symmetrizable Kac-Moody algebras III- positive rational case},
%Asian J. Math. {\bf 2}  (1998), 779--832.

%\bibitem[KT2]{KT18} M.~Kashiwara, T.~Tanisaki {\em Kazhdan-Lusztig
%conjecture for affine Lie algebras with
%negative level}, Duke Math. {\bf 77} (1995), 21--62.


\bibitem[L]{Lang} S.~Lang, {\em Algebra}, Addison-Wesley, 1965.

%\bibitem[Li]{Li} H.~Li, {\em Some finiteness properties of regular vertex
%operator algebras}, J. of Algebra, {\bf 212}, (1999), 495--514.

\bibitem[M]{M} M.~Miyamoto, {\em Modular invariance of vertex
operator algebras satisfying $C_2$-cofiniteness}, Duke Math. J.,
{\bf 122}, (2004), 51--91.


\bibitem[MP]{MP} R.~V.~Moody, A.~Pianzola, {\em Lie algebras with triangular
decompositions}, Canadian math. Soc. series of monographs and advanced
texts, A Wiley-Interscience Publication, John Wiley \& Sons, 1995.

\bibitem[Sh]{Sh} N.~Shapovalov, {\em On a bilinear form on the
universal enveloping algebra of a complex semisimple Lie algebra},
 Functional Anal. Appl. {\bf 6} (1972), p. 307--312 (in Russian).


\bibitem[Zh]{Zh} Y.~Zhu, {\em Modular invariance of characters of
vertex operator algebras},  J. Amer. Math. Soc. {\bf 9} (1996), 237--302.

\end{thebibliography}
\end{document}